%% file: main.tex
\documentclass[a4paper,USenglish,cleveref,autoref]{lipics-v2019}
\usepackage{ifdraft}

\title{A Linear-logical Reconstruction of Intuitionistic Modal Logic S4}

\author{Yosuke Fukuda}{Graduate School of Informatics, Kyoto University, Japan \and \url{http://www.fos.kuis.kyoto-u.ac.jp/~yfukuda}}{yfukuda@fos.kuis.kyoto-u.ac.jp}{}{}

\author{Akira Yoshimizu}{French Institute for Research in Computer Science and Automation\,(INRIA), France \and \url{http://www.cs.unibo.it/~akira.yoshimizu}}{akira.yoshimizu@inria.fr}{}{}

\authorrunning{Y. Fukuda and A. Yoshimizu}

\Copyright{Yosuke Fukuda and Akira Yoshimizu}



\category{}

\relatedversion{}

\supplement{}

\acknowledgements{The authors thank Kazushige Terui for his helpful comments on our work and pointing out related work on multicolored linear logic,
and Atsushi Igarashi for his helpful comments on earlier drafts.
Special thanks are also due to anonymous reviewers of FSCD 2019 for their fruitful comments.
This work was supported by the Research Institute for Mathematical Sciences, an International Joint Usage/Research Center located in Kyoto University.}

\nolinenumbers 

\hideLIPIcs  

\EventEditors{Herman Geuvers}
\EventNoEds{1}
\EventLongTitle{4th International Conference on Formal Structures for Computation and Deduction (FSCD 2019)}
\EventShortTitle{FSCD 2019}
\EventAcronym{FSCD}
\EventYear{2019}
\EventDate{June 24--30, 2019}
\EventLocation{Dortmund, Germany}
\EventLogo{}
\SeriesVolume{131}
\ArticleNo{20}

\usepackage{amssymb}
\usepackage{bussproofs}
\usepackage{cmll}
\usepackage{color}
\usepackage{changepage}
\usepackage{graphicx}
\usepackage{stmaryrd}
\usepackage{fancyhdr}
\usepackage[all]{xy}
\usepackage{environ}
\usepackage{scalerel}

\theoremstyle{plain}
\newtheorem{fact}[theorem]{Fact}

\EnableBpAbbreviations

\newcommand{\defeq}{\overset{\mathrm{def}}{=}}
\newcommand{\fv}[1]{\mathrm{FV}(#1)}
\newcommand{\fdom}[1]{\mathrm{dom}(#1)}
\newcommand{\fint}[1]{\llbracket #1 \rrbracket}
\newcommand{\fgirardtrans}[1]{ \lceil #1 \rceil }
\newcommand{\ftrans}[1]{ \mathcal{T}\fgirardtrans{#1} }
\newcommand{\fembed}[1]{ (#1)^\ddagger }
\newcommand{\fnettrans}[1]{ (#1)^\bullet }
\newcommand{\flambdanettrans}[1]{ (#1)^\dagger }
\newcommand{\fbracket}[1]{\langle #1 \rangle}
\newcommand{\ftypeof}[1]{ \mathrm{typeof}(#1) }
\newcommand{\mset}{\mathbb{M}}

\NewEnviron{rulefigure}[2]{
\begin{figure}[t]
  \begin{tabular}{|c|}
    \hline
    {
      \begin{minipage}{0.96\hsize}
        \vspace{0.2em}
        \BODY
      \end{minipage}
    }\\
    \hline
  \end{tabular}
  \vspace{-0.5em}
  \caption{#2}
  \label{#1}
\end{figure}
}

\newcommand{\bangbox}{\mathrel{\Box\hspace{-.48em}\raise0.20ex\hbox{\scalebox{0.7}{$!$}}}}
\newcommand{\smallbangbox}{\mathrel{\Box\hspace{-.39em}\raise0.20ex\hbox{\scalebox{0.5}{$!$}}}\hspace{.035em}}
\newcommand{\whynotdia}{\mathrel{\Diamond\hspace{-.475em}\raise0.20ex\hbox{\scalebox{0.6}{$?$}}}\hspace{-0.1em}}
\newcommand{\sarrow}{\supset}
\newcommand{\stensor}{\otimes}
\newcommand{\spar}{\parr}
\newcommand{\slimp}{\multimap}
\newcommand{\bangzero}{\underset{0}{!}}
\newcommand{\whynotzero}{\underset{0}{?}}
\newcommand{\bangone}{\underset{1}{!}}
\newcommand{\whynotone}{\underset{1}{?}}

\newcommand{\sbangs}{\delta}
\newcommand{\mcont}{\mathcal{M}}
\newcommand{\bcont}{\mathcal{B}}
\newcommand{\ncont}{\mathcal{N}}

\newcommand{\il}{\mathrm{IL}}
\newcommand{\isf}{\mathrm{IS4}}
\newcommand{\csf}{\mathrm{CS4}}
\newcommand{\imell}{\mathrm{IMELL}}
\newcommand{\cmell}{\mathrm{CMELL}}
\newcommand{\imellbox}{\mathrm{IMELL}^{\Box}}

\newcommand{\njbangbox}{\mathrm{NJ}^{\smallbangbox}\,}
\newcommand{\hjbangbox}{\mathrm{HJ}^{\smallbangbox}\,}
\newcommand{\ljbox}{\mathrm{LJ}^{\Box}}
\newcommand{\mell}{\mathrm{MELL}}
\newcommand{\imellbangbox}{\imell^{\smallbangbox}\,}
\newcommand{\cmellbangbox}{\mathrm{CMELL}^{\smallbangbox}\,}
\newcommand{\lambdabox}{\lambda^{\Box}}
\newcommand{\lambdabangbox}{\lambda^{\smallbangbox}\,}
\newcommand{\lambdabanglimp}{\lambda^{!, \slimp}}
\newcommand{\bill}{\textbf{2-LL}}

\newcommand{\clvbangbox}{\mathrm{CL}^{\smallbangbox}\,}

\newcommand{\rax}{\mathrm{Ax}}
\newcommand{\rmax}{\Box\mathrm{Ax}}
\newcommand{\rimpi}{\sarrow\!\mathrm{I}}
\newcommand{\rimpe}{\sarrow\!\mathrm{E}}
\newcommand{\rboxi}{\Box\mathrm{I}}
\newcommand{\rboxe}{\Box\mathrm{E}}
\newcommand{\rbax}{!\mathrm{Ax}}
\newcommand{\rbbax}{\bangbox\mathrm{Ax}}
\newcommand{\rlax}{\mathrm{LinAx}}
\newcommand{\rlimpi}{\slimp\!\mathrm{I}}
\newcommand{\rlimpe}{\slimp\!\mathrm{E}}
\newcommand{\rbangi}{!\mathrm{I}}
\newcommand{\rbange}{!\mathrm{E}}
\newcommand{\rbangboxi}{\bangbox\mathrm{I}}
\newcommand{\rbangboxe}{\bangbox\mathrm{E}}
\newcommand{\hcomb}{\mathrm{Ax}}
\newcommand{\hmp}{\mathrm{MP}}
\newcommand{\hbnec}{!}
\newcommand{\hbbnec}{\bangbox}
\newcommand{\rcut}{\mathrm{Cut}}
\newcommand{\rbcut}{!\mathrm{Cut}}
\newcommand{\rbbcut}{\bangbox\mathrm{Cut}}
\newcommand{\rtensor}{\stensor}
\newcommand{\rpar}{\spar}
\newcommand{\rbang}{!}
\newcommand{\rbangbox}{\bangbox}
\newcommand{\rwhynot}{?}
\newcommand{\rwhynotdia}{\whynotdia}
\newcommand{\rweak}{\mathrm{W}}
\newcommand{\rbweak}{!\mathrm{W}}
\newcommand{\rbbweak}{\bangbox\mathrm{W}}
\newcommand{\rqweak}{?\mathrm{W}}
\newcommand{\rmweak}{\whynotdia\mathrm{W}}
\newcommand{\rcont}{\mathrm{C}}
\newcommand{\rbcont}{!\mathrm{C}}
\newcommand{\rbbcont}{\bangbox\mathrm{C}}
\newcommand{\rqcont}{?\mathrm{C}}
\newcommand{\rmcont}{\whynotdia\mathrm{C}}

\newcommand{\rprom}{!\mathrm{R}}
\newcommand{\rdere}{!\mathrm{L}}
\newcommand{\rlimpl}{\slimp\!\mathrm{L}}
\newcommand{\rlimpr}{\slimp\!\mathrm{R}}
\newcommand{\rimpl}{\sarrow\!\mathrm{L}}
\newcommand{\rimpr}{\sarrow\!\mathrm{R}}
\newcommand{\rboxl}{\Box\mathrm{L}}
\newcommand{\rboxr}{\Box\mathrm{R}}
\newcommand{\rbangl}{!\mathrm{L}}
\newcommand{\rbangr}{!\mathrm{R}}
\newcommand{\rbangboxl}{\bangbox\mathrm{L}}
\newcommand{\rbangboxr}{\bangbox\mathrm{R}}
\newcommand{\rpseudonec}{\Box'}

\newcommand{\betarule}[1]{(\beta #1)}
\newcommand{\rbetaimp}{\betarule{\sarrow}}
\newcommand{\rbetalimp}{\betarule{\slimp}}
\newcommand{\rbetabox}{\betarule{\Box}}
\newcommand{\rbetabang}{\betarule{!}}
\newcommand{\rbetabangbox}{\betarule{\bangbox~}}

\input{ott_definition.tex}

\begin{document}

\maketitle

\begin{abstract}
We propose a \emph{modal linear logic} to reformulate intuitionistic modal logic S4\,($\isf$)
in terms of linear logic, establishing an S4-version of Girard translation from $\isf$ to it.
While the Girard translation from intuitionistic logic to linear logic is well-known,
its extension to modal logic is non-trivial since a naive combination of the S4 modality and the exponential modality
causes an undesirable interaction between the two modalities.
To solve the problem, we introduce an extension of intuitionistic multiplicative exponential linear logic with a modality combining the S4 modality and the exponential modality,
and show that it admits a sound translation from $\isf$.
Through the Curry--Howard correspondence we further obtain a Geometry of Interaction Machine semantics of
the modal $\lambda$-calculus by Pfenning and Davies for staged computation.
\end{abstract}

\section{Introduction}
Linear logic discovered by Girard\,\cite{G:linear_logic} is, as he wrote, not an alternative logic but should be regarded as 
an ``extension'' of usual logics.
Whereas usual logics such as classical logic and intuitionistic logic admit the structural rules of weakening and contraction,
linear logic does not allow to use the rules freely, but it reintroduces them in a controlled manner by using the exponential modality `$!$' (and its dual `$?$'). 
Usual logics are then reconstructed 
in terms of linear logic with the power of the exponential modalities, via the Girard translation.
 
In this paper, we aim to extend the framework of linear-logical reconstruction to the $(\Box, \sarrow)$-fragment of intuitionistic modal logic S4\,($\isf$)
by establishing what we call ``modal linear logic'' and an S4-version of Girard translation from $\isf$ into it.
However, the crux to give a faithful translation is that 
a naive combination of the $\Box$-modality and the $!$-modality causes an undesirable interaction between the inference rules of the two modalities.
To solve the problem, 
we define the modal linear logic as an extension of intuitionistic multiplicative exponential linear logic with a modality `$\bangbox$~'\,(pronounced by ``bangbox'') that integrates `$\Box$' and `$!$',
and show that it admits a faithful translation from $\isf$.

As an application, we consider a computational interpretation of the modal linear logic.
A typed $\lambda$-calculus that we will define corresponds to a natural deduction for the modal linear logic through the Curry--Howard correspondence, 
and it can be seen as a reconstruction of the modal $\lambda$-calculus by
Pfenning and Davies\,\cite{PD:judgmental_reconstruction, DP:modal_analysis} for the so-called staged computation.
Thanks to our linear-logical reconstruction, we can further obtain a Geometry of Interaction Machine\,(GoIM) for the modal $\lambda$-calculus.

The remainder of this paper is organized as follows.
In Section~\ref{sec:background} we review some formalizations of linear logic and $\isf$.
In Section~\ref{sec:linear_logical_reconstruction} we explain a linear-logical reconstruction of $\isf$.
First, we discuss how a naive combination of linear logic and modal logic fails to obtain a faithful translation.
Then, we propose a modal linear logic with the $\bangbox$~-modality that admits a faithful translation from $\isf$.
In Section~\ref{sec:curry-howard} we give a computational interpretation of modal linear logic through a typed $\lambda$-calculus.
In Section~\ref{sec:axiomatization} we provide an axiomatization of modal linear logic by a Hilbert-style deductive system.
In Section~\ref{sec:goim} we obtain a GoIM of our typed $\lambda$-calculus as an application of our linear-logical reconstruction.
In Sections~\ref{sec:related_work} and \ref{sec:conclusion} we discuss related work and conclude our work, respectively.

\section{Preliminaries}
\label{sec:background}
We recall several systems of linear logic and modal logic.
In this paper, we consider the minimal setting to give an S4-version of Girard translation and its computational interpretation.
Thus, every system we will use only contain an implication and a modality as operators.

\subsection{Intuitionistic MELL and its Girard translation}

\begin{rulefigure}{fig:imell}{Definition of $\imell$.}
  \input{definitions/imell.tex}
\end{rulefigure}

\begin{rulefigure}{fig:girard_translation}{Definition of the Girard translation from intuitionistic logic.}
  \input{definitions/girard_translation.tex}
\end{rulefigure}

\noindent
Figure~\ref{fig:imell} shows the standard definition of the ($!$, $\slimp$)-fragment of \emph{intuitionistic multiplicative exponential linear logic}, which we refer to as $\imell$.
A formula is either a propositional variable, a linear implication, or an exponential modality.
We let $p$ range over the set of propositional variables, and $A$, $B$, $C$ range over formulae.
A \emph{context} $\Gamma$ is defined to be a multiset of formulae, and hence the exchange rule is assumed as a meta-level operation.
A \emph{judgment} consists of a context and a formula, written as $\Gamma \, \vdash \, \ottnt{A}$.
As a convention, we often write $\Gamma \, \vdash \, \ottnt{A}$ to mean that the judgment is derivable (and we assume similar conventions throughout this paper).
The notation $\ottsym{!}  \Gamma$ in the rule $\rprom$ denotes the multiset $\{ !A ~|~ A \in \Gamma \}$.

Figure~\ref{fig:girard_translation} defines the Girard translation\footnote{This is known to be the \emph{call-by-name} Girard translation (cf.\,\cite{M+:cbn_cbv}) and we only follow this version in later discussions. However, we conjectured that our work can apply to other versions.} from the $\sarrow$-fragment of
intuitionistic propositional logic $\il$.
For an $\il$-formula $A$, $ \fgirardtrans{ \ottnt{A} } $ will be an $\imell$-formula;
and $ \fgirardtrans{ \Gamma } $ a multiset of $\imell$-formulae.
Then, we can show that the Girard translation from $\il$ to $\imell$ is sound.

\begin{theorem}[Soudness of the translation]
  \label{thm:girard_translation}
  If $\Gamma \, \vdash \, \ottnt{A}$ in $\il$, then $\ottsym{!}   \fgirardtrans{ \Gamma }  \, \vdash \,  \fgirardtrans{ \ottnt{A} } $ in $\imell$.
\end{theorem}

\subsection{Intuitionistic S4}

We review a formalization of the $(\Box, \sarrow)$-fragment of intuitionistic propositional modal logic S4\,($\isf$).
In what follows, we use a sequent calculus for the logic, called $\ljbox$.
The calculus $\ljbox$ used here is defined in a standard manner
in the literature\,(e.g. it can be seen as the $\isf$-fragment of \textbf{G1s} for classical modal logic S4 by Troelstra and Schwichtenberg\,\cite{TS:basic_proof_theory}).

Figure~\ref{fig:ljbox} shows the definition of $\ljbox$. 
A formula is either a propositional variable, an intuitionistic implication, or a box modality.
A \emph{context} and a \emph{judgment} are defined similarly in $\imell$. 
The notation $ \Box \Gamma $ in the rule $\rboxr$ denotes the multiset $\{  \Box \ottnt{A}  ~|~ A \in \Gamma \}$.

\begin{rulefigure}{fig:ljbox}{Definition of $\ljbox$.}
  \input{definitions/ljbox.tex}
\end{rulefigure}

\begin{remark}
  It is worth noting that 
  the $!$-exponential in $\imell$ and the $\Box$-modality in $\ljbox$ have similar structures.
  To see this, let us imagine the rules $\rboxr$ and $\rboxl$ replacing the symbol `$\Box$' with `$!$'. The results will be exactly the same as $\rprom$ and $\rdere$.
  In fact, the $!$-exponential satisfies the S4 axiomata in $\imell$, which is the reason we also call it as a modality.
\end{remark}

\subsection{Typed $\lambda$-calculus of the intuitionistic S4}
We review the modal $\lambda$-calculus developed by Pfenning and Davies\,\cite{PD:judgmental_reconstruction, DP:modal_analysis}, which we call $\lambdabox$.
The system $\lambdabox$ is essentially the same calculus as $\lambda^{\rightarrow\Box}_{e}$ in \cite{DP:modal_analysis}, although some syntax are changed to fit our notation in this paper.
$\lambdabox$ is known to correspond to a natural deduction system for $\isf$, as is shown in \cite{PD:judgmental_reconstruction}.

\begin{rulefigure}{fig:lambdabox}{Definition of $\lambdabox$.}
  \input{definitions/lambdabox.tex}
\end{rulefigure}

Figure~\ref{fig:lambdabox} shows the definition of $\lambdabox$.
The set of types corresponds to that of formulae of $\isf$.
We let $x$ range over the set of term variables, and $M, N, L$ range over the set of terms.
The first three terms are as in the simply-typed $\lambda$-calculus.
The terms $ \Box  \ottnt{M} $ and $\ottkw{let} \,  \Box  \ottmv{x}   \ottsym{=}  \ottnt{M} \, \ottkw{in} \, \ottnt{N}$ is used to represent a constructor and a destructor for types $ \Box \ottnt{A} $, respectively.
The variable $x$ in $\lambda  \ottmv{x}  \ottsym{:}  \ottnt{A}  \ottsym{.}  \ottnt{M}$ and $\ottkw{let} \,  \Box  \ottmv{x}   \ottsym{=}  \ottnt{M} \, \ottkw{in} \, \ottnt{N}$ is supposed to be \emph{bound} in the usual sense
and the \emph{scope} of the biding is $M$ and $N$, respectively.
The set of \emph{free} (i.e., unbound) variables in $M$ is denoted by $\fv{M}$.
We write the \emph{capture-avoiding substitution} $ \ottnt{M}  [  \ottmv{x}  :=  \ottnt{N}  ] $ to denote the result of replacing $N$ for every free occurrence of $x$ in $M$.

A \emph{(type) context} is defined to be the set of pairs of a term variable $x_i$ and a type $A_i$
such that all the variables are distinct, which is written as $x_1 : A_1, \cdots, x_n : A_n$ and is denoted by $\Gamma$, $\Delta$, $\Sigma$, etc.
Then, a \emph{(type) judgment} is defined, in the so-called \emph{dual-context} style, to consists of two contexts, a term, and a type, written as $\Delta  \ottsym{;}  \Gamma \, \vdash \, \ottnt{M}  \ottsym{:}  \ottnt{A}$.

The intuition behind the judgment $\Delta  \ottsym{;}  \Gamma \, \vdash \, \ottnt{M}  \ottsym{:}  \ottnt{A}$ is that the context $\Delta$ is intended to implicitly represent assumptions for types of form $ \Box \ottnt{A} $,
while the context $\Gamma$ is used to represent ordinary assumptions as in the simply-typed $\lambda$-calculus.

The typing rules are summarized as follows.
$\rax$, $\rimpi$, and $\rimpe$ are all standard, although they are defined in the dual-context style.
$\rmax$ is another variable rule, which can be seen as what to formalize
the modal axiom \textit{T} (i.e., $\vdash \,  \Box \ottnt{A}  \supset  \ottnt{A} $) from the logical viewpoint.
$\rboxi$ is a rule for the constructor of $ \Box \ottnt{A} $, which corresponds to the necessitation rule for the $\Box$-modality.
Similarly, $\rboxe$ is for the destructor of $ \Box \ottnt{A} $, which corresponds to the elimination rule.

The \emph{reduction} $ \leadsto $ is defined to be the least compatible relation on terms generated by $\rbetaimp$ and $\rbetabox$.
The \emph{multistep reduction} $ \leadsto^+ $ is defined to be the transitive closure of $ \leadsto $.

\section{Linear-logical reconstruction}
\label{sec:linear_logical_reconstruction}

\subsection{Naive attempt at the linear-logical reconstruction}

It is natural for a ``linear-logical reconstruction'' of $\isf$
to define a system that has both properties of linear logic and modal logic, so as to be a target system for an S4-version of Girard translation.
However, a naive combination of linear logic and modal logic is not suitable to establish a faithful translation.

Let us consider what happens if we adopt a naive system.
The simplest way to define a target system for the S4-version of Girard translation is to make
an extension of $\imell$ with the $\Box$-modality.
Suppose that a deductive system $\imellbox$ is such a calculus, that is, 
the formulae of $\imellbox$ are defined by the following grammar:
  \begin{align*}
    A, B &::= p ~|~ \ottnt{A}  \multimap  \ottnt{B} ~|~ \ottsym{!}  \ottnt{A} ~|~  \Box \ottnt{A} 
  \end{align*}
with the inference rules being those of $\imell$, along with the rules $\rboxr$ and $\rboxl$ of $\ljbox$.

As in the case of Girard translation from $\il$ to $\imell$, 
we have to establish the following theorem for some translation $\fgirardtrans{-}$:

\begin{center}
  If $\Gamma \, \vdash \, \ottnt{A}$ is derivable in $\ljbox$, then so is $\ottsym{!}   \fgirardtrans{ \Gamma }  \, \vdash \,  \fgirardtrans{ \ottnt{A} } $ in $\imellbox$.
\end{center}

\noindent
but, if we extend our previous translation $\fgirardtrans{-}$ from $\il$ to $\imell$
with $ \fgirardtrans{  \Box \ottnt{A}  }  \defeq  \Box  \fgirardtrans{ \ottnt{A} }  $,
we get stuck in the case of $\rboxr$. This is because we need to establish the inference $\rpseudonec$ in Figure~\ref{fig:bad_translation},
which means that we have to be able to obtain a derivation of form $\ottsym{!}   \fgirardtrans{  \Box \Gamma  }  \, \vdash \,  \Box  \fgirardtrans{ \ottnt{A} }  $ from that of $\ottsym{!}   \fgirardtrans{  \Box \Gamma  }  \, \vdash \,  \fgirardtrans{ \ottnt{A} } $ in $\imellbox$.

\begin{figure}[htbp]
    \begin{tabular}{|c|c|}
      \hline
      {
      \begin{minipage}{0.46\hsize}
        {
          \vspace{-1em}
          \begin{tabular}{c}
            \hspace{-1.5em}
            \begin{minipage}{0.355\hsize}
              \begin{prooftree}
                      \AXC{$ \vdots $}
                    \noLine
                    \UIC{$ \Box \Gamma  \, \vdash \, \ottnt{A}$}
                  \RightLabel{$ \rboxr $}
                  \UIC{$ \Box \Gamma  \, \vdash \,  \Box \ottnt{A} $}
                \noLine
                \UIC{ in $\ljbox$ }
              \end{prooftree}
            \end{minipage}
          
            \begin{minipage}{0.195\hsize}
              \[
                \xymatrix{
                   \ar@{|->}[r]^-{ \fgirardtrans{-} } & 
                }
              \]
            \end{minipage}
          
            \begin{minipage}{0.30\hsize}
              \begin{prooftree}
                      \AXC{$ \vdots $}
                    \noLine
                    \UIC{$\ottsym{!}   \fgirardtrans{  \Box \Gamma  }  \, \vdash \,  \fgirardtrans{ \ottnt{A} } $}
                  \doubleLine
                  \RightLabel{$ \rpseudonec $}
                  \UIC{$\ottsym{!}   \fgirardtrans{  \Box \Gamma  }  \, \vdash \,  \fgirardtrans{  \Box \ottnt{A}  } $}
                \noLine
                \UIC{ in $\imellbox$ }
              \end{prooftree}
            \end{minipage}
          \end{tabular}
        }
        \caption{Translation for the case of $\rboxr$.}
        \label{fig:bad_translation}
      \end{minipage}
      }
      &
      {
      \begin{minipage}{0.45\hsize}
        \begin{tabular}{c}
          \begin{minipage}{0.90\hsize}
            \begin{prooftree}
                \AXC{$ \Box \ottsym{(}  \ottmv{p}  \supset  \ottmv{q}  \ottsym{)}  \ottsym{,}   \Box \ottmv{p}   \, \vdash \, \ottmv{q}$}
              \RightLabel{$ \rboxr $}
              \UIC{$ \Box \ottsym{(}  \ottmv{p}  \supset  \ottmv{q}  \ottsym{)}  \ottsym{,}   \Box \ottmv{p}   \, \vdash \,  \Box \ottmv{q} $}
            \end{prooftree}
          \end{minipage}
          \vspace{-0.7em}
        \end{tabular}
        \caption{ Valid inference in $\ljbox$. }
        \label{fig:valid_inference}

        \vspace{-0.5em}
        \rule{\textwidth}{0.4pt}

        \begin{tabular}{c}
          \begin{minipage}{0.90\hsize}
            \begin{prooftree}
                \AXC{$\ottsym{!}   \Box \ottsym{(}  \ottsym{!}  \ottmv{p}  \multimap  \ottmv{q}  \ottsym{)}  \ottsym{,}  \ottsym{!}   \Box \ottmv{p}   \, \vdash \, \ottmv{q}$} 
              \RightLabel{$ \rboxr $}
              \UIC{$\ottsym{!}   \Box \ottsym{(}  \ottsym{!}  \ottmv{p}  \multimap  \ottmv{q}  \ottsym{)}  \ottsym{,}  \ottsym{!}   \Box \ottmv{p}   \, \vdash \,  \Box \ottmv{q} $}
            \end{prooftree}
          \end{minipage}
          \vspace{-0.7em}
        \end{tabular}
        \caption{ Invalid inference in $\imellbox$. }
        \label{fig:invalid_inference}
      \end{minipage}
      }\\
      \hline
    \end{tabular}
\end{figure}

However, the inference $\rpseudonec$ is invalid in $\imellbox$ in general,
because there exists a counterexample. First,
the inference shown in Figure~\ref{fig:valid_inference} is valid, and the judgment $ \Box \ottsym{(}  \ottmv{p}  \supset  \ottmv{q}  \ottsym{)}  \ottsym{,}   \Box \ottmv{p}   \, \vdash \,  \Box \ottmv{q} $ is indeed derivable in $\ljbox$.
However, the corresponding inference via $\fgirardtrans{-}$ is invalid as Figure~\ref{fig:invalid_inference} shows. In the figure, the judgments correspond to those in Figure~\ref{fig:valid_inference} via $\fgirardtrans{ - }$, 
but the inference $\rboxr$ in Figure~\ref{fig:invalid_inference} is invalid in $\imellbox$ due to the side-condition of $\rboxr$.
Even worse, we can see that the judgment $\ottsym{!}   \Box \ottsym{(}  \ottsym{!}  \ottmv{p}  \multimap  \ottmv{q}  \ottsym{)}  \ottsym{,}  \ottsym{!}   \Box \ottmv{p}   \, \vdash \,  \Box \ottmv{q} $ is itself underivable in
$\imellbox$\footnote{Precisely speaking, this can be shown as a consequence of the cut-elimination theorem of $\imellbox$,
and the theorem was shown in the authors' previous work\,\cite{FY:higher-arity}.}.

Moreover, one may think the other cases that we extend the original translation
$\fgirardtrans{-}$ from $\il$ to $\imell$ with
$ \fgirardtrans{  \Box \ottnt{A}  }  \defeq ~\ottsym{!}   \Box  \fgirardtrans{ \ottnt{A} }  $ or $ \fgirardtrans{  \Box \ottnt{A}  }  \defeq  \Box \ottsym{!}   \fgirardtrans{ \ottnt{A} }  $ will work to obtain a faithful translation.
However, the judgment $ \Box \ottmv{p}  \, \vdash \,  \Box  \Box \ottmv{p}  $ will be a counter-example in either case.
 
All in all, the problem of the naive combination formulated as $\imellbox$
intuitively came from an undesirable interaction between the right rules of the two modalities:

\begin{center}
\begin{tabular}{c}
  \begin{minipage}{0.30\hsize}
    \begin{prooftree}
        \AXC{$\ottsym{!}  \Gamma \, \vdash \, \ottnt{A}$}
      \RightLabel{$ \rprom $}
      \UIC{$\ottsym{!}  \Gamma \, \vdash \, \ottsym{!}  \ottnt{A}$}
    \end{prooftree}
  \end{minipage}

  \begin{minipage}{0.30\hsize}
    \begin{prooftree}
        \AXC{$ \Box \Gamma  \, \vdash \, \ottnt{A}$}
      \RightLabel{$ \rboxr $}
      \UIC{$ \Box \Gamma  \, \vdash \,  \Box \ottnt{A} $}
    \end{prooftree}
  \end{minipage}
\end{tabular}
\end{center}

\noindent
Each of these rules has a side-condition: the conclusion $!A$ in $\rprom$ must be derived from the modalized context $\ottsym{!}  \Gamma$, and similarly for $ \Box \ottnt{A} $ in $\rboxr$. This makes it hard to obtain a faithful S4-version of Girard translation for this naive extension.

\subsection{Modal linear logic}

We propose a \emph{modal linear logic} to give a faithful S4-version of Girard translation from $\isf$.

First of all, the problem we have identified essentially came from the fact that there is no relationship between `$!$' and `$\Box$',
and hence the side-conditions of $\rprom$ and $\rboxr$ do not hold when we intuitively expect them to hold.
Thus, we introduce a modality, `$\bangbox$~' combining `$!$' and `$\Box$', to solve this problem.

\begin{rulefigure}{fig:imellbangbox}{Definition of $\imellbangbox$.}
  \input{definitions/imellbangbox.tex}
\end{rulefigure}

Our modal linear logic, which is called $\imellbangbox$, is defined by a sequent calculus which is given in Figure~\ref{fig:imellbangbox}.
As we mentioned, the formulae are defined as an extension of those of $\imell$ with the $\bangbox$~-modality.
A point is that the $!$-modality is still there with the $\bangbox~$-modality.

The $\bangbox~$-modality is defined so as to have properties of both `$!$' and `$\Box$',
but `$!$' still behaves similarly to $\imell$.
Therefore, all the intuitions of the inference rules except $\rbangr$ and $\rbangboxr$ should be clear.
The rules $\rbangr$ and $\rbangboxr$ reflect the ``strength'' between the modalities `$!$' and `$\bangbox$~'.
Indeed, `$!$' and `$\bangbox~$' satisfy the S4 axiomata and `$\bangbox~$' is stronger than `$!$'.

\begin{example}
  \label{ex:axiomata}
  The following hold:
  \begin{enumerate}
    \item $\vdash \, \ottsym{!}  \ottnt{A}  \multimap  \ottnt{A}$ and $\vdash \,  \bangbox \ottnt{A}  \multimap  \ottnt{A} $
    \item $\vdash \, \ottsym{!}   (  \ottnt{A}  \multimap  \ottnt{B}  )   \multimap  \ottsym{!}  \ottnt{A}  \multimap  \ottsym{!}  \ottnt{B}$ and $\vdash \,  \bangbox  (  \ottnt{A}  \multimap  \ottnt{B}  )   \multimap   \bangbox \ottnt{A}   \multimap   \bangbox \ottnt{B}  $
    \item $\vdash \, \ottsym{!}  \ottnt{A}  \multimap  \ottsym{!}  \ottsym{!}  \ottnt{A}$ and $\vdash \,  \bangbox \ottnt{A}  \multimap   \bangbox\,\,\bangbox \ottnt{A}  $
    \item $\vdash \,  \bangbox \ottnt{A}  \multimap  \ottsym{!}  \ottnt{A} $ but $\nvdash \, \ottsym{!}  \ottnt{A}  \multimap   \bangbox \ottnt{A} $
  \end{enumerate}
\end{example}

\begin{remark}
  In Example~\ref{ex:axiomata}, the first three represent the so-called S4 axiomata: \textit{T}, \textit{K}, and \textit{4}.
  The last one represents the strength of the two modalities.
  Actually, assuming the $!$-modality and the $\bangbox~$-modality to satisfy the S4 axiomata and the ``strength'' axiom $\vdash \,  \bangbox \ottnt{A}  \multimap  \ottsym{!}  \ottnt{A} $
  is enough to characterize our modal linear logic (see Section~\ref{sec:axiomatization} for more details).
\end{remark}

The cut-elimination theorem for $\imellbangbox$ is shown similarly to the case of $\imell$, and hence $\imellbangbox$ is consistent.
The addition of `$\bangbox~$' causes no problems in the proof.

\begin{definition}[Cut-degree and degree]
  For an application of $\rcut$ in a proof, its \emph{cut-degree}
  is defined to be the number of logical connectives in the cut-formula.
  The \emph{degree} of a proof is defined to be the maximal cut-degree of the proof\,(and $0$ if there is no application of $\rcut$).
\end{definition}

\begin{theorem}[Cut-elimination]
  The rule $\rcut$ in $\imellbangbox$ is admissible, i.e.,
  if $\Gamma \, \vdash \, \ottnt{A}$ is derivable, then there is a derivation of the same judgment without any applications of $\rcut$.
\end{theorem}

\begin{proof}
  We follow the proof for propositional linear logic by Lincoln et al.\,\cite{L+:decision_problems}.
  To show the admissibility of $\rcut$, we consider the admissibility of the following cut rules:

  \begin{tabular}{c}
    \begin{minipage}{0.45\hsize}
      \begin{prooftree}
          \AXC{$\Gamma \, \vdash \, \ottsym{!}  \ottnt{A}$}
          \AXC{$\Gamma'  \ottsym{,}   ( \ottsym{!}  \ottnt{A} )^{n}  \, \vdash \, \ottnt{B}$}
        \RightLabel{$ \rbcut $}
        \BIC{$\Gamma  \ottsym{,}  \Gamma' \, \vdash \, \ottnt{B}$}
      \end{prooftree}
    \end{minipage}

    \begin{minipage}{0.45\hsize}
      \begin{prooftree}
          \AXC{$\Gamma \, \vdash \,  \bangbox \ottnt{A} $}
          \AXC{$\Gamma'  \ottsym{,}   (  \bangbox \ottnt{A}  )^{n}  \, \vdash \, \ottnt{B}$}
        \RightLabel{$ \rbbcut $}
        \BIC{$\Gamma  \ottsym{,}  \Gamma' \, \vdash \, \ottnt{B}$}
      \end{prooftree}
    \end{minipage}
  \end{tabular}

  \noindent
  where $ ( \ottnt{C} )^{n} $ denotes the multiset that has $n$ occurrences of $C$
  and $n$ is assumed to be positive as a side-condition; and $\Gamma'$ in $\rbcut$ (resp. in $\rbbcut$) is supposed to contain no formulae of form $!A$\,(resp. $ \bangbox \ottnt{A} $).
  The \emph{cut-degrees} of $\rbcut$ and $\rbbcut$ are defined similarly to that of $\rcut$.

  Then, all the three rules ($\rcut$, $\rbcut$, $\rbbcut$) are shown to be admissible by simultaneous induction 
  on the lexicographic complexity $\fbracket{\delta, h}$,
  where $\delta$ is the degree of the assumed derivation and $h$ is its height.
  See the appendix for details of the proof. 
\end{proof}

\begin{corollary}[Consistency]
  $\imellbangbox$ is consistent, i.e., there exists an underivable judgment.
\end{corollary}

\begin{rulefigure}{fig:modal_girard_translation}{Definition of the S4-version of Girard translation.}
  \input{definitions/modal_girard_translation.tex}
\end{rulefigure}

Then, we can define an S4-version of Girard translation as in Figure~\ref{fig:modal_girard_translation}, and it can be justified by the following theorem,
which is readily shown by induction on the derivation.

\begin{theorem}[Soundness]
  \label{thm:modal_girard_translation}
  If $ \Box \Delta   \ottsym{,}  \Gamma \, \vdash \, \ottnt{A}$ in $\ljbox$,
  then $ \bangbox  \fgirardtrans{ \Delta }    \ottsym{,}  \ottsym{!}   \fgirardtrans{ \Gamma }  \, \vdash \,  \fgirardtrans{ \ottnt{A} } $ in $\imellbangbox$.
\end{theorem}

\section{Curry--Howard correspondence}
\label{sec:curry-howard}

In this section, we give a computational interpretation for our modal linear logic through the Curry--Howard correspondence
and establish the corresponding S4-version of Girard translation for the modal linear logic in terms of typed $\lambda$-calculus.

\subsection{Typed $\lambda$-calculus for the intuitionistic modal linear logic}
We introduce $\lambdabangbox$\,(pronounced by ``lambda bangbox'') that is a typed $\lambda$-calculus corresponding to
the modal linear logic under the Curry--Howard correspondence.
The calculus $\lambdabangbox$ can be seen as an integration of $\lambdabox$ of Pfenning and Davies
and the linear $\lambda$-calculus for dual intuitionistic linear logic of Barber\,\cite{B:dual_intuitionistic_linear_logic}.
The rules of $\lambdabangbox$ are designed considering the ``necessity'' of modal logic and the ``linearity'' of linear logic,
and formally defined as in Figure~\ref{fig:lambdabangbox}.

\begin{rulefigure}{fig:lambdabangbox}{Definition of $\lambdabangbox$.}
  \input{definitions/lambdabangbox.tex}
\end{rulefigure}

The structure of types are exactly the same as that of formulae in $\imellbangbox$.
Terms are defined as an extension of the simply-typed $\lambda$-calculus with the following:
the terms $ !  \ottnt{M} $ and $\ottkw{let} \,  !  \ottmv{x}   \ottsym{=}  \ottnt{M} \, \ottkw{in} \, \ottnt{N}$, which are a constructor and a destructor for types $!A$, respectively;
and the terms $ \bangbox  \ottnt{M} $ and $\ottkw{let} \,  \bangbox  \ottmv{x}   \ottsym{=}  \ottnt{M} \, \ottkw{in} \, \ottnt{N}$, which are those for types $ \bangbox \ottnt{A} $ similarly.
Note that the variable $x$ in $\ottkw{let} \,  !  \ottmv{x}   \ottsym{=}  \ottnt{M} \, \ottkw{in} \, \ottnt{N}$ and $\ottkw{let} \,  \bangbox  \ottmv{x}   \ottsym{=}  \ottnt{M} \, \ottkw{in} \, \ottnt{N}$ is supposed to be \emph{bound}.

A \emph{(type) context} is defined by the same way as $\lambdabox$ and a \emph{(type) judgment} consists of three contexts, a term and a type, written as $\Delta  \ottsym{;}  \Gamma  \ottsym{;}  \Sigma \, \vdash \, \ottnt{M}  \ottsym{:}  \ottnt{A}$.
These three contexts of a judgment $\Delta  \ottsym{;}  \Gamma  \ottsym{;}  \Sigma \, \vdash \, \ottnt{M}  \ottsym{:}  \ottnt{A}$ have the following intuitive meaning:
(1) $\Delta$ implicitly represents a context for modalized types of form $ \bangbox \ottnt{A} $;
(2) $\Gamma$ implicitly represents a context for modalized types of form $ \Box \ottnt{A} $;
(3) $\Sigma$ represents an ordinary context but its elements must be used linearly.

The intuitive meanings of the typing rules are as follows.
Each of the first three rules is a variable rule depending on the context's kind.
It is allowed for the $\Delta$-part and the $\Gamma$-part to weaken the antecedent in these rules, 
but is not for the $\Sigma$-part since it must satisfy the linearity condition. 
The rules $\rlimpi$ and $\rlimpe$ are for the type $\slimp$,
and again, the $\rlimpe$ is designed to satisfy the linearity.
The remaining rules are for types $\ottsym{!}  \ottnt{A}$ and $ \bangbox \ottnt{A} $.

The \emph{reduction} $ \leadsto $ is defined to be the least compatible relation on terms generated by $\rbetalimp$, $\rbetabang$, and $\rbetabangbox$.
The \emph{multistep reduction} $ \leadsto^+ $ is defined as in the case of $\lambdabox$.

Then, we can show the subject reduction and the strong normalization of $\lambdabangbox$ as follows. 

\begin{lemma}[Substitution]\ 
  \label{lem:substitution}
  \begin{enumerate}
    \item If $\Delta  \ottsym{;}  \Gamma  \ottsym{;}  \Sigma  \ottsym{,}  \ottmv{x}  \ottsym{:}  \ottnt{A} \, \vdash \, \ottnt{M}  \ottsym{:}  \ottnt{B}$ and $\Delta  \ottsym{;}  \Gamma  \ottsym{;}  \Sigma' \, \vdash \, \ottnt{N}  \ottsym{:}  \ottnt{A}$, then $\Delta  \ottsym{;}  \Gamma  \ottsym{;}  \Sigma  \ottsym{,}  \Sigma' \, \vdash \,  \ottnt{M}  [  \ottmv{x}  :=  \ottnt{N}  ]   \ottsym{:}  \ottnt{B}$;
    \item If $\Delta  \ottsym{;}  \Gamma  \ottsym{,}  \ottmv{x}  \ottsym{:}  \ottnt{A}  \ottsym{;}  \Sigma \, \vdash \, \ottnt{M}  \ottsym{:}  \ottnt{B}$ and $\Delta  \ottsym{;}  \Gamma  \ottsym{;}  \emptyset \, \vdash \, \ottnt{N}  \ottsym{:}  \ottnt{A}$, then $\Delta  \ottsym{;}  \Gamma  \ottsym{;}  \Sigma \, \vdash \,  \ottnt{M}  [  \ottmv{x}  :=  \ottnt{N}  ]   \ottsym{:}  \ottnt{B}$;
    \item If $\Delta  \ottsym{,}  \ottmv{x}  \ottsym{:}  \ottnt{A}  \ottsym{;}  \Gamma  \ottsym{;}  \Sigma \, \vdash \, \ottnt{M}  \ottsym{:}  \ottnt{B}$ and $\Delta  \ottsym{;}  \emptyset  \ottsym{;}  \emptyset \, \vdash \, \ottnt{N}  \ottsym{:}  \ottnt{A}$, then $\Delta  \ottsym{;}  \Gamma  \ottsym{;}  \Sigma \, \vdash \,  \ottnt{M}  [  \ottmv{x}  :=  \ottnt{N}  ]   \ottsym{:}  \ottnt{B}$.
  \end{enumerate}
\end{lemma}

\begin{theorem}[Subject reduction]
  If $\Delta  \ottsym{;}  \Gamma  \ottsym{;}  \Sigma \, \vdash \, \ottnt{M}  \ottsym{:}  \ottnt{A}$ and $\ottnt{M}  \leadsto  \ottnt{N}$,
  then $\Delta  \ottsym{;}  \Gamma  \ottsym{;}  \Sigma \, \vdash \, \ottnt{N}  \ottsym{:}  \ottnt{A}$.
\end{theorem}

\begin{proof}
  By induction on the derivation of $\Delta  \ottsym{;}  \Gamma  \ottsym{;}  \Sigma \, \vdash \, \ottnt{M}  \ottsym{:}  \ottnt{A}$ together with Lemma~\ref{lem:substitution}.
\end{proof}

\begin{theorem}[Strong normalization]
  For well-typed term $M$, there are no infinite reduction sequences starting from $M$.
\end{theorem}

\begin{proof}
  By embedding to a typed $\lambda$-calculus of the $(!, \slimp)$-fragment of dual intuitionistic linear logic, named $\lambdabanglimp$,
  which is shown to be strongly normalizing by Ohta and Hasegawa\,\cite{OH:terminating_linear_lambda}.

  The details are in the appendix, but the intuition is described as follows.  
  First,
  for every well-typed term $M$,
  we define the term $ \fembed{ \ottnt{M} } $ by replacing the occurrences of $ \bangbox  \ottnt{N} $ and $\ottkw{let} \,  \bangbox  \ottmv{x}   \ottsym{=}  \ottnt{N} \, \ottkw{in} \, \ottnt{L}$ in $M$
  with $ !   \fembed{ \ottnt{N} }  $ and $\ottkw{let} \,  !  \ottmv{x}   \ottsym{=}   \fembed{ \ottnt{N} }  \, \ottkw{in} \,  \fembed{ \ottnt{L} } $, respectively.
  Then, we can show that $ \fembed{ \ottnt{M} } $ is typable in $\lambdabanglimp$,
  because the structure of `$\bangbox$~' collapses to that of `$!$',
  and that the embedding $\fembed{-}$ preserves reductions. Therefore, $\lambdabangbox$ is strongly normalizing.
\end{proof}

As we mentioned, we can view that $\lambdabangbox$ is indeed a typed $\lambda$-calculus for the intuitionistic modal linear logic.
A natural deduction that corresponds to $\lambdabangbox$
is obtained as the ``logical-part'' of the calculus,
and we can show that the natural deduction is equivalent to $\imellbangbox$.

\begin{definition}[Natural deduction]
  A natural deduction for modal linear logic,
  called $\njbangbox$, is defined to be one that is extracted from $\lambdabangbox$ by erasing term annotations.
\end{definition}

\begin{fact}[Curry--Howard correspondence]
  There is a one-to-one correspondence between $\njbangbox$ and $\lambdabangbox$,
  which preserves provability/typability and proof-normalizability/reducibility.
\end{fact}

\begin{lemma}[Judgmental reflection]
\label{lem:reflection}
  The following hold in $\njbangbox$.
  \begin{enumerate}
    \item $\Delta  \ottsym{;}  \Gamma  \ottsym{;}  \Sigma  \ottsym{,}  \ottsym{!}  \ottnt{A} \, \vdash \, \ottnt{B}$ if and only if $\Delta  \ottsym{;}  \Gamma  \ottsym{,}  \ottnt{A}  \ottsym{;}  \Sigma \, \vdash \, \ottnt{B}$;
    \item $\Delta  \ottsym{;}  \Gamma  \ottsym{;}  \Sigma  \ottsym{,}   \bangbox \ottnt{A}  \, \vdash \, \ottnt{B}$ if and only if $\Delta  \ottsym{,}  \ottnt{A}  \ottsym{;}  \Gamma  \ottsym{;}  \Sigma \, \vdash \, \ottnt{B}$.
  \end{enumerate}
\end{lemma}

\begin{theorem}[Equivalence]
  \label{thm:equivalence}
  $\Delta  \ottsym{;}  \Gamma  \ottsym{;}  \Sigma \, \vdash \, \ottnt{A}$ in $\njbangbox$ if and only if 
  $ \bangbox \Delta   \ottsym{,}  \ottsym{!}  \Gamma  \ottsym{,}  \Sigma \, \vdash \, \ottnt{A}$ in $\imellbangbox$.
\end{theorem}

\begin{proof}
  By straightforward induction. 
  Lemma~\ref{lem:reflection} is used to show the if-part.
\end{proof}

\subsection{Embedding from the modal $\lambda$-calculus by Pfenning and Davies} 

We give a translation from Pfenning and Davies' $\lambdabox$ to our $\lambdabangbox$.
We also show that the translation preserves the reductions of $\lambdabox$,
and thus it can be seen as the S4-version of Girard translation on the level of proofs through the Curry--Howard correspondence. 

To give the translation, we introduce two meta $\lambda$-terms in $\lambdabangbox$ to encode the function space $\sarrow$ of $\lambdabox$.
The simulation of reduction of $ \ottsym{(}  \lambda  \ottmv{x}  \ottsym{:}  \ottnt{A}  \ottsym{.}  \ottnt{M}  \ottsym{)} \, \ottnt{N} $ in $\lambdabox$ can be shown readily.

\begin{definition}
  Let $M$ and $N$ be terms such that 
  $\Delta  \ottsym{;}  \Gamma  \ottsym{,}  \ottmv{x}  \ottsym{:}  \ottnt{A}  \ottsym{;}  \Sigma \, \vdash \, \ottnt{M}  \ottsym{:}  \ottnt{B}$ and $\Delta  \ottsym{;}  \Gamma  \ottsym{;}  \emptyset \, \vdash \, \ottnt{N}  \ottsym{:}  \ottnt{A}$.
  Then, $\overline{\lambda}  \ottmv{x}  \ottsym{:}  \ottnt{A}  \ottsym{.}  \ottnt{M}$ and $\ottnt{M}  \overline{@}  \ottnt{N}$ are defined as the terms 
$\lambda  \ottmv{y}  \ottsym{:}  \ottsym{!}  \ottnt{A}  \ottsym{.}  \ottkw{let} \,  !  \ottmv{x}   \ottsym{=}  \ottmv{y} \, \ottkw{in} \, \ottnt{M}$ and 
$ \ottnt{M} \, \ottsym{(}   !  \ottnt{N}   \ottsym{)} $, respectively, where $y$ is chosen to be fresh, i.e., it is a variable satisfying $y \not\in (\fv{M} \cup \{ x \})$.
\end{definition}

\begin{lemma}[Derivable full-function space]
  The following rules are derivable in $\lambdabangbox$:

  \begin{tabular}{c}
    \begin{minipage}{0.37\hsize}
      \begin{prooftree}
          \AXC{$\Delta  \ottsym{;}  \Gamma  \ottsym{,}  \ottmv{x}  \ottsym{:}  \ottnt{A}  \ottsym{;}  \Sigma \, \vdash \, \ottnt{M}  \ottsym{:}  \ottnt{B}$}
        \doubleLine
        \UIC{$\Delta  \ottsym{;}  \Gamma  \ottsym{;}  \Sigma \, \vdash \, \ottsym{(}  \overline{\lambda}  \ottmv{x}  \ottsym{:}  \ottnt{A}  \ottsym{.}  \ottnt{M}  \ottsym{)}  \ottsym{:}  \ottsym{!}  \ottnt{A}  \multimap  \ottnt{B}$}
      \end{prooftree}
    \end{minipage}

    \begin{minipage}{0.53\hsize}
      \begin{prooftree}
          \AXC{$\Delta  \ottsym{;}  \Gamma  \ottsym{;}  \Sigma \, \vdash \, \ottnt{M}  \ottsym{:}  \ottsym{!}  \ottnt{A}  \multimap  \ottnt{B}$}
          \AXC{$\Delta  \ottsym{;}  \Gamma  \ottsym{;}  \emptyset \, \vdash \, \ottnt{N}  \ottsym{:}  \ottnt{A}$}
        \doubleLine
        \BIC{$\Delta  \ottsym{;}  \Gamma  \ottsym{;}  \Sigma \, \vdash \, \ottnt{M}  \overline{@}  \ottnt{N}  \ottsym{:}  \ottnt{B}$}
      \end{prooftree}
    \end{minipage}
  \end{tabular}

  \noindent
  Moreover, it holds that $\ottsym{(}  \overline{\lambda}  \ottmv{x}  \ottsym{:}  \ottnt{A}  \ottsym{.}  \ottnt{M}  \ottsym{)}  \overline{@}  \ottnt{N}  \leadsto^+   \ottnt{M}  [  \ottmv{x}  :=  \ottnt{N}  ] $ in $\lambdabangbox$.
\end{lemma}

\begin{rulefigure}{fig:translation}{Definitions of the S4-version of Girard translation in term of typed $\lambda$-calculus.}
  \input{definitions/modal_embedding.tex}
\end{rulefigure}

Together with the above meta $\lambda$-terms $\overline{\lambda}  \ottmv{x}  \ottsym{:}  \ottnt{A}  \ottsym{.}  \ottnt{M}$ and $\ottnt{M}  \overline{@}  \ottnt{N}$,
we can define the translation from $\lambdabox$ into $\lambdabangbox$
and show that it preserves typability and reducibility.

\begin{definition}[Translation]
  The \emph{translation} from $\lambdabox$ to $\lambdabangbox$
  is defined to be the triple of
  the type/context/term translations $ \fgirardtrans{ \ottnt{A} } $, $ \fgirardtrans{ \Gamma } $, and $ \ftrans{ \ottnt{M} } $ defined in Figure~\ref{fig:translation}.
\end{definition}

\begin{theorem}[Embedding]
  \label{thm:embedding}
  $\lambdabox$ can be embedded into $\lambdabangbox$, i.e., the following hold:
  \begin{enumerate}
    \item If $\Delta  \ottsym{;}  \Gamma \, \vdash \, \ottnt{M}  \ottsym{:}  \ottnt{A}$ in $\lambdabox$, then $ \fgirardtrans{ \Delta }   \ottsym{;}   \fgirardtrans{ \Gamma }   \ottsym{;}  \emptyset \, \vdash \,  \ftrans{ \ottnt{M} }   \ottsym{:}   \fgirardtrans{ \ottnt{A} } $ in $\lambdabangbox$.
    \item If $\ottnt{M}  \leadsto  \ottnt{M'}$ in $\lambdabox$, then $ \ftrans{ \ottnt{M} }   \leadsto^+   \ftrans{ \ottnt{M'} } $ in $\lambdabangbox$.
  \end{enumerate}
\end{theorem}

\begin{proof}
  By induction on the derivation of $\Delta  \ottsym{;}  \Gamma \, \vdash \, \ottnt{M}  \ottsym{:}  \ottnt{A}$ and $\ottnt{M}  \leadsto  \ottnt{M'}$ in $\lambdabox$, respectively.
\end{proof}

From the logical point of view, Theorem~\ref{thm:embedding}.1 can be seen as another S4-version of Girard translation (in the style of natural deduction) that corresponds to Theorem~\ref{thm:modal_girard_translation}; and Theorem~\ref{thm:embedding}.2 gives a justification that the S4-version of Girard translation is correct with respect to the level of proofs, i.e., it preserves proof-normalizations as well as provability.

\section{Axiomatization of modal linear logic}
\label{sec:axiomatization}

\begin{rulefigure}{fig:clbangbox}{Definition of $\clvbangbox$.}
  \input{definitions/clbangbox.tex}
\end{rulefigure}

\noindent
We give an axiomatic characterization of the intuitionistic modal linear logic.
To do so, we define a typed combinatory logic, called $\clvbangbox$, which can be seen as a Hilbert-style deductive system of modal linear logic through the Curry--Howard correspondence.
In this section, we only aim to provide the equivalence between $\njbangbox$ and the Hilbert-style,
while $\clvbangbox$ satisfies several desirable properties, e.g., the subject reduction and the strong normalizability.

The definition of $\clvbangbox$ is given in Figure~\ref{fig:clbangbox}.
The set of types has the same structure as that in $\lambdabangbox$.
A term is either a \emph{variable}, a \emph{combinator}, a necessitated term by `$!$', or a necessitated term by `$\bangbox$~'.
The notions of \emph{(type) context} and \emph{(type) judgment} are defined similarly to those of $\lambdabangbox$.

Every combinator $c$ has its type as defined in the list in the figure, and is denoted by $\ftypeof{c}$.
Then, the typing rules are described as follows: $\hcomb$ and $\hmp$ are the standard rules,
which logically correspond to an axiom rule of the set of axiomata, and \emph{modus ponens}, respectively.
The others are defined by the same way as in $\lambdabangbox$.

The \emph{reduction} $ \leadsto $ of combinators is defined to be the least compatible relation on terms generated by the reduction rules listed in the figure.

\begin{remark}
$\clvbangbox$
can be seen as
an extension of \emph{linear combinatory algebra} of Abramsky et al.\,\cite{AHS:goi_and_lca} with the $\bangbox~$-modality,
or equivalently, a linear-logical reconstruction of Pfenning's modally-typed combinatory logic\,\cite{P:modal_combinatory_logic}.
The combinators $ \mathrm{T}^! $, $ \mathrm{D}^! $, $ \mathrm{4}^! $ represent the S4 axiomata for the $!$-modality,
and similarly, $ \mathrm{T}^{\smallbangbox} $, $ \mathrm{D}^{\smallbangbox} $, $ \mathrm{4}^{\smallbangbox}\, $ represent those for the $\bangbox~$-modality.
$ \mathrm{E} $ is the only one combinator to characterize the strength between the two modalities.
\end{remark}

As we defined $\njbangbox$ from $\lambdabangbox$, 
we can define the Hilbert-style deductive system (with open assumptions) for the intuitionistic modal linear logic 
via $\clvbangbox$.

\begin{definition}[Hilbert-style]
  A Hilbert-style deductive system for modal linear logic, called $\hjbangbox$,
  is defined to be one that is extracted from $\clvbangbox$ by erasing term annotations.
\end{definition}

\begin{fact}[Curry--Howard correspondnece]
  \label{fact:Hilbert_curry-howard}
  There is a one-to-one correspondence between $\hjbangbox$ and $\clvbangbox$,
  which preserves provability/typability and proof-normalizability/reducibility.
\end{fact}

The deduction theorem of $\hjbangbox$ can be obtained as a consequence of the so-called
\emph{bracket abstraction} of $\clvbangbox$ through Fact~\ref{fact:Hilbert_curry-howard},
which allows us to show the equivalence between $\hjbangbox$ and $\njbangbox$.
Therefore, the modal linear logic is indeed axiomatized by $\hjbangbox$. 

\begin{theorem}[Deduction theorem]\ 
  \label{thm:deduction_theorem}
  \begin{enumerate}
    \item If $\Delta  \ottsym{;}  \Gamma  \ottsym{;}  \Sigma  \ottsym{,}  \ottmv{x}  \ottsym{:}  \ottnt{A} \, \vdash \, \ottnt{M}  \ottsym{:}  \ottnt{B}$, then $\Delta  \ottsym{;}  \Gamma  \ottsym{;}  \Sigma \, \vdash \, \ottsym{(}   \lambda_{*}  \ottmv{x} .  \ottnt{M}   \ottsym{)}  \ottsym{:}   (  \ottnt{A}  \multimap  \ottnt{B}  ) $;
    \item If $\Delta  \ottsym{;}  \Gamma  \ottsym{,}  \ottmv{x}  \ottsym{:}  \ottnt{A}  \ottsym{;}  \Sigma \, \vdash \, \ottnt{M}  \ottsym{:}  \ottnt{B}$, then $\Delta  \ottsym{;}  \Gamma  \ottsym{;}  \Sigma \, \vdash \, \ottsym{(}   \lambda^{!}_{*}  \ottmv{x} .  \ottnt{M}   \ottsym{)}  \ottsym{:}   (  \ottsym{!}  \ottnt{A}  \multimap  \ottnt{B}  ) $;
    \item If $\Delta  \ottsym{,}  \ottmv{x}  \ottsym{:}  \ottnt{A}  \ottsym{;}  \Gamma  \ottsym{;}  \Sigma \, \vdash \, \ottnt{M}  \ottsym{:}  \ottnt{B}$, then $\Delta  \ottsym{;}  \Gamma  \ottsym{;}  \Sigma \, \vdash \, \ottsym{(}   \lambda^{\smallbangbox}_{*}  \ottmv{x} .  \ottnt{M}   \ottsym{)}  \ottsym{:}   (   \bangbox \ottnt{A}   \multimap  \ottnt{B}  ) $.
  \end{enumerate}
  where $\ottsym{(}   \lambda_{*}  \ottmv{x} .  \ottnt{M}   \ottsym{)}$, $\ottsym{(}   \lambda^{!}_{*}  \ottmv{x} .  \ottnt{M}   \ottsym{)}$, $\ottsym{(}   \lambda^{\smallbangbox}_{*}  \ottmv{x} .  \ottnt{M}   \ottsym{)}$
  are bracket abstraction operations that take a variable $x$ and a $\clvbangbox$-term $M$ 
  and returns a $\clvbangbox$-term,
  and the definitions are given in the appendix.
\end{theorem}

\begin{proof}
  By induction on the derivation. The proof is just a type-checking of the result of the bracket abstraction operations.
\end{proof}

\begin{theorem}[Equivalence]
  $\Delta  \ottsym{;}  \Gamma  \ottsym{;}  \Sigma \, \vdash \, \ottnt{A}$ in $\hjbangbox$ if and only if $\Delta  \ottsym{;}  \Gamma  \ottsym{;}  \Sigma \, \vdash \, \ottnt{A}$ in $\njbangbox$.
\end{theorem}

\begin{proof}
  By straightforward induction. 
  We use Theorem~\ref{thm:deduction_theorem} and Fact~\ref{fact:Hilbert_curry-howard} to show the if-part.
\end{proof}

\begin{corollary}
  $\imellbangbox$, $\njbangbox$, and $\hjbangbox$ are equivalent with respect to provability.
\end{corollary}

\section{Geometry of Interaction Machine}
\label{sec:goim}

In this section, we show a dynamic semantics, called \emph{context semantics}, for the modal linear logic
in the style of geometry of interaction machine\,\cite{M:goim, M:goim_popl}.
As in the usual linear logic, we first define a notion of \emph{proof net} and then
define the machine as a token-passing system over those proof nets.
Thanks to the simplicity of our logic, the definitions are mostly straightforward extension of those for classical $\mell$\,($\cmell$). 

\subsection{Sequent calculus for classical modal linear logic}

We define a sequent calculus of classical modal linear logic, called $\cmellbangbox$.
The reason why we define it in the classical setting is for ease of defining the proof nets in the latter part.

\begin{rulefigure}{fig:cmellbangbox}{Definition of $\cmellbangbox$.}
  \input{definitions/cmellbangbox.tex}
\end{rulefigure}

Figure~\ref{fig:cmellbangbox} shows the definition of $\cmellbangbox$.
The set of formulae are defined as an extension of $\cmell$-formulae with the two modalities `$\bangbox$\,' and `$\whynotdia$\,'.
A \emph{dual formula} of $A$, written $ \ottnt{A} ^\bot $, is defined by the standard dual formulae in $\cmell$
along with $  (   \bangbox \ottnt{A}   )  ^\bot  \defeq  \whynotdia  (   \ottnt{A} ^\bot   )  $ and $  (   \whynotdia \ottnt{A}   )  ^\bot  \defeq  \bangbox  (   \ottnt{A} ^\bot   )  $.
Here, the $\whynotdia~$-modality is the dual of the $\bangbox~$-modality by definition, and it can be seen as an integration of the $?$-modality and the $\Diamond$-modality.
The \emph{linear implication} $\ottnt{A}  \multimap  \ottnt{B}$ is defined as $  \ottnt{A} ^\bot   \spar  \ottnt{B} $ as usual.
The inference rules are defined as a simple extension of $\imellbangbox$ to the classical setting in the style of ``one-sided'' sequent.

Then, the cut-elimination theorem for $\cmellbangbox$ can be shown similarly to the case of $\imellbangbox$,
and we can see that there exists a trivial embedding from $\imellbangbox$ to $\cmellbangbox$.

\begin{theorem}[Cut-elimination]
  The rule $\rcut$ in $\cmellbangbox$ is admissible.
\end{theorem}

\begin{theorem}[Embedding]
  If $\Gamma \, \vdash \, \ottnt{A}$ in $\imellbangbox$, then $\vdash \,  \Gamma ^\bot   \ottsym{,}  \ottnt{A}$ in $\cmellbangbox$.
\end{theorem}

\subsection{Proof-nets formalization}
First, we define \emph{proof structures} for $\cmellbangbox$.
The \emph{proof nets} are then defined to be those proof structures satisfying a condition called \emph{correctness criterion}.
Intuitively, a proof net corresponds to an (equivalence class of) proof in $\cmellbangbox$.
\begin{definition}
  A \emph{node} is one of the graph-theoretic node shown in Figure~\ref{fig:nodesAndBox}
  equipped with $\cmellbangbox$ types on the edges.
  They are all directed from top to bottom: for example, the $\mathsf{\parr}$ node has two
  incoming edges and one outgoing edge.
  A $\rbang$-node (resp.\ $\bangbox$~-node) has one outgoing edge typed by $\rbang A$
  (resp. $\bangbox A$)
  and arbitrarily many (possibly zero) outgoing edges typed by $\rwhynot A_i$ and $\rwhynotdia B_i$ (resp.\ $\rwhynotdia A_i$).

  \begin{figure}[hbtp]
    \centering
    \includegraphics[scale=.7]{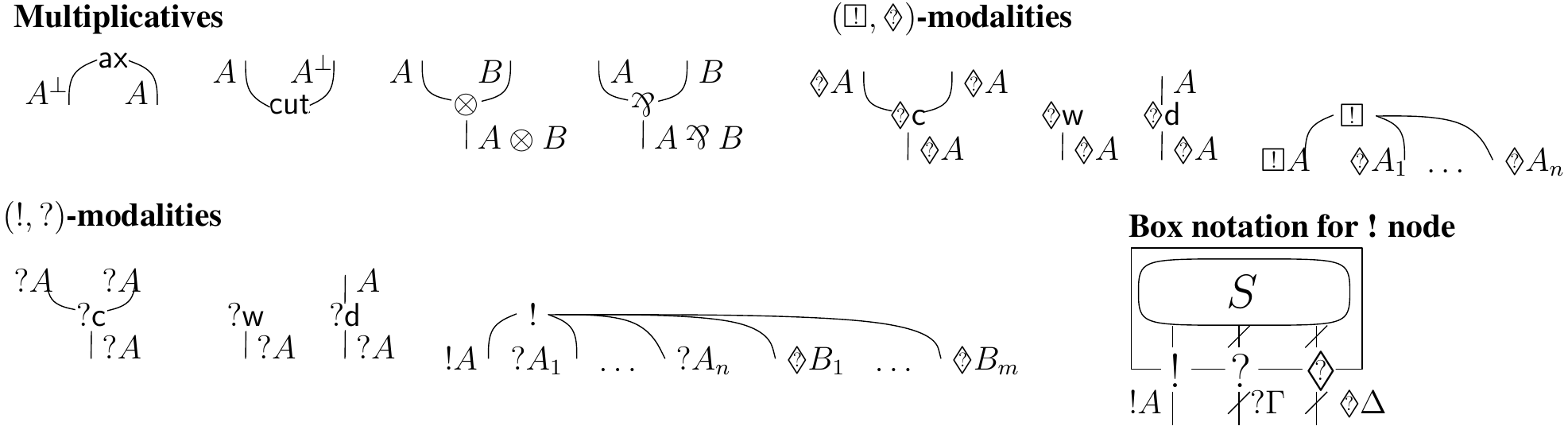}
    \caption{Nodes of proof net and box notation.}
    \label{fig:nodesAndBox}
  \end{figure}

  A \emph{proof structure} is a finite directed graph that satisfies the following conditions:
  \begin{itemize}
  \item each edge is with a type that matches the types specified by the nodes (in Figure~\ref{fig:nodesAndBox}) it is connected to;
  \item some edges may not be connected to any node (called \emph{dangling edges}).
    Those dangling edges and also the types on those edges are called the
    \emph{conclusions of the structure};
  \item the graph is associated with a total map from all the $\rbang$-nodes and
    $\bangbox$~-nodes in it to proof structures called the \emph{contents} of the $\rbang$/$\bangbox$~-nodes.
    The map satisfies that the types of the conclusions of a $\rbang$-node (resp.\ $\bangbox$~-node)
    coincide with the conclusions of its content.
  \end{itemize}
\end{definition}

\begin{remark}
  Formally, a $\rbang$-node (resp.\ a $\bangbox$~-node) and its content are distinctive
  objects and they are not connected as a directed graph.
  Though, it is convenient to depict them as if the $\rbang$-node (resp.\ $\bangbox$~-node)
  represents a ``box'' filled with its content, as shown at the bottom-right of Figure~\ref{fig:nodesAndBox}.
  We also depict multiple edges by an edge with a diagonal line.
  In what follows, we adopt this ``box'' notation and multiple edges notation without explicit note.
\end{remark}

\newcommand{\net}{S}

\begin{definition}\label{def:switchingPath}
  Given a proof structure $\net$, a \emph{switching path} is an undirected path on $\net$
  (meaning that the path is allowed to traverse an edge forward or backward)
  satisfying that on each $\parr$ node, $?c$ node, and $\whynotdia c$ node,
  the path uses at most one of the premises, and that the path uses any edge at most once.
\end{definition}

\begin{definition}\label{def:modalMellCorrectCrit}
  The \emph{correctness criterion} is the following condition:
  given a proof structure $\net$, switching paths of $\net$ and all contents of $!$-nodes, $\bangbox~$-nodes in $\net$ are all acyclic and connected.
  A proof structure satisfying the correctness criterion is called a \emph{proof net}.
\end{definition}

As a counterpart of cut-elimination process in $\cmellbangbox$, the notion of \emph{reduction} is defined for proof structures (and hence for proof nets): this intuition is made precise by Lemma~\ref{lem:simulation_seqnet} where $\fnettrans{-}$ is the translation from $\cmellbangbox$ to proof nets, 
whose definition is omitted here since it is defined analogously to that of $\cmell$ and $\cmell$ proof net. The lemmata below are naturally obtained by extending the case for $\cmell$ since the $\bangbox~$-modality has mostly the same logical structure as the $\rbang$-modality.

\begin{definition}
  \emph{Reductions} of proof structures are local graph reductions defined by the set of rules
  depicted in Figure~\ref{fig:reduction_rules}.
\end{definition}

\begin{figure}[hbtp]
  \centering
  \includegraphics[scale=.684]{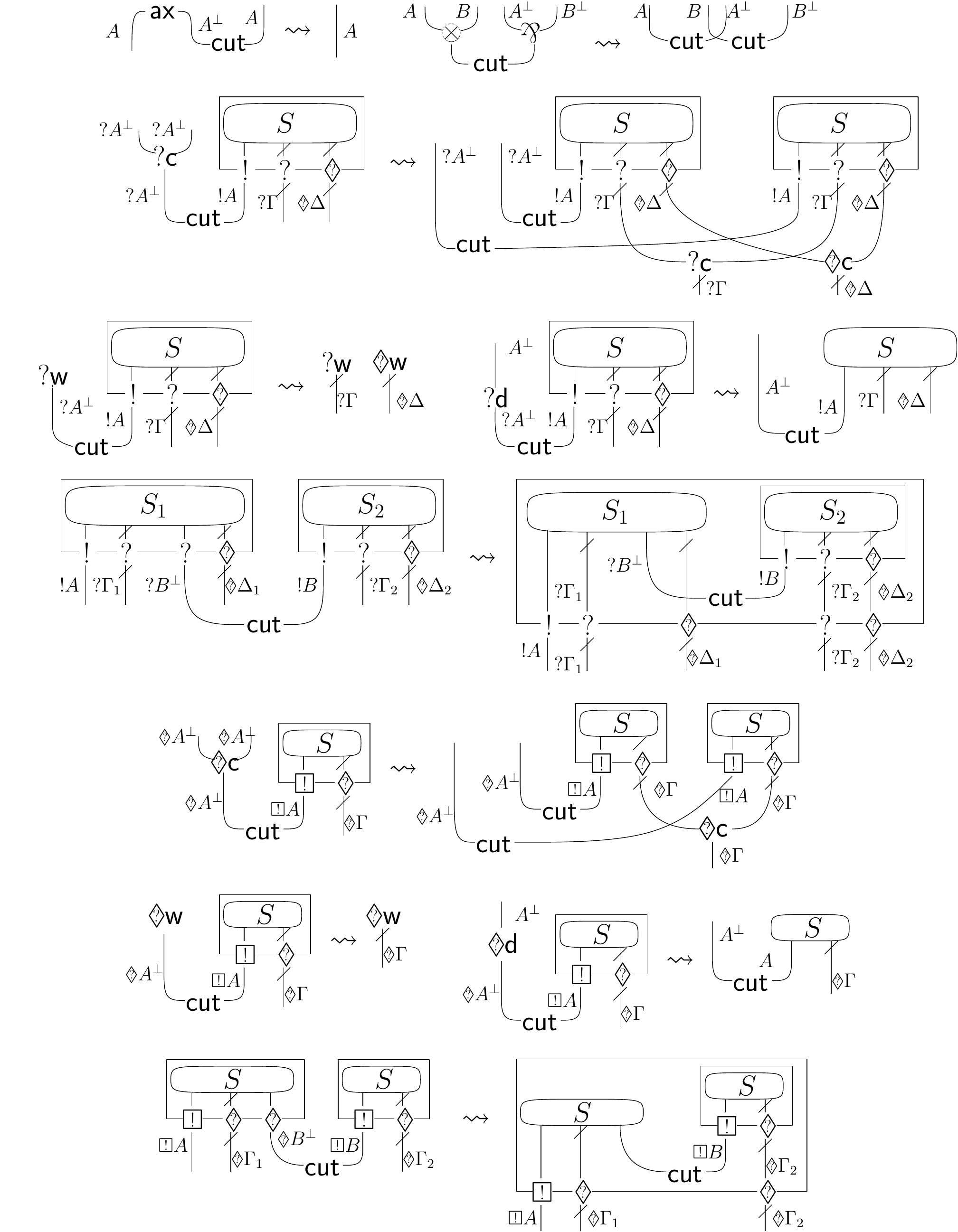}
  \caption{Reduction rules.}
  \label{fig:reduction_rules}
\end{figure}

\begin{lemma} \label{lem:net_preservation}
  Let $\net \to \net'$ be a reduction between proof structures.
  If $\net$ is a proof net (i.e., satisfies the correctness criterion),
  so is $\net'$.
\end{lemma}

\begin{lemma} \label{lem:simulation_seqnet}
  Let $\Pi$ be a proof of $\vdash \,  \whynotdia  \Delta ^\bot    \ottsym{,}  ?   \Gamma ^\bot   \ottsym{,}   \Sigma ^\bot   \ottsym{,}  \ottnt{A}$ and
  suppose that $\Pi$ reduces to another proof $\Pi'$.
  Then there is a sequence of reductions $\fnettrans{\Pi} \to^* \fnettrans{\Pi'}$ between the proof nets.
\end{lemma}

\subsection{Computational interpretation}

\begin{definition}
  A \emph{context} is a triple $(\mcont, \bcont, \ncont)$ where $\mcont, \bcont, \ncont$ are generated by the following grammar:

  \vspace{0.5em}
  {\small
    $
    \mcont ::= \varepsilon ~|~ l.\mcont ~|~ r.\mcont
    \quad \bcont ::= \varepsilon ~|~ L.\bcont ~|~ R.\bcont ~|~ \fbracket{\bcont,\bcont} ~|~ \star
    \quad \ncont ::= \varepsilon ~|~ L'.\,\ncont ~|~ R'.\,\ncont ~|~ \fbracket{\ncont,\ncont} ~|~ \star
    $
  }
\end{definition}

The intuition of a context is an intermediate state while ``evaluating'' the proof net
(and, by translating into a proof net, a term in $\lambdabangbox$).
The geometry of interaction machine calculates the semantic value of a net by traversing
the net from a conclusion to another;
to traverse the net in a ``right way'' (more precisely, in a way invariant under net reduction),
the context accumulates the information about the path that is already passed.
Then, \emph{how} the net is traversed is defined by the notion of \emph{path} over a proof net as we define below.

\begin{definition}
  The \emph{extended dynamic algebra} $\Lambda^{\Box *}$ is a single-sorted $\Sigma$ algebra
  that contains $0,1,p,q,r,r',s,s',t,t',d,d' \colon \Sigma$ as constants,
  has an associative operator $\cdot\colon\Sigma\times\Sigma\to\Sigma$
  and operators $(-)^*\colon\Sigma\to\Sigma$, ${!}\colon\Sigma\to\Sigma$, ${\bangbox{}}\colon\Sigma\to\Sigma$,
  equipped with a formal sum $+$,
  and satisfies the equations below.
  Hereafter,
  we write $xy$ for $x\cdot y$ where $x$ and $y$ are metavariables over $\Sigma$.
  {\small
    \begin{alignat*}{10}
      \qquad0^* = !0 = 0 &\quad 1^* = !1 = 1 &\quad 0x = x0 = 0 &\quad 1x = x1 = x\\
      \qquad!(x)^* = !(x^*) &\quad (xy)^* = y^*x^* &\quad (x^*)^* = x &\quad !(x)!(y) = !(xy) \\
      \qquad\bangbox(x)\bangbox(y) = \bangbox(xy) &\quad p^*p = q^*q = 1 &\quad q^*p = p^*q = 0 &\quad r^*r = s^*s = 1\\
      \qquad s^*r = r^*s = 0 &\quad d^*d= 1 &\quad t^*t = 1 &\quad p'^*p' = q'^*q' = 1\\
      \qquad q'^*p' = p'^*q' = 0 &\quad r'^*r' = s'^*s' = 1 &\quad s'^*r' = r'^*s' = 0 &\quad d'^*d'= 1 \\
      \qquad t'^*t' = 1  &\quad !(x)r = r!(x) &\quad !(x)s = s!(x) &\quad !(x)t = t!!(x)\\
      \qquad !(x)d = dx &\quad \bangbox(x)r' = r'\bangbox(x) &\quad \bangbox(x)s' = s'\bangbox(x) &\quad \bangbox(x)t' = t'\bangbox~\bangbox(x)\\
      \qquad\bangbox(x)d' = d'x  &\quad x+y = y+x &\quad x+0 = x&\quad (x+y)z = xz+yz\\
      \qquad z(x+y) = zx+zy &\quad (x+y)^* = x^*+y^* &\quad !(x+y) = !x+!y &\quad \bangbox{(x+y)} = \bangbox{x}+\bangbox{y}
    \end{alignat*}
  }
\end{definition}

\begin{remark}
  The equations in the definition above are mostly the same as the standard dynamic algebra
  $\Lambda^*$~\cite{M:goim, M:goim_popl} except those equations concerning the symbols with $'$ and the operator $\bangbox{}$, and their structures are analogous to those for $!$ operator.
  This again reflects the fact that the logical structure of rules for $\bangbox{}$ is analogous to that of $!$.
\end{remark}

\begin{definition}
  A \emph{label} is an element of $\Lambda^{\Box *}$ that is associated to
  edges of proof structures as in Figure~\ref{fig:labels}.
  Let $\net$ be a proof structure and $T_\net$ be the set of edge traversals in the structure.
  $\net$ is associated with a function $w\colon T_\net \to \Lambda^{\Box *}$
  defined by $w(e) = l$ (resp.\ $l^*$) if $e$ is a forward (resp.\ backward) traversal
  of an edge $e$ and $l$ is the label of the edge;
  $w(e_1e_2) = w(e_1)w(e_2)$.
  \begin{figure}[hbtp]
    \centering
    \includegraphics[scale=.7]{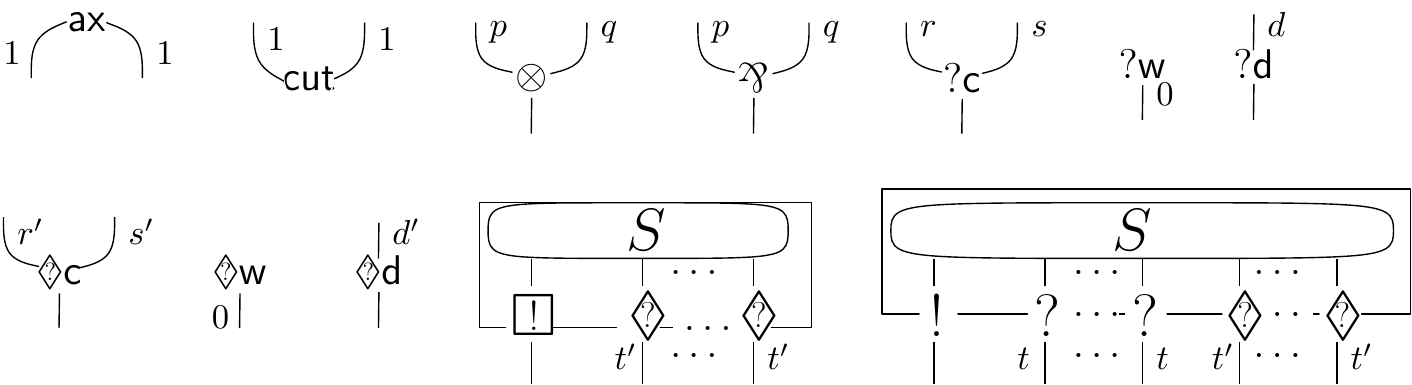}
    \caption{Labels on edges.}
    \label{fig:labels}
  \end{figure}
\end{definition}

\begin{definition}
  A \emph{walk} over a proof structure $\net$ is an element of $\Lambda^{\Box *}$
  that is obtained by concatenating labels along a graph-theoretic path over $\net$
  such that the graph-theoretic path does not traverse an edge forward (resp.\ backward)
  immediately after the same edge backward (resp.\ forward);
  and does not traverse a premise of one of $\otimes, \parr, c$ node and another premise of
  the same node immediately after that.
  A \emph{path} is a walk that is not proved to be equal to $0$.
  A path is called \emph{maximal} if it starts and finishes at a conclusion.
\end{definition}

The intuition of the notion of path is that a path is a ``correct way'' of traversing
a proof net, in the sense that any path is preserved before and after a reduction.
All the other walks that are not paths will be broken, which is represented by the constant $0$ of $\Lambda^{\Box *}$.
Then, we obtain a \emph{context semantics} from paths in the following way.

\begin{definition}
  Given a monomial path $a$, its action $\fint{a} : \Sigma \rightharpoonup \Sigma$ on contexts is defined as follows.
  We define $\fint{1}$ as the identity mapping on contexts.
  There is no definition of $\fint{0}$.
  The $\fint{f^*}$ is the inverse translation, i.e., $\fint{f}^{-1}$.
  The transformer of the composition of $a$ and $b$ is defined as
  $\fint{ab}(m) \defeq \fint{a}(\fint{b}(m))$.
  For the other labels, the interpretation are defined as follows where exponential morphisms $!$ and $\bangbox\,$ are defined by the meta-level pattern matchings:
    \begin{alignat*}{4}
      \fint{p}(\mcont, \bcont, \ncont) &\defeq (l . \mcont, \bcont, \ncont)
      &\quad \fint{q}(\mcont, \bcont, \ncont) &\defeq (r . \mcont, \bcont, \ncont)\\
      \fint{r}(\mcont, \bcont, \ncont) &\defeq (\mcont, L . \bcont, \ncont)
      &\quad \fint{s}(\mcont, \bcont, \ncont) &\defeq (\mcont, R . \bcont, \ncont)\\
      \fint{t}(\mcont, \fbracket{\bcont_1, \fbracket{\bcont_2, \bcont_3}}, \ncont)
      &\defeq (\mcont, \fbracket{\fbracket{\bcont_1, \bcont_2}, \bcont_3}, \ncont)
      &\quad \fint{d}(\mcont, \bcont, \ncont) &\defeq (\mcont, \star . \bcont, \ncont)\\
      \fint{r'}(\mcont, \bcont, \ncont) &\defeq (\mcont, \bcont, L' . \ncont)
      &\quad \fint{s'}(\mcont, \bcont, \ncont) &\defeq (\mcont, \bcont, R' . \ncont)\\
      \fint{t'}(\mcont, \bcont, \fbracket{\ncont_1, \fbracket{\ncont_2, \ncont_3}})
      &\defeq (\mcont, \bcont, \fbracket{\fbracket{\ncont_1, \ncont_2}, \ncont_3})
      &\quad \fint{d'}(\mcont, \bcont, \ncont) &\defeq (\mcont, \bcont, \star . \ncont)\\
    \end{alignat*}
    \vspace{-1.3cm} 
    \begin{align*}
      \fint{!(f)}(\mcont, \fbracket{\bcont_1, \bcont_2}, \ncont)
      &\defeq \textbf{let}~(\mcont', \bcont_2', \ncont') = \fint{f}(\mcont, \bcont_2, \ncont)~\textbf{in}~(\mcont', \fbracket{\bcont_1, \bcont_2'}, \ncont') \\
      \fint{\bangbox(f)}(\mcont, \bcont, \fbracket{\ncont_1, \ncont_2})
      &\defeq \textbf{let}~(\mcont', \bcont', \ncont_2') = \fint{f}(\mcont, \bcont, \ncont_2)~\textbf{in}~(\mcont', \bcont', \fbracket{\ncont_1, \ncont_2'})
    \end{align*}
    Given a path $a$, its action $\fint{a} : \Sigma \rightharpoonup \mset(\Sigma)$ is 
    defined by the rules above (regarding the codomain as a multiset) and
    $\fint{a+b}(m) = (\fint{a}(m)) \uplus (\fint{b}(m))$ where $\uplus$ is the multiset sum.
\end{definition}

\begin{remark}
  In Mackie's work~\cite{M:goim}, the multiset in the codomain is not used since the main interest of his
  work is on terms of a base type: in that setting any proof net corresponding to a term
  has an execution formula that is monomial.
  In general, this style of context semantics
  is slightly degenerated compared to Girard's original version and its successors
  because the information of ``current position'' is dropped from the definition of contexts.
\end{remark}

\begin{definition}
  \label{def:context_semantics} Let $\net$ be a closed proof net and $\chi$ be the set of maximal paths between conclusions of $\net$.
  The \emph{execution formula} is defined by $\mathcal{EX}(\net) = \Sigma_{\phi \in \chi}\phi$ where the RHS is the sum of all paths in $\chi$.
  The \emph{context semantics} of $\net$ is defined to be $\fint{\mathcal{EX}(\net)}\colon\Sigma \rightharpoonup \mset(\Sigma)$.
\end{definition}

\begin{definition}
  Let $M$ be a closed well-typed term in $\lambdabangbox$.
  The \emph{context semantics} of $M$ is defined to be $\fint{ \flambdanettrans{M} }$, where
  $\flambdanettrans{-}$ is a straightforward translation from $\lambdabangbox$-terms to proof nets, defined by
  constructing proof nets from $\lambdabangbox$-derivations
  as in Figure~\ref{fig:bangboxnettranslation}.
\end{definition}

\begin{figure}[hbtp]
  \centering
  \includegraphics[scale=.70]{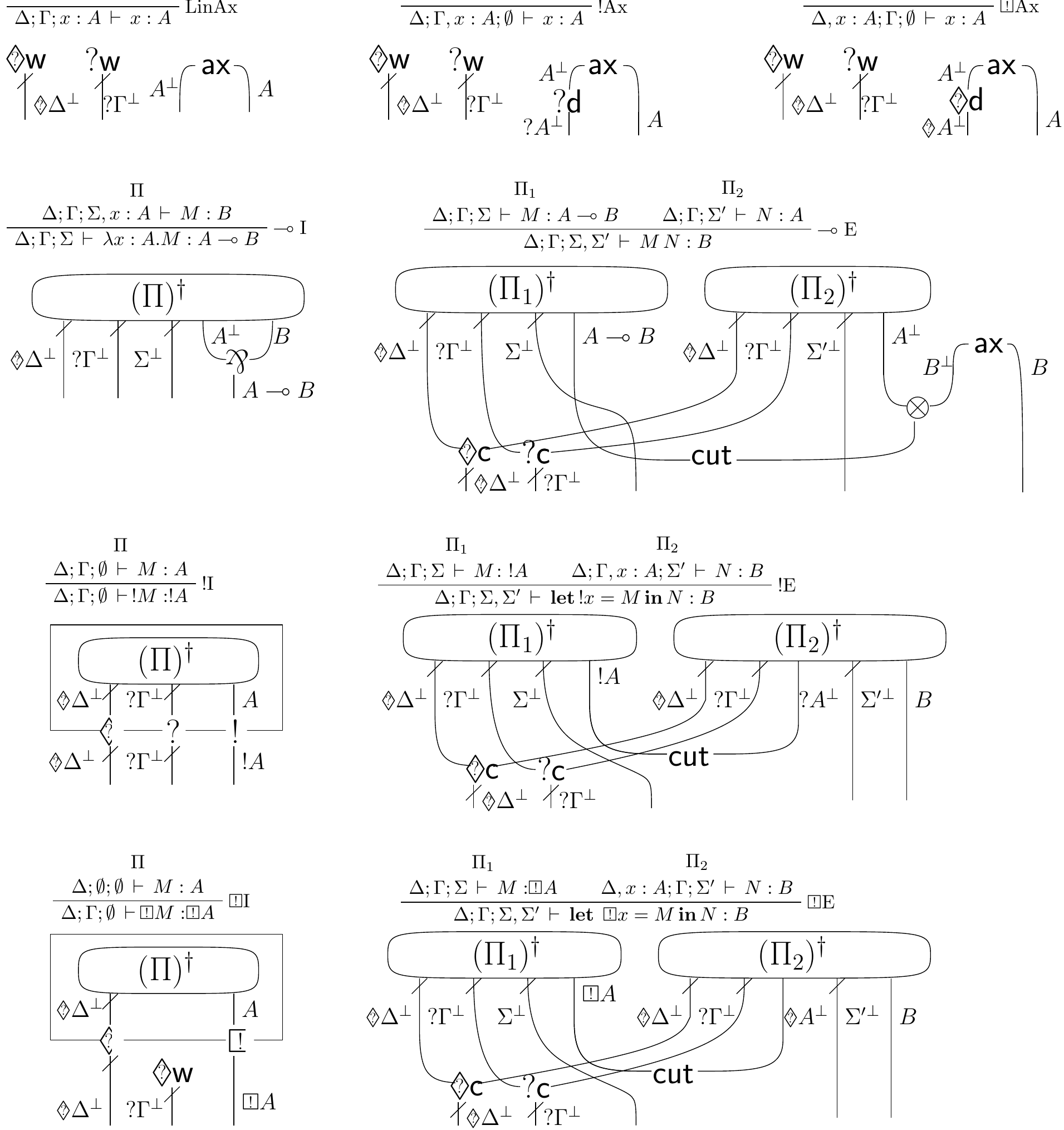}
  \caption{Translation from $\lambdabangbox$ to $\cmellbangbox$ proof nets.}
  \label{fig:bangboxnettranslation}
\end{figure}

\begin{lemma}
  Let $\net$ be a closed proof net and $\net'$ be its normal form.
  Then $\fint{\net} = \fint{\net'}$. 
\end{lemma}

The lemma is proved through two auxiliary lemmata below.

\begin{lemma}
  Let $\phi$ be a path from a conclusion of a closed net $\net$ ending at a node $a$.
  Let $(\mcont', \bcont', \ncont') = \fint{\phi}(\mcont, \varepsilon, \varepsilon)$.
  The height of $\bcont'$ (resp.\ $\ncont'$) matches with the
  number of exponential (resp.\ necessitation) boxes containing the node $a$.
\end{lemma}

\begin{proof}
  By spectating the rules of actions above:
  the height of stacks only changes at doors of a box.
\end{proof}

\begin{lemma}
  Let $\phi$ be a path inside a box of a closed net $\net$.
  $\fint{\phi}(\mcont, \sigma . \bcont, \tau . \ncont)$
  is in the form $(\mcont', \sigma' . \bcont, \tau . \ncont)$.
\end{lemma}

\begin{proof}
  Again, by spectating the rules of actions.
\end{proof}

\begin{theorem}
  If a closed term $M$ in $\lambdabangbox$ is typable and
  $\ottnt{M}  \leadsto  \ottnt{M'}$,
  then $\fint{\flambdanettrans{M}} = \fint{\flambdanettrans{M'}}$.
\end{theorem}

\begin{remark}
This notion of context semantics inherently captures the ``dynamics'' of computation, and indeed
Mackie exploited\,\cite{M:goim, M:goim_popl} the character to implement a compiler, in the level of machine code, for PCF.
In this paper we do not cover such a concrete compiler,
but the definition of $\fint{-}$ can be seen as
``context transformers'' of virtual machine that
is mathematically rigorous enough to model the computation of $\lambdabangbox$ (and hence of $\lambdabox$).
\end{remark}

\section{Related work}
\label{sec:related_work}

\subsection{Linear-logical reconstruction of modal logic}
The work on translations from modal logic to linear logic goes back to Martini and Masini\,\cite{MM:modal_view}.
They proposed a translation from classical S4 ($\csf$) to full propositional linear logic
by means of the Grisin--Ono translation.
However, their work only discusses provability.

The most similar work to ours is a ``linear analysis'' of $\csf$ by Schellinx\,\cite{S:linear_approach}, in which Girard translation from $\csf$ with respect to proofs is proposed.
He uses a bi-colored linear logic, a subsystem of multicolored linear logic by Danos et al.\,\cite{DJS:structure_of_exponentials},
called $\bill$, for a target calculus of the translation.
It has two pairs of exponentials
$\fbracket{\bangzero, \whynotzero}$ and $\fbracket{\bangone, \whynotone}$,
called \emph{subexponentials} following the terminology by Nigam and Miller\,\cite{NM:subexponentials}, with the following rules:

\begin{tabular}{c}
  \begin{minipage}{0.45\hsize}
    \begin{prooftree}
        \AXC{$ \bangone \Gamma, \bangzero \Gamma' \vdash A, \whynotone \Delta, \whynotzero \Delta' $}
      \RightLabel{ $\bangzero\mathrm{R}$ } 
      \UIC{$ \bangone \Gamma, \bangzero \Gamma' \vdash \bangzero A, \whynotone \Delta, \whynotzero \Delta' $}
    \end{prooftree}
  \end{minipage}

  \begin{minipage}{0.45\hsize}
    \begin{prooftree}
        \AXC{$ \bangone \Gamma \vdash A, \whynotone \Delta  $}
      \RightLabel{ $\bangone\mathrm{R}$ } 
      \UIC{$ \bangone \Gamma \vdash \bangone A, \whynotone \Delta  $}
    \end{prooftree}
  \end{minipage}
\end{tabular}

\noindent
These rules have, while they are defined in the classical setting, 
essentially the same structure to what we defined as $\rbangr$ and $\rbangboxr$ for $\imellbangbox$, respectively.

To mention the difference between the results of Schellinx and ours,
his work has investigated only in terms of proof theory. Neither a typed $\lambda$-calculus nor a Geometry of Interaction interpretation was given.
However, even so, he already gave a reduction-preserving Girard translation for the sequent calculi of $\csf$ and $\bill$, 
and his \emph{linear decoration}\,(cf. \cite{S:linear_approach, DJS:structure_of_exponentials}) allows us to obtain the cut-elimination theorem for $\csf$ as a corollary of that of $\bill$.
Thus, it should be interesting to investigate a relationship between his work and ours.

Furthermore, there also exist two uniform logical frameworks
that can encode various logics including $\isf$ and $\csf$.
One is the work by Nigam et al.\,\cite{N+:extended_framework} which based on Nigam and Miller's linear-logical framework with subexponentials
and on the notion of focusing by Andreoli.
The other work is \emph{adjoint logic} by Pruiksma et al.\,\cite{P+:adjoint_logic} which based on, again, subexponentials, and
the so-called LNL model for intuitionistic linear logic by Benton.
While our present work is still far from the two works, it seems fruitful to take our discussion into their frameworks 
to give linear-logical computational interpretations for various logics.

\subsection{Computation of modal logic and its relation to metaprogramming}
Computational interpretations of modal logic have been considered not only for intuitionistic S4
but also for various logics,
including the modal logics $\mathrm{K}$, $\mathrm{T}$, $\mathrm{K4}$, and $\mathrm{GL}$,
and a few constructive temporal logics\,(cf. the survey by Kavvos in \cite{K:dual-context_lics}).
This field of modal logics is known to be connected
to ``metaprogramming''
in the theory of programming languages and has been substantially studied. 
One of the studies is \emph{(multi-)staged computation}\,(cf.\,\cite{TS:metaml}), 
which is a programming paradigm that supports Lisp-like \emph{quasi-quote}, \emph{unquote}, and \emph{eval}.
The work of $\lambdabox$ by Davies and Pfenning\,\cite{DP:modal_analysis} is actually one of logical investigations of it.

Furthermore, the multi-stage programming is not a mere theory but has ``real'' implementations such as MetaML\,\cite{TS:metaml} and MetaOCaml\,(cf. a survey in \cite{C+:inference_for_classifiers})
in the style of functional programming languages. Some core calculi of these implementations are formalized as type systems\,(e.g.\,\cite{T:environment_classifiers, C+:inference_for_classifiers})
and investigated from the logical point of view (e.g.\,\cite{TI:logical_foundation}).

\section{Conclusion}
\label{sec:conclusion}
We have presented a linear-logical reconstruction of the intuitionistic modal logic S4,
by establishing the modal linear logic with the $\bangbox~$-modality and the S4-version of Girard translation from $\isf$.
The translation from $\isf$ to the modal linear logic is shown to be correct with respect to the level of proofs,
through the Curry--Howard correspondence.

While the proof-level Girard translation for modal logic is already proposed by Schellinx,
our typed $\lambda$-calculus $\lambdabangbox$ and its Geometry of Interaction Machine (GoIM) are novel.
Also, the significance of our formalization
is its simplicity.
All we need to establish the linear-logical reconstruction of modal logic is the $\bangbox~$-modality,
an integration of $!$-modality and $\Box$-modality,
that gives the structure of modal logic into linear logic.
Thanks to the simplicity, our $\lambda$-calculus and the GoIM can be obtained as simple extensions of existing works.

As a further direction, we plan to enrich our framework to cover other modal logics such as $\mathrm{K}$, $\mathrm{T}$, and $\mathrm{K4}$,
following the work of contextual modal calculi by Kavvos\,\cite{K:dual-context_lics}.
Moreover, reinvestigating of the modal-logical foundation for multi-stage programming by Tsukada and Igarashi\,\cite{TI:logical_foundation}
via our methods and extending Mackie's GoIM for PCF\,\cite{M:goim_popl} to the modal-logical setting
seem to be interesting from the viewpoint of programming languages.

Lastly,  we have also left a semantical study for modal linear logic with respect to the validity.
At the present stage, we think that we could give a sound-and-complete characterization of modal linear logic by an integration of Kripke semantics of modal logic and phase semantics of linear logic,
but details will be studied in a future paper.

\bibliography{ref}

\appendix
\section{Appendix}

\subsection{Cut-elimination theorem for the intuitionistic modal linear logic}

In this section, we give a complete proof of the cut-elimination theorem of $\imellbangbox$.

\begin{theorem}[Cut-elimination]
  The rule $\rcut$ in $\imellbangbox$ is admissible, i.e.,
  if $\Gamma \, \vdash \, \ottnt{A}$ is derivable, then there is a derivation of the same judgment without any applications of $\rcut$.
\end{theorem}

\begin{proof}
  As we mentioned in the body,
  we will show the following rules are admissible.

  \begin{tabular}{c}
    \begin{minipage}{0.45\hsize}
      \begin{prooftree}
          \AXC{$\Gamma \, \vdash \, \ottsym{!}  \ottnt{A}$}
          \AXC{$\Gamma'  \ottsym{,}   ( \ottsym{!}  \ottnt{A} )^{n}  \, \vdash \, \ottnt{B}$}
        \RightLabel{$ \rbcut $}
        \BIC{$\Gamma  \ottsym{,}  \Gamma' \, \vdash \, \ottnt{B}$}
      \end{prooftree}
    \end{minipage}

    \begin{minipage}{0.45\hsize}
      \begin{prooftree}
          \AXC{$\Gamma \, \vdash \,  \bangbox \ottnt{A} $}
          \AXC{$\Gamma'  \ottsym{,}   (  \bangbox \ottnt{A}  )^{n}  \, \vdash \, \ottnt{B}$}
        \RightLabel{$ \rbbcut $}
        \BIC{$\Gamma  \ottsym{,}  \Gamma' \, \vdash \, \ottnt{B}$}
      \end{prooftree}
    \end{minipage}
  \end{tabular}

  \noindent
  The admissibility of $\rcut$, $\rbcut$, $\rbbcut$ are shown by
  simultaneous induction on the derivation of $\Gamma \, \vdash \, \ottnt{A}$
  with the lexicographic complexity $\fbracket{\delta, h}$,
  where $\delta$ is the degree of the assumed derivation and $h$ is its height.
  Therefore, it is enough to show that for every application of cuts, one of the following hold: (1) it can be reduced to a cut with a smaller cut-degree; (2) it can be reduced to a cut with a smaller height; (3) it can be eliminated immediately.

  In what follows, we will explain the admissibility of each cut rule separately although the actual proofs are done simultaneously.
  \begin{itemize}
    \item The admissibility of $\rcut$.
      We show that every application of the rule $\rcut$ whose cut-degree is maximal is eliminable.
      Thus, consider an application of $\rcut$ in the derivation:
      \begin{prooftree}
            \AXC{$ \Pi_0 $}
          \noLine
          \UIC{$\Gamma \, \vdash \, \ottnt{A}$}
            \AXC{$ \Pi_1 $}
          \noLine
          \UIC{$\Gamma'  \ottsym{,}  \ottnt{A} \, \vdash \, \ottnt{B}$}
        \RightLabel{$ \rcut $}
        \BIC{$\Gamma  \ottsym{,}  \Gamma' \, \vdash \, \ottnt{B}$}
      \end{prooftree}
      such that its cut-degree is maximal and its height is minimal (comparing to the other applications whose cut-degree is maximal).
      The proof proceeds by case analysis on $\Pi_0$.
      \begin{itemize}
        \item $\Pi_0$ ends with $\rcut$.
          In this case the derivation is as follows:
          \begin{prooftree}
                  \AXC{$ \vdots $}
                \noLine
                \UIC{$\Gamma_{{\mathrm{0}}} \, \vdash \, \ottnt{C}$}
                  \AXC{$ \vdots $}
                \noLine
                \UIC{$\Gamma_{{\mathrm{1}}}  \ottsym{,}  \ottnt{C} \, \vdash \, \ottnt{A}$}
              \RightLabel{$ \rcut $}
              \BIC{$\Gamma_{{\mathrm{0}}}  \ottsym{,}  \Gamma_{{\mathrm{1}}} \, \vdash \, \ottnt{A}$}
                \AXC{$ \Pi_1 $}
              \noLine
              \UIC{$\Gamma'  \ottsym{,}  \ottnt{A} \, \vdash \, \ottnt{B}$}
            \RightLabel{$ \rcut $}
            \BIC{$\Gamma_{{\mathrm{0}}}  \ottsym{,}  \Gamma_{{\mathrm{1}}}  \ottsym{,}  \Gamma' \, \vdash \, \ottnt{B}$}
          \end{prooftree}
          Since the bottom application of $\rcut$ was chosen to have the maximal cut-degree and the minimum height,
          the cut-degree of the above is less than that of the bottom.
          Therefore, the derivation can be translated to the following:
          \begin{prooftree}
                \AXC{$ \vdots $}
              \noLine
              \UIC{$\Gamma_{{\mathrm{0}}} \, \vdash \, \ottnt{C}$}
                \AXC{$ \vdots $}
                \noLine
                \UIC{$\Gamma_{{\mathrm{1}}}  \ottsym{,}  \ottnt{C} \, \vdash \, \ottnt{A}$}
                  \AXC{$ \Pi_1 $}
                \noLine
                \UIC{$\Gamma'  \ottsym{,}  \ottnt{A} \, \vdash \, \ottnt{B}$}
              \doubleLine
              \RightLabel{ I.H. }
              \BIC{$\Gamma_{{\mathrm{1}}}  \ottsym{,}  \ottnt{C}  \ottsym{,}  \Gamma' \, \vdash \, \ottnt{B}$}
            \doubleLine
            \RightLabel{ I.H. }
            \BIC{$\Gamma_{{\mathrm{0}}}  \ottsym{,}  \Gamma_{{\mathrm{1}}}  \ottsym{,}  \Gamma' \, \vdash \, \ottnt{B}$}
          \end{prooftree}
        \item $\Pi_0$ ends with $\rlimpr$.
          In this case, the derivation is as follows:
          \begin{prooftree}
                  \AXC{$ \vdots $}
                \noLine
                \UIC{$\Gamma  \ottsym{,}  \ottnt{A_{{\mathrm{0}}}} \, \vdash \, \ottnt{A_{{\mathrm{1}}}}$}
              \RightLabel{$ \rlimpr $}
              \UIC{$\Gamma \, \vdash \, \ottnt{A_{{\mathrm{0}}}}  \multimap  \ottnt{A_{{\mathrm{1}}}}$}
                \AXC{$ \Pi_1 $}
              \noLine
              \UIC{$\Gamma'  \ottsym{,}  \ottnt{A_{{\mathrm{0}}}}  \multimap  \ottnt{A_{{\mathrm{1}}}} \, \vdash \, \ottnt{B}$}
            \RightLabel{$ \rcut $}
            \BIC{$\Gamma  \ottsym{,}  \Gamma' \, \vdash \, \ottnt{B}$}
          \end{prooftree}
          for some $A_0$ and $A_1$ such that $A \equiv \ottnt{A_{{\mathrm{0}}}}  \multimap  \ottnt{A_{{\mathrm{1}}}}$.
          If the last step in $\Pi_1$ is $\rax$, then the result is obtained as $\Pi_0$.
          If the last step in $\Pi_1$ is $\rlimpl$, the derivation is as follows:
          \begin{prooftree}
                  \AXC{$ \vdots $}
                \noLine
                \UIC{$\Gamma  \ottsym{,}  \ottnt{A_{{\mathrm{0}}}} \, \vdash \, \ottnt{A_{{\mathrm{1}}}}$}
              \RightLabel{$ \rlimpr $}
              \UIC{$\Gamma \, \vdash \, \ottnt{A_{{\mathrm{0}}}}  \multimap  \ottnt{A_{{\mathrm{1}}}}$}
                  \AXC{$ \vdots $}
                \noLine
                \UIC{$\Gamma' \, \vdash \, \ottnt{A_{{\mathrm{0}}}}$}
                  \AXC{$ \vdots $}
                \noLine
                \UIC{$\Gamma''  \ottsym{,}  \ottnt{A_{{\mathrm{1}}}} \, \vdash \, \ottnt{B}$}
              \RightLabel{$ \rlimpl $}
              \BIC{$\Gamma'  \ottsym{,}  \Gamma''  \ottsym{,}  \ottnt{A_{{\mathrm{0}}}}  \multimap  \ottnt{A_{{\mathrm{1}}}} \, \vdash \, \ottnt{B}$}
            \RightLabel{$ \rcut $}
            \BIC{$\Gamma  \ottsym{,}  \Gamma'  \ottsym{,}  \Gamma'' \, \vdash \, \ottnt{B}$}
          \end{prooftree}
          which is translated to the following:
          \begin{prooftree}
                  \AXC{$ \vdots $}
                \noLine
                \UIC{$\Gamma' \, \vdash \, \ottnt{A_{{\mathrm{0}}}}$}
                  \AXC{$ \vdots $}
                \noLine
                \UIC{$\Gamma  \ottsym{,}  \ottnt{A_{{\mathrm{0}}}} \, \vdash \, \ottnt{A_{{\mathrm{1}}}}$}
              \doubleLine
              \RightLabel{ I.H. }
              \BIC{$\Gamma  \ottsym{,}  \Gamma' \, \vdash \, \ottnt{A_{{\mathrm{1}}}}$}
                \AXC{$ \vdots $}
              \noLine
              \UIC{$\Gamma''  \ottsym{,}  \ottnt{A_{{\mathrm{1}}}} \, \vdash \, \ottnt{B}$}
            \doubleLine
            \RightLabel{ I.H. }
            \BIC{$\Gamma  \ottsym{,}  \Gamma'  \ottsym{,}  \Gamma'' \, \vdash \, \ottnt{B}$}
          \end{prooftree}
          since the cut-degrees of $A_0$ and $A_1$ are less than that of $A$.
          The other cases can be shown by simple commutative conversions.
        \item $\Pi_0$ ends with $\rbangr$.
          This case is dealt as a special case of the case $\rbangr$ in $\rbcut$.
        \item $\Pi_0$ ends with $\rbangboxr$.
          This case is dealt as a special case of the case $\rbangboxr$ in $\rbbcut$.
        \item $\Pi_0$ ends with the other rules. Easy.
      \end{itemize}
    \item The admissibility of $\rbcut$.
      As in the case of $\rcut$, consider an application of $\rbcut$:
      \begin{prooftree}
            \AXC{$ \Pi_0 $}
          \noLine
          \UIC{$\Gamma \, \vdash \, \ottsym{!}  \ottnt{A}$}
            \AXC{$ \Pi_1 $}
          \noLine
          \UIC{$\Gamma'  \ottsym{,}   ( \ottsym{!}  \ottnt{A} )^{n}  \, \vdash \, \ottnt{B}$}
        \RightLabel{$ \rbcut $}
        \BIC{$\Gamma  \ottsym{,}  \Gamma' \, \vdash \, \ottnt{B}$}
      \end{prooftree}
      such that its cut-degree is maximal and its height is minimal.
      By case analysis on $\Pi_0$.
      \begin{itemize}
        \item $\Pi_0$ ends with $\rax$.
          In this case the cut-elimination is done as follows:
    
          \begin{tabular}{c}
            \begin{minipage}{0.45\hsize}
            \begin{prooftree}
                  \AXC{$ $}
                \UIC{$\ottsym{!}  \ottnt{A} \, \vdash \, \ottsym{!}  \ottnt{A}$}
                  \AXC{$ \Pi_1 $}
                \noLine
                \UIC{$\Gamma'  \ottsym{,}   ( \ottsym{!}  \ottnt{A} )^{n}  \, \vdash \, \ottnt{B}$}
              \RightLabel{$ \rbcut $}
              \BIC{$\ottsym{!}  \ottnt{A}  \ottsym{,}  \Gamma' \, \vdash \, \ottnt{B}$}
            \end{prooftree}
            \end{minipage}

            \begin{minipage}{0.09\hsize}
              $\overset{\text{Cut elim.}}{\Longrightarrow}$
            \end{minipage}

            \begin{minipage}{0.30\hsize}
            \begin{prooftree}
                  \AXC{$ \Pi_1 $}
                \noLine
                \UIC{$\Gamma'  \ottsym{,}   ( \ottsym{!}  \ottnt{A} )^{n}  \, \vdash \, \ottnt{B}$}
              \doubleLine
              \RightLabel{$ \rbcont $}
              \UIC{$\Gamma'  \ottsym{,}  \ottsym{!}  \ottnt{A} \, \vdash \, \ottnt{B}$}
            \end{prooftree}
            \end{minipage}
          \end{tabular}
        \item $\Pi_0$ ends with $\rbangr$.
          In this case the derivation is as follows:
          \begin{prooftree}
                  \AXC{$ \vdots $}
                \noLine
                \UIC{$ \bangbox \Gamma_{{\mathrm{0}}}   \ottsym{,}  \ottsym{!}  \Gamma_{{\mathrm{1}}} \, \vdash \, \ottnt{A}$}
              \RightLabel{$ \rbangr $}
              \UIC{$ \bangbox \Gamma_{{\mathrm{0}}}   \ottsym{,}  \ottsym{!}  \Gamma_{{\mathrm{1}}} \, \vdash \, \ottsym{!}  \ottnt{A}$}
                \AXC{$ \Pi_1 $}
              \noLine
              \UIC{$\Gamma'  \ottsym{,}   ( \ottsym{!}  \ottnt{A} )^{n}  \, \vdash \, \ottnt{B}$}
            \RightLabel{$ \rbcut $}
            \BIC{$ \bangbox \Gamma_{{\mathrm{0}}}   \ottsym{,}  \ottsym{!}  \Gamma_{{\mathrm{1}}}  \ottsym{,}  \Gamma' \, \vdash \, \ottnt{B}$} 
          \end{prooftree}
          Due to the side-condition of $\rbangr$, we have to do case analysis on $\Pi_1$ further as follows.
          \begin{itemize}
            \item $\Pi_1$ ends with $\rcut$.
              In this case the derivation is as follows:
              \begin{prooftree}
                    \AXC{$ \Pi_0 $}
                  \noLine
                  \UIC{$ \bangbox \Gamma_{{\mathrm{0}}}   \ottsym{,}  \ottsym{!}  \Gamma_{{\mathrm{1}}} \, \vdash \, \ottsym{!}  \ottnt{A}$}
                      \AXC{$ \vdots $}
                    \noLine
                    \UIC{$\Gamma'  \ottsym{,}   ( \ottsym{!}  \ottnt{A} )^{k}  \, \vdash \, \ottnt{C}$}
                      \AXC{$ \vdots $}
                    \noLine
                    \UIC{$\Gamma''  \ottsym{,}   ( \ottsym{!}  \ottnt{A} )^{l}   \ottsym{,}  \ottnt{C} \, \vdash \, \ottnt{B}$}
                  \RightLabel{$ \rcut $}
                  \BIC{$\Gamma'  \ottsym{,}  \Gamma''  \ottsym{,}   ( \ottsym{!}  \ottnt{A} )^{n}  \, \vdash \, \ottnt{B}$}
                \RightLabel{$ \rbcut $}
                \BIC{$ \bangbox \Gamma_{{\mathrm{0}}}   \ottsym{,}  \ottsym{!}  \Gamma_{{\mathrm{1}}}  \ottsym{,}  \Gamma'  \ottsym{,}  \Gamma'' \, \vdash \, \ottnt{B}$}
              \end{prooftree}
              where $n = k + l$.
              We only deal with the case of $k > 0$ and $l > 0$, and the other cases are easy.
              Then, the derivation can be translated to the following:
              \begin{prooftree}
                \def\defaultHypSeparation{\hskip .1in}
                        \AXC{$ \Pi_0 $}
                      \noLine
                      \UIC{$ \bangbox \Gamma_{{\mathrm{0}}}   \ottsym{,}  \ottsym{!}  \Gamma_{{\mathrm{1}}} \, \vdash \, \ottsym{!}  \ottnt{A}$}
                        \AXC{$ \vdots $}
                      \noLine
                      \UIC{$\Gamma'  \ottsym{,}   ( \ottsym{!}  \ottnt{A} )^{k}  \, \vdash \, \ottnt{C}$}
                    \doubleLine
                    \RightLabel{ I.H. }
                    \BIC{$ \bangbox \Gamma_{{\mathrm{0}}}   \ottsym{,}  \ottsym{!}  \Gamma_{{\mathrm{1}}}  \ottsym{,}  \Gamma' \, \vdash \, \ottnt{C}$}
                        \AXC{$ \Pi_0 $}
                      \noLine
                      \UIC{$ \bangbox \Gamma_{{\mathrm{0}}}   \ottsym{,}  \ottsym{!}  \Gamma_{{\mathrm{1}}} \, \vdash \, \ottsym{!}  \ottnt{A}$}
                        \AXC{$ \vdots $}
                      \noLine
                      \UIC{$\Gamma''  \ottsym{,}   ( \ottsym{!}  \ottnt{A} )^{l}   \ottsym{,}  \ottnt{C} \, \vdash \, \ottnt{B}$}
                    \doubleLine
                    \RightLabel{ I.H. }
                    \BIC{$ \bangbox \Gamma_{{\mathrm{0}}}   \ottsym{,}  \ottsym{!}  \Gamma_{{\mathrm{1}}}  \ottsym{,}  \Gamma''  \ottsym{,}  \ottnt{C} \, \vdash \, \ottnt{B}$}
                  \doubleLine
                  \RightLabel{ I.H. }
                  \BIC{$ (  \bangbox \Gamma_{{\mathrm{0}}}  )^2   \ottsym{,}   ( \ottsym{!}  \Gamma_{{\mathrm{1}}} )^2   \ottsym{,}  \Gamma'  \ottsym{,}  \Gamma'' \, \vdash \, \ottnt{B}$}
                \doubleLine
                \RightLabel{$ \rbcont, \rbbcont $}
                \UIC{$ \bangbox \Gamma_{{\mathrm{0}}}   \ottsym{,}  \ottsym{!}  \Gamma_{{\mathrm{1}}}  \ottsym{,}  \Gamma'  \ottsym{,}  \Gamma'' \, \vdash \, \ottnt{B}$}
              \end{prooftree}
              since the cut-degree of $\rbcut$ is less than that of $\rcut$ from the assumption.
            \item $\Pi_1$ ends with $\rbangl$.
              If the formula introduced by $\rbangl$ is not the cut-formula, then it is easy.
              For the other case, the derivation is as follows:
              \begin{prooftree}
                      \AXC{$ \vdots $}
                    \noLine
                    \UIC{$ \bangbox \Gamma_{{\mathrm{0}}}   \ottsym{,}  \ottsym{!}  \Gamma_{{\mathrm{1}}} \, \vdash \, \ottnt{A}$}
                  \RightLabel{$ \rbangr $}
                  \UIC{$ \bangbox \Gamma_{{\mathrm{0}}}   \ottsym{,}  \ottsym{!}  \Gamma_{{\mathrm{1}}} \, \vdash \, \ottsym{!}  \ottnt{A}$}
                      \AXC{$ \vdots $}
                    \noLine
                    \UIC{$\Gamma'  \ottsym{,}   ( \ottsym{!}  \ottnt{A} )^{n-1}   \ottsym{,}  \ottnt{A} \, \vdash \, \ottnt{B}$}
                  \RightLabel{$ \rbangl $}
                  \UIC{$\Gamma'  \ottsym{,}   ( \ottsym{!}  \ottnt{A} )^{n}  \, \vdash \, \ottnt{B}$}
                \RightLabel{$ \rbcut $}
                \BIC{$ \bangbox \Gamma_{{\mathrm{0}}}   \ottsym{,}  \ottsym{!}  \Gamma_{{\mathrm{1}}}  \ottsym{,}  \Gamma' \, \vdash \, \ottnt{B}$}
              \end{prooftree}
              which is translated to the following:
              \begin{prooftree}
                      \AXC{$ \vdots $}
                    \noLine
                    \UIC{$ \bangbox \Gamma_{{\mathrm{0}}}   \ottsym{,}  \ottsym{!}  \Gamma_{{\mathrm{1}}} \, \vdash \, \ottnt{A}$}
                        \AXC{$ \Pi_0 $}
                      \noLine
                      \UIC{$ \bangbox \Gamma_{{\mathrm{0}}}   \ottsym{,}  \ottsym{!}  \Gamma_{{\mathrm{1}}} \, \vdash \, \ottsym{!}  \ottnt{A}$}
                        \AXC{$ \vdots $}
                      \noLine
                      \UIC{$\Gamma'  \ottsym{,}   ( \ottsym{!}  \ottnt{A} )^{n-1}   \ottsym{,}  \ottnt{A} \, \vdash \, \ottnt{B}$}
                    \doubleLine
                    \RightLabel{ I.H. }
                    \BIC{$ \bangbox \Gamma_{{\mathrm{0}}}   \ottsym{,}  \ottsym{!}  \Gamma_{{\mathrm{1}}}  \ottsym{,}  \Gamma'  \ottsym{,}  \ottnt{A} \, \vdash \, \ottnt{B}$}
                  \doubleLine
                  \RightLabel{ I.H. }
                  \BIC{$ (  \bangbox \Gamma_{{\mathrm{0}}}  )^2   \ottsym{,}   ( \ottsym{!}  \Gamma_{{\mathrm{1}}} )^2   \ottsym{,}  \Gamma' \, \vdash \, \ottnt{B}$}
                \doubleLine
                \RightLabel{$ \rbcont, \rbbcont $}
                \UIC{$ \bangbox \Gamma_{{\mathrm{0}}}   \ottsym{,}  \ottsym{!}  \Gamma_{{\mathrm{1}}}  \ottsym{,}  \Gamma' \, \vdash \, \ottnt{B}$}
              \end{prooftree}
            \item $\Pi_1$ ends with $\rbcont$.
              If the formula introduced by $\rbcont$ is not the cut-formula, then it is easy.
              For the other case, the cut-elimination is done as follows: 

              \begin{tabular}{c}
                \begin{minipage}{0.38\hsize}
                  \begin{prooftree}
                    \def\ScoreOverhang{2pt}
                    \def\defaultHypSeparation{\hskip 0.05in}
                        \AXC{$ \Pi_0 $}
                      \noLine
                      \UIC{$ \bangbox \Gamma_{{\mathrm{0}}}   \ottsym{,}  \ottsym{!}  \Gamma_{{\mathrm{1}}} \, \vdash \, \ottsym{!}  \ottnt{A}$}
                          \AXC{$ \vdots $}
                        \noLine
                        \UIC{$\Gamma'  \ottsym{,}   ( \ottsym{!}  \ottnt{A} )^{n+1}  \, \vdash \, \ottnt{B}$}
                      \RightLabel{$ \rbcont $}
                      \UIC{$\Gamma'  \ottsym{,}   ( \ottsym{!}  \ottnt{A} )^{n}  \, \vdash \, \ottnt{B}$}
                    \RightLabel{$ \rbcut $}
                    \BIC{$ \bangbox \Gamma_{{\mathrm{0}}}   \ottsym{,}  \ottsym{!}  \Gamma_{{\mathrm{1}}}  \ottsym{,}  \Gamma' \, \vdash \, \ottnt{B}$}
                  \end{prooftree}
                \end{minipage}

                \begin{minipage}{0.08\hsize}
                  $\overset{\text{Cut elim.}}{\Longrightarrow}$
                \end{minipage}

                \begin{minipage}{0.38\hsize}
                \begin{prooftree}
                  \def\ScoreOverhang{2pt}
                  \def\defaultHypSeparation{\hskip 0.05in}
                      \AXC{$ \Pi_0 $}
                    \noLine
                    \UIC{$ \bangbox \Gamma_{{\mathrm{0}}}   \ottsym{,}  \ottsym{!}  \Gamma_{{\mathrm{1}}} \, \vdash \, \ottsym{!}  \ottnt{A}$}
                      \AXC{$ \vdots $}
                    \noLine
                    \UIC{$\Gamma'  \ottsym{,}   ( \ottsym{!}  \ottnt{A} )^{n+1}  \, \vdash \, \ottnt{B}$}
                  \RightLabel{ I.H. }
                  \doubleLine
                  \BIC{$ \bangbox \Gamma_{{\mathrm{0}}}   \ottsym{,}  \ottsym{!}  \Gamma_{{\mathrm{1}}}  \ottsym{,}  \Gamma' \, \vdash \, \ottnt{B}$}
                \end{prooftree}
                \end{minipage}
              \end{tabular}

              Note that the whole proof has not been proceeding by induction on $n$, and hence the number of occurrences of $!A$ does not matter in this case.
            \item $\Pi_1$ ends with the other rules. Easy.
          \end{itemize}
        \item $\Pi_0$ ends with the other rules. Easy.
      \end{itemize}
    \item The admissibility of $\rbbcut$. Similar to the case of $\rbcut$. \qedhere
  \end{itemize}
\end{proof}

\subsection{Strong normalizability of the typed $\lambda$-calculus for modal linear logic}

We complete the proof of the strong normalization theorem for $\lambdabangbox$.
As we mentioned, this is done by an embedding to 
a typed $\lambda$-calculus
for the $(!, \slimp)$-fragment of dual intuitionistic linear logic,
studied by Ohta and Hasegawa\,\cite{OH:terminating_linear_lambda}, and shown to be strongly normalizing.

\begin{rulefigure}{fig:lambdabanglimp}{Definition of $\lambdabanglimp$ (some syntax are changed to fit the present paper's notation).}
  \input{definitions/lambdabanglimp.tex}

\end{rulefigure}

The calculus of Ohta and Hasegawa, named 
$\lambdabanglimp$ here, is given in Figure~\ref{fig:lambdabanglimp}.
The syntax and the typing rules can be read in the same way as (the ($!$, $\slimp$)-fragment of) $\lambdabangbox$.
There are somewhat many reduction rules in contrast to those of $\lambdabangbox$, but these are due to the purpose of Ohta and Hasegawa to consider $\eta$-rules and commutative conversions.
The different sets of reduction rules do not cause any problems to prove the strong normalizability of $\lambdabangbox$.

\begin{rulefigure}{fig:embedding}{Definition of the embeddings $ \fembed{ \ottnt{A} } $, $ \fembed{ \Gamma } $, and $ \fembed{ \ottnt{M} } $.}
  \input{definitions/embedding.tex}
\end{rulefigure}

\begin{definition}[Embedding]
  An \emph{embedding} from $\lambdabangbox$ to $\lambdabanglimp$
  is defined to be the triple of the translations $ \fembed{ \ottnt{A} } $, $ \fembed{ \Gamma } $, and $ \fembed{ \ottnt{M} } $
  given in Figure~\ref{fig:embedding}.
\end{definition}

\begin{lemma}[Preservation of typing and reduction]\ 
  \label{lem:preservation}
  \begin{enumerate}
    \item If $\Delta  \ottsym{;}  \Gamma  \ottsym{;}  \Sigma \, \vdash \, \ottnt{M}  \ottsym{:}  \ottnt{A}$ in $\lambdabangbox$, then $ \fembed{ \Delta  \ottsym{,}  \Gamma }   \ottsym{;}   \fembed{ \Sigma }  \, \vdash \,  \fembed{ \ottnt{M} }   \ottsym{:}   \fembed{ \ottnt{A} } $ in $\lambdabanglimp$.
    \item If $\ottnt{M}  \leadsto  \ottnt{N}$ in $\lambdabangbox$, then $ \fembed{ \ottnt{M} }   \leadsto   \fembed{ \ottnt{N} } $ in $\lambdabanglimp$.
  \end{enumerate}
\end{lemma}

\begin{proof}
  By induction on $\Delta  \ottsym{;}  \Gamma  \ottsym{;}  \Sigma \, \vdash \, \ottnt{M}  \ottsym{:}  \ottnt{A}$ and $\ottnt{M}  \leadsto  \ottnt{N}$, respectively.
\end{proof}

\begin{theorem}[Strong normalization]
  In $\lambdabangbox$,
  there are no infinite reduction sequences starting from $M$ for all well-typed term $M$.
\end{theorem}

\begin{proof}
  Suppose that there exists an infinite reduction sequence starting from $M$ in $\lambdabangbox$.
  Then, the term $ \fembed{ \ottnt{M} } $ is well-typed in $\lambdabanglimp$ and yields an infinite reduction sequence in $\lambdabanglimp$
  by Lemma~\ref{lem:preservation}. However, this contradicts the strong normalizability of $\lambdabanglimp$.
\end{proof}

\subsection{Bracket abstraction algorithm}

We show the definition of bracket abstraction operators in this section.

\begin{rulefigure}{fig:bracket_abstraction}{Definitions of $\ottsym{(}   \lambda_{*}  \ottmv{x} .  \ottnt{M}   \ottsym{)}$, $\ottsym{(}   \lambda^{!}_{*}  \ottmv{x} .  \ottnt{M}   \ottsym{)}$, and $\ottsym{(}   \lambda^{\smallbangbox}_{*}  \ottmv{x} .  \ottnt{M}   \ottsym{)}$ for bracket abstraction.}
  \input{definitions/bracket_abstraction.tex}
\end{rulefigure}

\begin{definition}[Bracket abstraction]
  Let $M$ be a term $M$ of $\clvbangbox$ such that $\Delta  \ottsym{;}  \Gamma  \ottsym{;}  \Sigma \, \vdash \, \ottnt{M}  \ottsym{:}  \ottnt{A}$ and $x \in \fv{M}$ for some $\Delta, \Gamma, \Sigma, A$ and $x$.
  Then, the \emph{bracket abstraction of $M$ with respect to $x$} is defined to be either one of the following, depending the variable kind of $x$:
  \begin{alignat*}{3}
    \hspace{0.30\textwidth}&\ottsym{(}   \lambda_{*}  \ottmv{x} .  \ottnt{M}   \ottsym{)}    &&\hspace{3em}\text{if $x \in \fdom{\Sigma}$\footnotemark;}\\
    \hspace{0.30\textwidth}&\ottsym{(}   \lambda^{!}_{*}  \ottmv{x} .  \ottnt{M}   \ottsym{)}   &&\hspace{3em}\text{if $x \in \fdom{\Gamma}$;}\\
    \hspace{0.30\textwidth}&\ottsym{(}   \lambda^{\smallbangbox}_{*}  \ottmv{x} .  \ottnt{M}   \ottsym{)}  &&\hspace{3em}\text{if $x \in \fdom{\Delta}$,}
  \end{alignat*}
  where each one of $\ottsym{(}   \lambda_{*}  \ottmv{x} .  \ottnt{M}   \ottsym{)}$, $\ottsym{(}   \lambda^{!}_{*}  \ottmv{x} .  \ottnt{M}   \ottsym{)}$, and $\ottsym{(}   \lambda^{\smallbangbox}_{*}  \ottmv{x} .  \ottnt{M}   \ottsym{)}$
  is the meta-level \emph{bracket abstraction} operation given in Figure~\ref{fig:bracket_abstraction},
  which takes the pair of $x$ and $M$, and yields a $\clvbangbox$-term.
  \footnotetext{$\fdom{\Gamma}$ is defined to be the set $\{ x ~|~ (x : A) \in \Gamma \}$ for all type contexts $\Gamma$.}
\end{definition}

\begin{remark}
  As in the case of standard bracket abstraction algorithm,
  the intuition behind the operations $\ottsym{(}   \lambda_{*}  \ottmv{x} .  \ottnt{M}   \ottsym{)}$, $\ottsym{(}   \lambda^{!}_{*}  \ottmv{x} .  \ottnt{M}   \ottsym{)}$, and $\ottsym{(}   \lambda^{!}_{*}  \ottmv{x} .  \ottnt{M}   \ottsym{)}$
  is that they are defined so as to mimic the $\lambda$-abstraction operation in the framework of combinatory logic.
  For instance, the denotation of $\ottsym{(}   \lambda_{*}  \ottmv{x} .  \ottnt{M}   \ottsym{)}$ is a $\clvbangbox$-term that represents a function with the parameter $x$, that is, it is a term that satisfies that $ \ottsym{(}   \lambda_{*}  \ottmv{x} .  \ottnt{M}   \ottsym{)} \, \ottnt{N}   \leadsto^+   \ottnt{M}  [  \ottmv{x}  :=  \ottnt{N}  ] $ in $\clvbangbox$,
  for all $\clvbangbox$-terms $N$.
\end{remark}

\begin{remark}
  There are no definitions for some cases in $\ottsym{(}   \lambda_{*}  \ottmv{x} .  \ottnt{M}   \ottsym{)}$ and $\ottsym{(}   \lambda^{!}_{*}  \ottmv{x} .  \ottnt{M}   \ottsym{)}$,
  e.g, the case that $\ottsym{(}   \lambda_{*}  \ottmv{x} .  \ottsym{(}   \ottnt{M} \, \ottnt{N}   \ottsym{)}   \ottsym{)}$ such that $x \in \fv{M}$ and $x \in \fv{N}$,
  and the case that $\ottsym{(}   \lambda^{!}_{*}  \ottmv{x} .  \ottsym{(}   \bangbox  \ottnt{M}   \ottsym{)}   \ottsym{)}$. This is because that these are actually unnecessary due to the linearity condition and the side condition of the rule $\hbbnec$~.
  Moreover, the well-definedness of the bracket abstraction operations can be shown by induction on $M$,
  and in reality, the proof of the deduction theorem can be seen as what justifies it.
  The intentions that $ \ottsym{(}   \lambda_{*}  \ottmv{x} .  \ottnt{M}   \ottsym{)} \, \ottnt{N}   \leadsto^+   \ottnt{M}  [  \ottmv{x}  :=  \ottnt{N}  ] $, etc. can also be shown by easy calculation.
\end{remark}

\end{document}

%% file: ott_definition.tex
\newcommand{\ottnt}[1]{\mathit{#1}}
\newcommand{\ottmv}[1]{\mathit{#1}}
\newcommand{\ottkw}[1]{\mathbf{#1}}
\newcommand{\ottsym}[1]{#1}


%% file: definitions/imell.tex
  \hspace{-1em}
  \begin{tabular}{c|c}
    {
    \hspace{-1em}
    \begin{minipage}{0.3\hsize}
      \vspace{-0.2em}
      \paragraph*{Syntactic category}
        \vspace{-1em}
        \begin{align*}
          \hspace{-1em}\text{Formulae}~~A, B, C ::= p ~|~ \ottnt{A}  \multimap  \ottnt{B} ~|~ \ottsym{!}  \ottnt{A}
        \end{align*}
    \end{minipage}
    }
    &
    {
    \hspace{-1em}
    \begin{minipage}{0.5\hsize}
      \vspace{0.2em}
      \paragraph*{Inference rule}
        \vspace{-0.8em}
        \begin{tabular}{c}
          \begin{minipage}{0.18\hsize}
            \begin{prooftree}
            \def\ScoreOverhang{2pt}
                \AXC{$ $}
              \RightLabel{$ \rax $}
              \UIC{$\ottnt{A} \, \vdash \, \ottnt{A}$}
            \end{prooftree}
          \end{minipage}

          \begin{minipage}{0.63\hsize}
            \begin{prooftree}
            \def\ScoreOverhang{2pt}
            \def\defaultHypSeparation{\hskip 0.5em}
                \AXC{$\Gamma \, \vdash \, \ottnt{A}$}
                \AXC{$\Gamma'  \ottsym{,}  \ottnt{A} \, \vdash \, \ottnt{B}$}
              \RightLabel{$ \rcut $}
              \BIC{$\Gamma  \ottsym{,}  \Gamma' \, \vdash \, \ottnt{B}$}
            \end{prooftree}
          \end{minipage}

          \begin{minipage}{0.18\hsize}
            \begin{prooftree}
            \def\ScoreOverhang{2pt}
                \AXC{$\ottsym{!}  \Gamma \, \vdash \, \ottnt{A}$}
              \RightLabel{$ \rprom $}
              \UIC{$\ottsym{!}  \Gamma \, \vdash \, \ottsym{!}  \ottnt{A}$}
            \end{prooftree}
          \end{minipage}
        \end{tabular}
    \end{minipage}
    }
  \end{tabular}
 
  \hspace{-0.75em}\rule{0.445\textwidth}{0.4pt}

  \hspace{-1em}
  \begin{tabular}{c}
    {
    \hspace{-1em}
    \begin{minipage}{\hsize}
        \begin{center}
        \vspace{-0.3em}
        \begin{tabular}{c}
          \begin{minipage}{0.195\hsize}
            \begin{prooftree}
            \def\ScoreOverhang{2pt}
                \AXC{$\Gamma  \ottsym{,}  \ottnt{A} \, \vdash \, \ottnt{B}$}
              \RightLabel{$ \rlimpr $}
              \UIC{$\Gamma \, \vdash \, \ottnt{A}  \multimap  \ottnt{B}$}
            \end{prooftree}
          \end{minipage}

          \begin{minipage}{0.275\hsize}
            \begin{prooftree}
            \def\defaultHypSeparation{\hskip 0.5em}
            \def\ScoreOverhang{2pt}
                \AXC{$\Gamma \, \vdash \, \ottnt{A}$}
                \AXC{$\Gamma'  \ottsym{,}  \ottnt{B} \, \vdash \, \ottnt{C}$}
              \RightLabel{$ \rlimpl $}
              \BIC{$\Gamma  \ottsym{,}  \Gamma'  \ottsym{,}  \ottnt{A}  \multimap  \ottnt{B} \, \vdash \, \ottnt{C}$}
            \end{prooftree}
          \end{minipage}

          \begin{minipage}{0.15\hsize}
            \begin{prooftree}
            \def\ScoreOverhang{2pt}
                \AXC{$\Gamma  \ottsym{,}  \ottnt{A} \, \vdash \, \ottnt{B}$}
              \RightLabel{$ \rdere $}
              \UIC{$\Gamma  \ottsym{,}  \ottsym{!}  \ottnt{A} \, \vdash \, \ottnt{B}$}
            \end{prooftree}
          \end{minipage}

          \begin{minipage}{0.16\hsize}
            \begin{prooftree}
            \def\ScoreOverhang{2pt}
                \AXC{$\Gamma \, \vdash \, \ottnt{B}$}
              \RightLabel{$ \rbweak $}
              \UIC{$\Gamma  \ottsym{,}  \ottsym{!}  \ottnt{A} \, \vdash \, \ottnt{B}$}
            \end{prooftree}
          \end{minipage}

          \begin{minipage}{0.16\hsize}
            \begin{prooftree}
            \def\ScoreOverhang{2pt}
                \AXC{$\Gamma  \ottsym{,}  \ottsym{!}  \ottnt{A}  \ottsym{,}  \ottsym{!}  \ottnt{A} \, \vdash \, \ottnt{B}$}
              \RightLabel{$ \rbcont $}
              \UIC{$\Gamma  \ottsym{,}  \ottsym{!}  \ottnt{A} \, \vdash \, \ottnt{B}$}
            \end{prooftree}
          \end{minipage}
        \end{tabular}
        \end{center}
    \end{minipage}
    }
  \end{tabular}

%% file: definitions/girard_translation.tex
  \begin{tabular}{c|c}
    {
      \begin{minipage}{0.60\hsize}
        \underline{$ \fgirardtrans{ \ottnt{A} } $}
        \vspace{-1.85em}
        \begin{align*}
          \hspace{1.5em} \fgirardtrans{ \ottmv{p} }  \defeq p, \hspace{3.5em}  \fgirardtrans{ \ottnt{A}  \supset  \ottnt{B} }  \defeq  (  \ottsym{!}   \fgirardtrans{ \ottnt{A} }   )   \multimap   \fgirardtrans{ \ottnt{B} } 
        \end{align*}
      \end{minipage}
    }
    &
    {
      \begin{minipage}{0.30\hsize}
        \underline{$ \fgirardtrans{ \Gamma } $}
        \vspace{-1.85em}
        \begin{align*}
          \hspace{1.5em} \fgirardtrans{ \Gamma }  &\defeq \{  \fgirardtrans{ \ottnt{A} }  ~|~ A \in \Gamma \}
        \end{align*}
      \end{minipage}
    }
  \end{tabular}

%% file: definitions/ljbox.tex
  \hspace{-1em}
  \begin{tabular}{c|c}
    {
      \hspace{-1em}
      \begin{minipage}{0.28\hsize}
        \vspace{-0.2em}
        \paragraph*{Syntactic category}
          \vspace{-1em}
          \begin{align*}
            \hspace{-1em}\text{Formulae}~~A, B, C ::= p ~|~ \ottnt{A}  \supset  \ottnt{B} ~|~  \Box \ottnt{A} 
          \end{align*}
      \end{minipage}
    }
    &
    {
      \hspace{-1em}
      \begin{minipage}{0.50\hsize}
        \vspace{0.2em}
        \paragraph*{Inference rule}
          \vspace{-0.8em}
          \begin{tabular}{c}
            \begin{minipage}{0.165\hsize}
              \begin{prooftree}
              \def\ScoreOverhang{2pt}
                  \AXC{$ $}
                \RightLabel{$ \rax $}
                \UIC{$\ottnt{A} \, \vdash \, \ottnt{A}$}
              \end{prooftree}
            \end{minipage}

            \begin{minipage}{0.598\hsize}
              \begin{prooftree}
              \def\ScoreOverhang{2pt}
              \def\defaultHypSeparation{\hskip 0.5em}
                  \AXC{$\Gamma \, \vdash \, \ottnt{A}$}
                  \AXC{$\Gamma'  \ottsym{,}  \ottnt{A} \, \vdash \, \ottnt{B}$}
                \RightLabel{$ \rcut $}
                \BIC{$\Gamma  \ottsym{,}  \Gamma' \, \vdash \, \ottnt{B}$}
              \end{prooftree}
            \end{minipage}

            \begin{minipage}{0.17\hsize}
              \begin{prooftree}
              \def\ScoreOverhang{2pt}
                  \AXC{$ \Box \Gamma  \, \vdash \, \ottnt{A}$}
                \RightLabel{$ \rboxr $}
                \UIC{$ \Box \Gamma  \, \vdash \,  \Box \ottnt{A} $}
              \end{prooftree}
            \end{minipage}
          \end{tabular}
      \end{minipage}
    }
  \end{tabular}

  \hspace{-0.75em}\rule{0.445\textwidth}{0.4pt}

  \hspace{-1em}
  \begin{tabular}{c}
  {
    \hspace{-1em}
    \begin{minipage}{\hsize}
        \begin{center}
        \vspace{-0.3em}
        \begin{tabular}{c}
          \begin{minipage}{0.18\hsize}
            \begin{prooftree}
              \def\ScoreOverhang{2pt}
                \AXC{$\Gamma  \ottsym{,}  \ottnt{A} \, \vdash \, \ottnt{B}$}
              \RightLabel{$ \rimpr $}
              \UIC{$\Gamma \, \vdash \, \ottnt{A}  \supset  \ottnt{B}$}
            \end{prooftree}
          \end{minipage}

          \begin{minipage}{0.30\hsize}
            \begin{prooftree}
              \def\ScoreOverhang{2pt}
                \AXC{$\Gamma \, \vdash \, \ottnt{A}$}
                \AXC{$\Gamma'  \ottsym{,}  \ottnt{B} \, \vdash \, \ottnt{C}$}
              \RightLabel{$ \rimpl $}
              \BIC{$\Gamma  \ottsym{,}  \Gamma'  \ottsym{,}  \ottnt{A}  \supset  \ottnt{B} \, \vdash \, \ottnt{C}$}
            \end{prooftree}
          \end{minipage}

          \begin{minipage}{0.165\hsize}
            \begin{prooftree}
              \def\ScoreOverhang{2pt}
                \AXC{$\Gamma  \ottsym{,}  \ottnt{A} \, \vdash \, \ottnt{B}$}
              \RightLabel{$ \rboxl $}
              \UIC{$\Gamma  \ottsym{,}   \Box \ottnt{A}  \, \vdash \, \ottnt{B}$}
            \end{prooftree}
          \end{minipage}

          \begin{minipage}{0.16\hsize}
            \begin{prooftree}
              \def\ScoreOverhang{2pt}
                \AXC{$\Gamma \, \vdash \, \ottnt{B}$}
              \RightLabel{$ \rweak $}
              \UIC{$\Gamma  \ottsym{,}  \ottnt{A} \, \vdash \, \ottnt{B}$}
            \end{prooftree}
          \end{minipage}

          \begin{minipage}{0.16\hsize}
            \begin{prooftree}
              \def\ScoreOverhang{2pt}
                \AXC{$\Gamma  \ottsym{,}  \ottnt{A}  \ottsym{,}  \ottnt{A} \, \vdash \, \ottnt{B}$}
              \RightLabel{$ \rcont $}
              \UIC{$\Gamma  \ottsym{,}  \ottnt{A} \, \vdash \, \ottnt{B}$}
            \end{prooftree}
          \end{minipage}
        \end{tabular}
        \end{center}
    \end{minipage}
  }
  \end{tabular}

%% file: definitions/lambdabox.tex
  \hspace{-1em}
  \begin{tabular}{c|c}
    {
    \hspace{-1em}
    \begin{minipage}{0.50\hsize}
      \paragraph*{Syntactic category}
        \vspace{-1em}
        \begin{alignat*}{3}
          &\hspace{-1em}\text{Types}~&A, B, C &::= p ~|~ \ottnt{A}  \supset  \ottnt{B} ~|~  \Box \ottnt{A} \\
          &\hspace{-1em}\text{Terms}~&M, N, L &::= x ~|~ \lambda  \ottmv{x}  \ottsym{:}  \ottnt{A}  \ottsym{.}  \ottnt{M} ~|~  \ottnt{M} \, \ottnt{N} \\
          &&&\,|~  \Box  \ottnt{M}  ~|~ \ottkw{let} \,  \Box  \ottmv{x}   \ottsym{=}  \ottnt{M} \, \ottkw{in} \, \ottnt{N}
        \end{alignat*}
    \end{minipage}
    }
    &
    {
    \hspace{-1em}
    \begin{minipage}{0.45\hsize}
      \vspace{-1.5em}
      \paragraph*{Reduction rule}
        \vspace{-1em}
        \begin{alignat*}{3}
          &\hspace{-1em}\rbetaimp\hspace{0.5em}& & \ottsym{(}  \lambda  \ottmv{x}  \ottsym{:}  \ottnt{A}  \ottsym{.}  \ottnt{M}  \ottsym{)} \, \ottnt{N} &             & \leadsto   \ottnt{M}  [  \ottmv{x}  :=  \ottnt{N}  ] \\
          &\hspace{-1em}\rbetabox\hspace{0.5em}& &\ottkw{let} \,  \Box  \ottmv{x}   \ottsym{=}   \Box  \ottnt{N}  \, \ottkw{in} \, \ottnt{M}& & \leadsto   \ottnt{M}  [  \ottmv{x}  :=  \ottnt{N}  ] 
        \end{alignat*}
    \end{minipage}
    }
  \end{tabular}

  \rule{\textwidth}{0.4pt}

  \hspace{-1em}
  \begin{tabular}{c}
    {
    \hspace{-1em}
    \begin{minipage}{0.9\hsize}
      \vspace{0.2em}
      \paragraph*{Typing rule}
        \vspace{-0.5em}
        \begin{center}
        \begin{tabular}{c}
          \begin{minipage}{0.45\hsize}
            \begin{prooftree}
                \AXC{$ $}
              \RightLabel{$ \rax $}
              \UIC{$\Delta  \ottsym{;}  \Gamma  \ottsym{,}  \ottmv{x}  \ottsym{:}  \ottnt{A} \, \vdash \, \ottmv{x}  \ottsym{:}  \ottnt{A}$}
            \end{prooftree}
          \end{minipage}

          \begin{minipage}{0.45\hsize}
            \begin{prooftree}
                \AXC{$ $}
              \RightLabel{$ \rmax $}
              \UIC{$\Delta  \ottsym{,}  \ottmv{x}  \ottsym{:}  \ottnt{A}  \ottsym{;}  \Gamma \, \vdash \, \ottmv{x}  \ottsym{:}  \ottnt{A}$}
            \end{prooftree}
          \end{minipage}
        \end{tabular}

        \begin{tabular}{c}
          \begin{minipage}{0.42\hsize}
            \begin{prooftree}
                \AXC{$\Delta  \ottsym{;}  \Gamma  \ottsym{,}  \ottmv{x}  \ottsym{:}  \ottnt{A} \, \vdash \, \ottnt{M}  \ottsym{:}  \ottnt{B}$}
              \RightLabel{$ \rimpi $}
              \UIC{$\Delta  \ottsym{;}  \Gamma \, \vdash \, \ottsym{(}  \lambda  \ottmv{x}  \ottsym{:}  \ottnt{A}  \ottsym{.}  \ottnt{M}  \ottsym{)}  \ottsym{:}  \ottnt{A}  \supset  \ottnt{B}$}
            \end{prooftree}
          \end{minipage}

          \begin{minipage}{0.52\hsize}
            \begin{prooftree}
                \AXC{$\Delta  \ottsym{;}  \Gamma \, \vdash \, \ottnt{M}  \ottsym{:}  \ottnt{A}  \supset  \ottnt{B}$}
                \AXC{$\Delta  \ottsym{;}  \Gamma \, \vdash \, \ottnt{N}  \ottsym{:}  \ottnt{A}$}
              \RightLabel{$ \rimpe $}
              \BIC{$\Delta  \ottsym{;}  \Gamma \, \vdash \,  \ottnt{M} \, \ottnt{N}   \ottsym{:}  \ottnt{B}$}
            \end{prooftree}
          \end{minipage}
        \end{tabular}

        \begin{tabular}{c}
          \begin{minipage}{0.38\hsize}
            \begin{prooftree}
                \AXC{$\Delta  \ottsym{;}  \emptyset \, \vdash \, \ottnt{M}  \ottsym{:}  \ottnt{A}$}
              \RightLabel{$ \rboxi $}
              \UIC{$\Delta  \ottsym{;}  \Gamma \, \vdash \,  \Box  \ottnt{M}   \ottsym{:}   \Box \ottnt{A} $}
            \end{prooftree}
          \end{minipage}

          \begin{minipage}{0.55\hsize}
            \begin{prooftree}
                \AXC{$\Delta  \ottsym{;}  \Gamma \, \vdash \, \ottnt{M}  \ottsym{:}   \Box \ottnt{A} $}
                \AXC{$\Delta  \ottsym{,}  \ottmv{x}  \ottsym{:}  \ottnt{A}  \ottsym{;}  \Gamma \, \vdash \, \ottnt{N}  \ottsym{:}  \ottnt{B}$}
              \RightLabel{$ \rboxe $}
              \BIC{$\Delta  \ottsym{;}  \Gamma \, \vdash \, \ottkw{let} \,  \Box  \ottmv{x}   \ottsym{=}  \ottnt{M} \, \ottkw{in} \, \ottnt{N}  \ottsym{:}  \ottnt{B}$}
            \end{prooftree}
          \end{minipage}
        \end{tabular}
        \end{center}
    \end{minipage}
    }
  \end{tabular}

%% file: definitions/imellbangbox.tex
  \hspace{-1em}
  \begin{tabular}{c|c}
    {
    \hspace{-1em}
    \begin{minipage}{0.5\hsize}
      \vspace{-0.2em}
      \paragraph*{Syntactic category}
        \vspace{-1em}
        \begin{align*}
          \text{Formulae}~~A, B, C ::= p ~|~ \ottnt{A}  \multimap  \ottnt{B} ~|~ \ottsym{!}  \ottnt{A} ~|~  \bangbox \ottnt{A} 
        \end{align*}
    \end{minipage}
    }
    &
    {
    \hspace{-1em}
    \begin{minipage}{0.5\hsize}
      \vspace{0.2em}
      \paragraph*{Inference rule}
        \vspace{-0.8em}
        \begin{tabular}{c}
          \begin{minipage}{0.30\hsize}
            \begin{prooftree}
              \def\ScoreOverhang{2pt}
                \AXC{$ $}
              \RightLabel{$ \rax $}
              \UIC{$\ottnt{A} \, \vdash \, \ottnt{A}$}
            \end{prooftree}
          \end{minipage}

          \begin{minipage}{0.60\hsize}
            \begin{prooftree}
              \def\ScoreOverhang{2pt}
              \def\defaultHypSeparation{\hskip 0.5em}
                \AXC{$\Gamma \, \vdash \, \ottnt{A}$}
                \AXC{$\ottnt{A}  \ottsym{,}  \Gamma' \, \vdash \, \ottnt{B}$}
              \RightLabel{$ \rcut $}
              \BIC{$\Gamma  \ottsym{,}  \Gamma' \, \vdash \, \ottnt{B}$}
            \end{prooftree}
          \end{minipage}
        \end{tabular}
    \end{minipage}
    }
  \end{tabular}

  \hspace{-0.75em}\rule{0.53\textwidth}{0.4pt}

  \hspace{-1em}
  \begin{tabular}{c}
    {
    \hspace{-1em}
    \begin{minipage}{\hsize}
        \vspace{-0.3em}
        \begin{tabular}{c}
          \begin{minipage}{0.195\hsize}
            \begin{prooftree}
              \def\ScoreOverhang{2pt}
                \AXC{$\Gamma  \ottsym{,}  \ottnt{A} \, \vdash \, \ottnt{B}$}
              \RightLabel{$ \rlimpr $}
              \UIC{$\Gamma \, \vdash \, \ottnt{A}  \multimap  \ottnt{B}$}
            \end{prooftree}
          \end{minipage}

          \begin{minipage}{0.27\hsize}
            \begin{prooftree}
              \def\ScoreOverhang{2pt}
              \def\defaultHypSeparation{\hskip 0.5em}
                \AXC{$\Gamma \, \vdash \, \ottnt{A}$}
                \AXC{$\Gamma'  \ottsym{,}  \ottnt{B} \, \vdash \, \ottnt{C}$}
              \RightLabel{$ \rlimpl $}
              \BIC{$\Gamma  \ottsym{,}  \Gamma'  \ottsym{,}  \ottnt{A}  \multimap  \ottnt{B} \, \vdash \, \ottnt{C}$}
            \end{prooftree}
          \end{minipage}

          \begin{minipage}{0.18\hsize}
            \begin{prooftree}
              \def\ScoreOverhang{2pt}
                \AXC{$ \bangbox \Delta   \ottsym{,}  \ottsym{!}  \Gamma \, \vdash \, \ottnt{A}$}
              \RightLabel{$ \rbangr $}
              \UIC{$ \bangbox \Delta   \ottsym{,}  \ottsym{!}  \Gamma \, \vdash \, \ottsym{!}  \ottnt{A}$}
            \end{prooftree}
          \end{minipage}

          \begin{minipage}{0.16\hsize}
            \begin{prooftree}
              \def\ScoreOverhang{2pt}
                \AXC{$\Gamma  \ottsym{,}  \ottnt{A} \, \vdash \, \ottnt{B}$}
              \RightLabel{$ \rbangl $}
              \UIC{$\Gamma  \ottsym{,}  \ottsym{!}  \ottnt{A} \, \vdash \, \ottnt{B}$}
            \end{prooftree}
          \end{minipage}

          \begin{minipage}{0.16\hsize}
            \begin{prooftree}
              \def\ScoreOverhang{2pt}
                \AXC{$ \bangbox \Delta  \, \vdash \, \ottnt{A}$}
              \RightLabel{$ \rbangboxr $}
              \UIC{$ \bangbox \Delta  \, \vdash \,  \bangbox \ottnt{A} $}
            \end{prooftree}
          \end{minipage}
        \end{tabular}

        \begin{tabular}{c}
          \begin{minipage}{0.175\hsize}
            \begin{prooftree}
              \def\ScoreOverhang{2pt}
                \AXC{$\Gamma  \ottsym{,}  \ottnt{A} \, \vdash \, \ottnt{B}$}
              \RightLabel{$ \rbangboxl $}
              \UIC{$\Gamma  \ottsym{,}   \bangbox \ottnt{A}  \, \vdash \, \ottnt{B}$}
            \end{prooftree}
          \end{minipage}

          \begin{minipage}{0.178\hsize}
            \begin{prooftree}
              \def\ScoreOverhang{2pt}
                \AXC{$\Gamma \, \vdash \, \ottnt{B}$}
              \RightLabel{$ \rbweak $}
              \UIC{$\Gamma  \ottsym{,}  \ottsym{!}  \ottnt{A} \, \vdash \, \ottnt{B}$}
            \end{prooftree}
          \end{minipage}

          \begin{minipage}{0.20\hsize}
            \begin{prooftree}
              \def\ScoreOverhang{2pt}
                \AXC{$\Gamma  \ottsym{,}  \ottsym{!}  \ottnt{A}  \ottsym{,}  \ottsym{!}  \ottnt{A} \, \vdash \, \ottnt{B}$}
              \RightLabel{$ \rbcont $}
              \UIC{$\Gamma  \ottsym{,}  \ottsym{!}  \ottnt{A} \, \vdash \, \ottnt{B}$}
            \end{prooftree}
          \end{minipage}

          \begin{minipage}{0.185\hsize}
            \begin{prooftree}
              \def\ScoreOverhang{2pt}
                \AXC{$\Gamma \, \vdash \, \ottnt{B}$}
              \RightLabel{$ \rbbweak $}
              \UIC{$\Gamma  \ottsym{,}   \bangbox \ottnt{A}  \, \vdash \, \ottnt{B}$}
            \end{prooftree}
          \end{minipage}

          \begin{minipage}{0.20\hsize}
            \begin{prooftree}
              \def\ScoreOverhang{2pt}
                \AXC{$\Gamma  \ottsym{,}   \bangbox \ottnt{A}   \ottsym{,}   \bangbox \ottnt{A}  \, \vdash \, \ottnt{B}$}
              \RightLabel{$ \rbbcont $}
              \UIC{$\Gamma  \ottsym{,}   \bangbox \ottnt{A}  \, \vdash \, \ottnt{B}$}
            \end{prooftree}
          \end{minipage}
        \end{tabular}
      \end{minipage}
    }
  \end{tabular}

%% file: definitions/modal_girard_translation.tex
  \hspace{-1em}
  \begin{tabular}{c|c}
    {
      \hspace{-1em}
      \begin{minipage}{0.60\hsize}
        \underline{$ \fgirardtrans{ \ottnt{A} } $}
        \vspace{-1.5em}
        \begin{align*}
          \hspace{0.4em}
           \fgirardtrans{ \ottmv{p} }        \defeq \ottmv{p}, \hspace{0.4em}
           \fgirardtrans{ \ottnt{A}  \supset  \ottnt{B} }   \defeq  (  \ottsym{!}   \fgirardtrans{ \ottnt{A} }   )   \multimap   \fgirardtrans{ \ottnt{B} } , \hspace{0.4em}
           \fgirardtrans{  \Box \ottnt{A}  }    \defeq  \bangbox  \fgirardtrans{ \ottnt{A} }  
        \end{align*}
      \end{minipage}
    }
    &
    {
      \hspace{-1em}
      \begin{minipage}{0.37\hsize}
        \underline{$ \fgirardtrans{ \Gamma } $}
        \vspace{-1.5em}
        \begin{align*}
           \fgirardtrans{ \Gamma }  &\defeq \{ (x :  \fgirardtrans{ \ottnt{A} } ) ~|~ (x : A) \in \Gamma \}
        \end{align*}
      \end{minipage}
    }
  \end{tabular}

%% file: definitions/lambdabangbox.tex
  \hspace{-1em}
  \begin{tabular}{c|c}
    {
    \hspace{-1em}
    \begin{minipage}{0.48\hsize}
      \paragraph*{Syntactic category}
        \vspace{-1em}
        \begin{alignat*}{3}
          &\hspace{-1em}\text{Types}~&A, B, C &::= p ~|~ \ottnt{A}  \multimap  \ottnt{B} ~|~ \ottsym{!}  \ottnt{A} ~|~  \bangbox \ottnt{A} \\
          &\hspace{-1em}\text{Terms}~&M, N, L &::= x ~|~ \lambda  \ottmv{x}  \ottsym{:}  \ottnt{A}  \ottsym{.}  \ottnt{M} ~|~  \ottnt{M} \, \ottnt{N} ~|~  !  \ottnt{M}  ~|~ \bangbox  \ottnt{M} \\
          &&&\,|~ \ottkw{let} \,  !  \ottmv{x}   \ottsym{=}  \ottnt{M} \, \ottkw{in} \, \ottnt{N}~|~ \ottkw{let} \,  \bangbox  \ottmv{x}   \ottsym{=}  \ottnt{M} \, \ottkw{in} \, \ottnt{N}
        \end{alignat*} 
    \end{minipage}
    }
    &
    {
    \hspace{-1em}
    \begin{minipage}{0.50\hsize}
      \paragraph*{Reduction rule}
        \vspace{-1em}
        \begin{alignat*}{3}
          &\hspace{-1.2em}\rbetalimp\hspace{0.5em}& & \ottsym{(}  \lambda  \ottmv{x}  \ottsym{:}  \ottnt{A}  \ottsym{.}  \ottnt{M}  \ottsym{)} \, \ottnt{N} &                      & \leadsto   \ottnt{M}  [  \ottmv{x}  :=  \ottnt{N}  ] \\
          &\hspace{-1.2em}\rbetabang\hspace{0.5em}& &\ottkw{let} \,  !  \ottmv{x}   \ottsym{=}   !  \ottnt{N}  \, \ottkw{in} \, \ottnt{M}&       & \leadsto   \ottnt{M}  [  \ottmv{x}  :=  \ottnt{N}  ] \\
          &\hspace{-1.2em}\rbetabangbox\hspace{0.5em}& &\ottkw{let} \,  \bangbox  \ottmv{x}   \ottsym{=}   \bangbox  \ottnt{N}  \, \ottkw{in} \, \ottnt{M}& & \leadsto   \ottnt{M}  [  \ottmv{x}  :=  \ottnt{N}  ] 
        \end{alignat*}
    \end{minipage}
    }
  \end{tabular}

  \rule{\textwidth}{0.4pt}

  \hspace{-1em}
  \begin{tabular}{c}
    {
    \hspace{-1.2em}
    \begin{minipage}{0.98\hsize}
      \vspace{0.2em}
      \paragraph*{Typing rule}
        \vspace{-1em}
        \begin{center}
        \begin{tabular}{c}
          \begin{minipage}{0.30\hsize}
            \begin{prooftree}
                \AXC{$ $}
              \RightLabel{$ \rlax $}
              \UIC{$\Delta  \ottsym{;}  \Gamma  \ottsym{;}  \ottmv{x}  \ottsym{:}  \ottnt{A} \, \vdash \, \ottmv{x}  \ottsym{:}  \ottnt{A}$}
            \end{prooftree}
          \end{minipage}
      
          \begin{minipage}{0.30\hsize}
            \begin{prooftree}
                \AXC{$ $}
              \RightLabel{$ \rbax $}
              \UIC{$\Delta  \ottsym{;}  \Gamma  \ottsym{,}  \ottmv{x}  \ottsym{:}  \ottnt{A}  \ottsym{;}  \emptyset \, \vdash \, \ottmv{x}  \ottsym{:}  \ottnt{A}$}
            \end{prooftree}
          \end{minipage}
      
          \begin{minipage}{0.32\hsize}
            \begin{prooftree}
                \AXC{$ $}
              \RightLabel{$ \rbbax $}
              \UIC{$\Delta  \ottsym{,}  \ottmv{x}  \ottsym{:}  \ottnt{A}  \ottsym{;}  \Gamma  \ottsym{;}  \emptyset \, \vdash \, \ottmv{x}  \ottsym{:}  \ottnt{A}$}
            \end{prooftree}
          \end{minipage}
        \end{tabular}
      
        \begin{tabular}{c}
          \begin{minipage}{0.38\hsize}
            \begin{prooftree}
                \AXC{$\Delta  \ottsym{;}  \Gamma  \ottsym{;}  \Sigma  \ottsym{,}  \ottmv{x}  \ottsym{:}  \ottnt{A} \, \vdash \, \ottnt{M}  \ottsym{:}  \ottnt{B}$}
              \RightLabel{$ \rlimpi $}
              \UIC{$\Delta  \ottsym{;}  \Gamma  \ottsym{;}  \Sigma \, \vdash \, \lambda  \ottmv{x}  \ottsym{:}  \ottnt{A}  \ottsym{.}  \ottnt{M}  \ottsym{:}  \ottnt{A}  \multimap  \ottnt{B}$}
            \end{prooftree}
          \end{minipage}
      
          \begin{minipage}{0.57\hsize}
            \begin{prooftree}
                \AXC{$\Delta  \ottsym{;}  \Gamma  \ottsym{;}  \Sigma \, \vdash \, \ottnt{M}  \ottsym{:}  \ottnt{A}  \multimap  \ottnt{B}$}
                \AXC{$\Delta  \ottsym{;}  \Gamma  \ottsym{;}  \Sigma' \, \vdash \, \ottnt{N}  \ottsym{:}  \ottnt{A}$}
              \RightLabel{$ \rlimpe $}
              \BIC{$\Delta  \ottsym{;}  \Gamma  \ottsym{;}  \Sigma  \ottsym{,}  \Sigma' \, \vdash \,  \ottnt{M} \, \ottnt{N}   \ottsym{:}  \ottnt{B}$}
            \end{prooftree}
          \end{minipage}
        \end{tabular}
      
        \begin{tabular}{c}
          \begin{minipage}{0.37\hsize}
            \begin{prooftree}
                \AXC{$\Delta  \ottsym{;}  \Gamma  \ottsym{;}  \emptyset \, \vdash \, \ottnt{M}  \ottsym{:}  \ottnt{A}$}
              \RightLabel{$ \rbangi $}
              \UIC{$\Delta  \ottsym{;}  \Gamma  \ottsym{;}  \emptyset \, \vdash \,  !  \ottnt{M}   \ottsym{:}  \ottsym{!}  \ottnt{A}$}
            \end{prooftree}
          \end{minipage}
      
          \begin{minipage}{0.55\hsize}
            \begin{prooftree}
                \AXC{$\Delta  \ottsym{;}  \Gamma  \ottsym{;}  \Sigma \, \vdash \, \ottnt{M}  \ottsym{:}  \ottsym{!}  \ottnt{A}$}
                \AXC{$\Delta  \ottsym{;}  \Gamma  \ottsym{,}  \ottmv{x}  \ottsym{:}  \ottnt{A}  \ottsym{;}  \Sigma' \, \vdash \, \ottnt{N}  \ottsym{:}  \ottnt{B}$}
              \RightLabel{$ \rbange $}
              \BIC{$\Delta  \ottsym{;}  \Gamma  \ottsym{;}  \Sigma  \ottsym{,}  \Sigma' \, \vdash \, \ottkw{let} \,  !  \ottmv{x}   \ottsym{=}  \ottnt{M} \, \ottkw{in} \, \ottnt{N}  \ottsym{:}  \ottnt{B}$}
            \end{prooftree}
          \end{minipage}
        \end{tabular}
      
        \begin{tabular}{c}
          \begin{minipage}{0.37\hsize}
            \begin{prooftree}
                \AXC{$\Delta  \ottsym{;}  \emptyset  \ottsym{;}  \emptyset \, \vdash \, \ottnt{M}  \ottsym{:}  \ottnt{A}$}
              \RightLabel{$ \rbangboxi $}
              \UIC{$\Delta  \ottsym{;}  \Gamma  \ottsym{;}  \emptyset \, \vdash \,  \bangbox  \ottnt{M}   \ottsym{:}   \bangbox \ottnt{A} $}
            \end{prooftree}
          \end{minipage}
      
          \begin{minipage}{0.57\hsize}
            \begin{prooftree}
                \AXC{$\Delta  \ottsym{;}  \Gamma  \ottsym{;}  \Sigma \, \vdash \, \ottnt{M}  \ottsym{:}   \bangbox \ottnt{A} $}
                \AXC{$\Delta  \ottsym{,}  \ottmv{x}  \ottsym{:}  \ottnt{A}  \ottsym{;}  \Gamma  \ottsym{;}  \Sigma' \, \vdash \, \ottnt{N}  \ottsym{:}  \ottnt{B}$}
              \RightLabel{$ \rbangboxe $}
              \BIC{$\Delta  \ottsym{;}  \Gamma  \ottsym{;}  \Sigma  \ottsym{,}  \Sigma' \, \vdash \, \ottkw{let} \,  \bangbox  \ottmv{x}   \ottsym{=}  \ottnt{M} \, \ottkw{in} \, \ottnt{N}  \ottsym{:}  \ottnt{B}$}
            \end{prooftree}
          \end{minipage}
        \end{tabular}
        \end{center}
    \end{minipage}
    }
  \end{tabular}

%% file: definitions/modal_embedding.tex
  \hspace{-1em}
  \begin{tabular}{c|c}
    {
    \begin{minipage}{0.38\hsize}
      \underline{$ \fgirardtrans{ \ottnt{A} } $}
      \vspace{-2em}
      \begin{align*}
         \fgirardtrans{ \ottmv{p} }       &\defeq p \\
         \fgirardtrans{ \ottnt{A}  \supset  \ottnt{B} }  &\defeq \ottsym{!}   \fgirardtrans{ \ottnt{A} }   \multimap   \fgirardtrans{ \ottnt{B} }  \\
         \fgirardtrans{  \Box \ottnt{A}  }   &\defeq  \bangbox  \fgirardtrans{ \ottnt{A} }  
      \end{align*}
      \vspace{-2.5em}

      \rule{\textwidth}{0.4pt}

      \underline{$ \fgirardtrans{ \Gamma } $}
      \vspace{-1.8em}
      \begin{align*}
        \hspace{0.3em} \fgirardtrans{ \Gamma }  &\defeq \{ (x :  \fgirardtrans{ \ottnt{A} } ) ~|~ (x : A) \in \Gamma \}
      \end{align*}
    \end{minipage}
    }
    &
    {
    \begin{minipage}{0.60\hsize}
      \vspace{0.3em}
      \underline{$ \ftrans{ \ottnt{M} } $}
      \vspace{-2.6em}
      \begin{align*}
         \ftrans{ \ottmv{x} }  &\defeq \ottmv{x}\\
         \ftrans{ \lambda  \ottmv{x}  \ottsym{:}  \ottnt{A}  \ottsym{.}  \ottnt{M} }  &\defeq \overline{\lambda}  \ottmv{x}  \ottsym{:}   \fgirardtrans{ \ottnt{A} }   \ottsym{.}   \ftrans{ \ottnt{M} } \\
         \ftrans{  \ottnt{M} \, \ottnt{N}  }  &\defeq  \ftrans{ \ottnt{M} }   \overline{@}   \ftrans{ \ottnt{N} } \\
         \ftrans{  \Box  \ottnt{M}  }  &\defeq  \bangbox   \ftrans{ \ottnt{M} }  \\
         \ftrans{ \ottkw{let} \,  \Box  \ottmv{x}   \ottsym{=}  \ottnt{M} \, \ottkw{in} \, \ottnt{N} }  &\defeq \ottkw{let} \,  \bangbox  \ottmv{x}   \ottsym{=}   \ftrans{ \ottnt{M} }  \, \ottkw{in} \,  \ftrans{ \ottnt{N} }  
      \end{align*}
    \end{minipage}
    }
  \end{tabular}

%% file: definitions/clbangbox.tex
  \hspace{-1em}
  \begin{tabular}{c}
    {
    \hspace{-1.2em}
    \begin{minipage}{0.98\hsize}
    \paragraph*{Syntactic category}
      \vspace{-1em}
      \begin{alignat*}{3}
        &\text{Types}~~&A, B, C &::= p ~|~ \ottnt{A}  \multimap  \ottnt{B} ~|~ \ottsym{!}  \ottnt{A} ~|~  \bangbox \ottnt{A} \\
        &\text{Terms}~~&M, N, L &::= x ~|~ c ~|~  !  \ottnt{M}  ~|~  \bangbox  \ottnt{M} 
      \end{alignat*}
    \end{minipage}
    }
  \end{tabular}

  \vspace{0.5em}
  \rule{\textwidth}{0.4pt}

  \vspace{0.2em}

  \hspace{-1em}
  \begin{tabular}{c}
    {
    \hspace{-1.2em}
    \begin{minipage}{0.98\hsize}
      \paragraph*{Typing rule}
        \vspace{-0.7em}
        \begin{center}
        \begin{tabular}{c}
          \begin{minipage}{0.30\hsize}
            \begin{prooftree}
                \AXC{($c$ is a combinator)}
              \RightLabel{$ \hcomb $}
              \UIC{$\Delta  \ottsym{;}  \Gamma  \ottsym{;}  \emptyset  \vdash  c : \ftypeof{c}$}
            \end{prooftree}
          \end{minipage}

          \begin{minipage}{0.60\hsize}
            \begin{prooftree}
                \AXC{$\Delta  \ottsym{;}  \Gamma  \ottsym{;}  \Sigma \, \vdash \, \ottnt{M}  \ottsym{:}  \ottnt{A}  \multimap  \ottnt{B}$}
                \AXC{$\Delta  \ottsym{;}  \Gamma  \ottsym{;}  \Sigma' \, \vdash \, \ottnt{N}  \ottsym{:}  \ottnt{A}$}
              \RightLabel{$ \hmp $}
              \BIC{$\Delta  \ottsym{;}  \Gamma  \ottsym{;}  \Sigma  \ottsym{,}  \Sigma' \, \vdash \,  \ottnt{M} \, \ottnt{N}   \ottsym{:}  \ottnt{B}$}
            \end{prooftree}
          \end{minipage}
        \end{tabular}

        \begin{tabular}{c}
          \begin{minipage}{0.30\hsize}
            \begin{prooftree}
                \AXC{$ $}
              \RightLabel{$ \rlax $}
              \UIC{$\Delta  \ottsym{;}  \Gamma  \ottsym{;}  \ottmv{x}  \ottsym{:}  \ottnt{A} \, \vdash \, \ottmv{x}  \ottsym{:}  \ottnt{A}$}
            \end{prooftree}
          \end{minipage}

          \begin{minipage}{0.30\hsize}
            \begin{prooftree}
                \AXC{$ $}
              \RightLabel{$ \rbax $}
              \UIC{$\Delta  \ottsym{;}  \Gamma  \ottsym{,}  \ottmv{x}  \ottsym{:}  \ottnt{A}  \ottsym{;}  \emptyset \, \vdash \, \ottmv{x}  \ottsym{:}  \ottnt{A}$}
            \end{prooftree}
          \end{minipage}

          \begin{minipage}{0.32\hsize}
            \begin{prooftree}
                \AXC{$ $}
              \RightLabel{$ \rbbax $}
              \UIC{$\Delta  \ottsym{,}  \ottmv{x}  \ottsym{:}  \ottnt{A}  \ottsym{;}  \Gamma  \ottsym{;}  \emptyset \, \vdash \, \ottmv{x}  \ottsym{:}  \ottnt{A}$}
            \end{prooftree}
          \end{minipage}
        \end{tabular}

        \begin{tabular}{c}
          \begin{minipage}{0.45\hsize}
            \begin{prooftree}
                \AXC{$\Delta  \ottsym{;}  \Gamma  \ottsym{;}  \emptyset \, \vdash \, \ottnt{M}  \ottsym{:}  \ottnt{A}$}
              \RightLabel{$ \hbnec $}
              \UIC{$\Delta  \ottsym{;}  \Gamma  \ottsym{;}  \emptyset \, \vdash \,  !  \ottnt{M}   \ottsym{:}  \ottsym{!}  \ottnt{A}$}
            \end{prooftree}
          \end{minipage}

          \begin{minipage}{0.45\hsize}
            \begin{prooftree}
                \AXC{$\Delta  \ottsym{;}  \emptyset  \ottsym{;}  \emptyset \, \vdash \, \ottnt{M}  \ottsym{:}  \ottnt{A}$}
              \RightLabel{$ \hbbnec $}
              \UIC{$\Delta  \ottsym{;}  \Gamma  \ottsym{;}  \emptyset \, \vdash \,  \bangbox  \ottnt{M}   \ottsym{:}   \bangbox \ottnt{A} $}
            \end{prooftree}
          \end{minipage}
        \end{tabular}
        \end{center}
    \end{minipage}
    }
  \end{tabular}

  \vspace{0.5em}
  \rule{\textwidth}{0.4pt}

  \vspace{0.2em}

  \hspace{-1em}
  \begin{tabular}{c|c}
  {
    \hspace{-1em}
    \begin{minipage}{0.55\hsize}
    \paragraph*{Combinator}
      \begin{itemize}
        \item $\vdash \,  \mathrm{I}   \ottsym{:}  \ottnt{A}  \multimap  \ottnt{A}$
        \item $\vdash \,  \mathrm{B}   \ottsym{:}   (  \ottnt{B}  \multimap  \ottnt{C}  )   \multimap   (  \ottnt{A}  \multimap  \ottnt{B}  )   \multimap  \ottnt{A}  \multimap  \ottnt{C}$
        \item $\vdash \,  \mathrm{C}   \ottsym{:}   (  \ottnt{A}  \multimap  \ottnt{B}  \multimap  \ottnt{C}  )   \multimap  \ottnt{B}  \multimap  \ottnt{A}  \multimap  \ottnt{C}$
        \item $\vdash \,  \mathrm{S}^{\sbangs}   \ottsym{:}   (   \sbangs  \ottnt{A}   \multimap  \ottnt{B}  \multimap  \ottnt{C}  )   \multimap   (   \sbangs  \ottnt{A}   \multimap  \ottnt{B}  )   \multimap   \sbangs  \ottnt{A}   \multimap  \ottnt{C}$
        \item $\vdash \,  \mathrm{K}^{\sbangs}   \ottsym{:}  \ottnt{A}  \multimap   \sbangs  \ottnt{B}   \multimap  \ottnt{A}$
        \item $\vdash \,  \mathrm{W}^{\sbangs}   \ottsym{:}   (   \sbangs  \ottnt{A}   \multimap   \sbangs  \ottnt{A}   \multimap  \ottnt{B}  )   \multimap   \sbangs  \ottnt{A}   \multimap  \ottnt{B}$
        \item $\vdash \,  \mathrm{T}^{\sbangs}   \ottsym{:}   \sbangs  \ottnt{A}   \multimap  \ottnt{A}$
        \item $\vdash \,  \mathrm{D}^{\sbangs}   \ottsym{:}   \sbangs   (  \ottnt{A}  \multimap  \ottnt{B}  )    \multimap   \sbangs  \ottnt{A}   \multimap   \sbangs  \ottnt{B} $
        \item $\vdash \,  \mathrm{4}^{\sbangs}   \ottsym{:}   \sbangs  \ottnt{A}   \multimap   \sbangs   \sbangs  \ottnt{A}  $
        \item $\vdash \,  \mathrm{E}   \ottsym{:}   \bangbox \ottnt{A}   \multimap  \ottsym{!}  \ottnt{A}$
      \end{itemize}
      where $\sbangs \in \{ !, \bangbox~ \}$
    \end{minipage}
  }
  &
  {
    \hspace{-0.8em}
    \begin{minipage}{0.40\hsize}
    \paragraph*{Reduction}
      \vspace{-1em}
      \begin{alignat*}{2}
        &  \mathrm{I}  \, \ottnt{M}                  & & \leadsto ~\ottnt{M}\\
        &    \mathrm{B}  \, \ottnt{M}  \, \ottnt{N}  \, \ottnt{L}              & & \leadsto ~ \ottnt{M} \, \ottsym{(}   \ottnt{N} \, \ottnt{L}   \ottsym{)} \\
        &    \mathrm{C}  \, \ottnt{M}  \, \ottnt{N}  \, \ottnt{L}              & & \leadsto ~  \ottnt{M} \, \ottnt{L}  \, \ottnt{N} \\
        &    \mathrm{S}^{\sbangs}  \, \ottnt{M}  \, \ottnt{N}  \, \ottsym{(}   \sbangs  \ottnt{L}   \ottsym{)}       & & \leadsto ~  \ottnt{M} \, \ottsym{(}   \sbangs  \ottnt{L}   \ottsym{)}  \, \ottsym{(}   \ottnt{N} \, \ottsym{(}   \sbangs  \ottnt{L}   \ottsym{)}   \ottsym{)} \\
        &   \mathrm{K}^{\sbangs}  \, \ottnt{M}  \, \ottsym{(}   \sbangs  \ottnt{N}   \ottsym{)}         & & \leadsto ~\ottnt{M}\\
        &   \mathrm{W}^{\sbangs}  \, \ottnt{M}  \, \ottsym{(}   \sbangs  \ottnt{N}   \ottsym{)}         & & \leadsto ~  \ottnt{M} \, \ottsym{(}   \sbangs  \ottnt{N}   \ottsym{)}  \, \ottsym{(}   \sbangs  \ottnt{N}   \ottsym{)} \\
        &  \mathrm{T}^{\sbangs}  \, \ottsym{(}   \sbangs  \ottnt{M}   \ottsym{)}           & & \leadsto ~\ottnt{M}\\
        &   \mathrm{D}^{\sbangs}  \, \ottsym{(}   \sbangs  \ottnt{M}   \ottsym{)}  \, \ottsym{(}   \sbangs  \ottnt{N}   \ottsym{)}  & & \leadsto ~ \sbangs  \ottsym{(}   \ottnt{M} \, \ottnt{N}   \ottsym{)} \\
        &  \mathrm{4}^{\sbangs}  \, \ottsym{(}   \sbangs  \ottnt{M}   \ottsym{)}           & & \leadsto ~ \sbangs   \sbangs  \ottnt{M}  \\
        &  \mathrm{E}  \,  \bangbox  \ottnt{M}                & & \leadsto ~ !  \ottnt{M} 
      \end{alignat*}
      \vspace{-2em}

      \noindent
      where $\sbangs \in \{ !, \bangbox~ \}$
    \end{minipage}
  }
  \end{tabular}

%% file: definitions/cmellbangbox.tex
  \hspace{-1em}
  \begin{tabular}{c|c}
    {
    \hspace{-1em}
    \begin{minipage}{0.7\hsize}
      \paragraph*{Syntactic category}
        \vspace{-1em}
        \begin{align*}
          \text{Formulae}~~A, B, C &::= p ~|~  \ottmv{p} ^\bot 
            ~|~  \ottnt{A}  \stensor  \ottnt{B}  ~|~  \ottnt{A}  \spar  \ottnt{B} 
            ~|~ \ottsym{!}  \ottnt{A} ~|~ ?  \ottnt{A}
            ~|~  \bangbox \ottnt{A}  ~|~  \whynotdia \ottnt{A} 
        \end{align*}
    \end{minipage}
    }
    &
    {
    \hspace{-1em}
    \begin{minipage}{0.20\hsize}
      \paragraph*{Inference rule}
        \vspace{-0.5em}
        \begin{prooftree}
          \def\ScoreOverhang{2pt}
            \AXC{$ $}
          \RightLabel{$ \rax $}
          \UIC{$\vdash \,  \ottnt{A} ^\bot   \ottsym{,}  \ottnt{A}$}
        \end{prooftree}
    \end{minipage}
    }
  \end{tabular}

  \hspace{-0.75em}\rule{0.78\textwidth}{0.4pt}

  \hspace{-1em}
  \begin{tabular}{c}
    {
    \hspace{-1em}
    \begin{minipage}{\hsize}
        \vspace{-0.2em}
        \begin{center}
        \begin{tabular}{c}
          \begin{minipage}{0.25\hsize}
            \begin{prooftree}
              \def\ScoreOverhang{2pt}
              \def\defaultHypSeparation{\hskip 0.5em}
                \AXC{$\vdash \, \Gamma  \ottsym{,}  \ottnt{A}$}
                \AXC{$\vdash \,  \ottnt{A} ^\bot   \ottsym{,}  \Gamma'$}
              \RightLabel{$ \rcut $}
              \BIC{$\vdash \, \Gamma  \ottsym{,}  \Gamma'$}
            \end{prooftree}
          \end{minipage}

          \begin{minipage}{0.22\hsize}
            \begin{prooftree}
              \def\ScoreOverhang{2pt}
              \def\defaultHypSeparation{\hskip 0.5em}
                \AXC{$\vdash \, \Gamma  \ottsym{,}  \ottnt{A}$}
                \AXC{$\vdash \, \Gamma'  \ottsym{,}  \ottnt{B}$}
              \RightLabel{$ \rtensor $}
              \BIC{$\vdash \, \Gamma  \ottsym{,}  \Gamma'  \ottsym{,}   \ottnt{A}  \stensor  \ottnt{B} $}
            \end{prooftree}
          \end{minipage}

          \begin{minipage}{0.17\hsize}
            \begin{prooftree}
              \def\ScoreOverhang{2pt}
                \AXC{$\vdash \, \Gamma  \ottsym{,}  \ottnt{A}  \ottsym{,}  \ottnt{B}$}
              \RightLabel{$ \rpar $}
              \UIC{$\vdash \, \Gamma  \ottsym{,}   \ottnt{A}  \spar  \ottnt{B} $}
            \end{prooftree}
          \end{minipage}

          \begin{minipage}{0.16\hsize}
            \begin{prooftree}
              \def\ScoreOverhang{2pt}
                \AXC{$\vdash \,  \whynotdia \Delta   \ottsym{,}  ?  \Gamma  \ottsym{,}  \ottnt{A}$}
              \RightLabel{$ \rbang $}
              \UIC{$\vdash \,  \whynotdia \Delta   \ottsym{,}  ?  \Gamma  \ottsym{,}  \ottsym{!}  \ottnt{A}$}
            \end{prooftree}
          \end{minipage}

          \begin{minipage}{0.12\hsize}
            \begin{prooftree}
              \def\ScoreOverhang{2pt}
                \AXC{$\vdash \, \Gamma  \ottsym{,}  \ottnt{A}$}
              \RightLabel{$ \rwhynot $}
              \UIC{$\vdash \, \Gamma  \ottsym{,}  ?  \ottnt{A}$}
            \end{prooftree}
          \end{minipage}
        \end{tabular}

        \begin{tabular}{c}
          \begin{minipage}{0.15\hsize}
            \begin{prooftree}
              \def\ScoreOverhang{2pt}
                \AXC{$\vdash \,  \whynotdia \Delta   \ottsym{,}  \ottnt{A}$}
              \RightLabel{$ \rbangbox $}
              \UIC{$\vdash \,  \whynotdia \Delta   \ottsym{,}   \bangbox \ottnt{A} $}
            \end{prooftree}
          \end{minipage}

          \begin{minipage}{0.13\hsize}
            \begin{prooftree}
              \def\ScoreOverhang{2pt}
                \AXC{$\vdash \, \Gamma  \ottsym{,}  \ottnt{A}$}
              \RightLabel{$ \rwhynotdia $}
              \UIC{$\vdash \, \Gamma  \ottsym{,}   \whynotdia \ottnt{A} $}
            \end{prooftree}
          \end{minipage}

          \begin{minipage}{0.15\hsize}
            \begin{prooftree}
              \def\ScoreOverhang{2pt}
                \AXC{$\vdash \, \Gamma$}
              \RightLabel{$ \rqweak $}
              \UIC{$\vdash \, \Gamma  \ottsym{,}  ?  \ottnt{A}$}
            \end{prooftree}
          \end{minipage}

          \begin{minipage}{0.17\hsize}
            \begin{prooftree}
              \def\ScoreOverhang{2pt}
                \AXC{$\vdash \, \Gamma  \ottsym{,}  ?  \ottnt{A}  \ottsym{,}  ?  \ottnt{A}$}
              \RightLabel{$ \rqcont $}
              \UIC{$\vdash \, \Gamma  \ottsym{,}  ?  \ottnt{A}$}
            \end{prooftree}
          \end{minipage}

          \begin{minipage}{0.15\hsize}
            \begin{prooftree}
              \def\ScoreOverhang{2pt}
                \AXC{$\vdash \, \Gamma$}
              \RightLabel{$ \rmweak $}
              \UIC{$\vdash \, \Gamma  \ottsym{,}   \whynotdia \ottnt{A} $}
            \end{prooftree}
          \end{minipage}

          \begin{minipage}{0.17\hsize}
            \begin{prooftree}
              \def\ScoreOverhang{2pt}
                \AXC{$\vdash \, \Gamma  \ottsym{,}   \whynotdia \ottnt{A}   \ottsym{,}   \whynotdia \ottnt{A} $}
              \RightLabel{$ \rmcont $}
              \UIC{$\vdash \, \Gamma  \ottsym{,}   \whynotdia \ottnt{A} $}
            \end{prooftree}
          \end{minipage}
        \end{tabular}
        \end{center}
    \end{minipage}
    }
  \end{tabular}

%% file: definitions/lambdabanglimp.tex
  \hspace{-1em}
  \begin{tabular}{c}
    {
    \hspace{-1.2em}
    \begin{minipage}{0.98\hsize}
      \paragraph*{Syntactic category}
        \vspace{-1em}
        \begin{alignat*}{3}
          &\text{Types}~~&A, B, C &::= p ~|~ \ottnt{A}  \multimap  \ottnt{B} ~|~ \ottsym{!}  \ottnt{A}\\
          &\text{Terms}~~&M, N, L &::= x ~|~ \lambda  \ottmv{x}  \ottsym{:}  \ottnt{A}  \ottsym{.}  \ottnt{M} ~|~  \ottnt{M} \, \ottnt{N} ~|~  !  \ottnt{M}  ~|~ \ottkw{let} \,  !  \ottmv{x}   \ottsym{=}  \ottnt{M} \, \ottkw{in} \, \ottnt{N}
        \end{alignat*} 
    \end{minipage}
    }
  \end{tabular}

  \vspace{0.5em}
  
  \hspace{-1em}
  \begin{tabular}{c}
    {
    \hspace{-1.2em}
    \begin{minipage}{0.98\hsize}
      \paragraph*{Typing rule}
        \begin{center}
        \vspace{-1em}
        \begin{tabular}{c}
          \begin{minipage}{0.45\hsize}
            \begin{prooftree}
                \AXC{$ $}
              \RightLabel{$ \rlax $}
              \UIC{$\Gamma  \ottsym{;}  \ottmv{x}  \ottsym{:}  \ottnt{A} \, \vdash \, \ottmv{x}  \ottsym{:}  \ottnt{A}$}
            \end{prooftree}
          \end{minipage}
      
          \begin{minipage}{0.45\hsize}
            \begin{prooftree}
                \AXC{$ $}
              \RightLabel{$ \rbax $}
              \UIC{$\Gamma  \ottsym{,}  \ottmv{x}  \ottsym{:}  \ottnt{A}  \ottsym{;}  \emptyset \, \vdash \, \ottmv{x}  \ottsym{:}  \ottnt{A}$}
            \end{prooftree}
          \end{minipage}
        \end{tabular}
      
        \begin{tabular}{c}
          \begin{minipage}{0.38\hsize}
            \begin{prooftree}
                \AXC{$\Gamma  \ottsym{;}  \Sigma  \ottsym{,}  \ottmv{x}  \ottsym{:}  \ottnt{A} \, \vdash \, \ottnt{M}  \ottsym{:}  \ottnt{B}$}
              \RightLabel{$ \rlimpi $}
              \UIC{$\Gamma  \ottsym{;}  \Sigma \, \vdash \, \lambda  \ottmv{x}  \ottsym{:}  \ottnt{A}  \ottsym{.}  \ottnt{M}  \ottsym{:}  \ottnt{A}  \multimap  \ottnt{B}$}
            \end{prooftree}
          \end{minipage}
      
          \begin{minipage}{0.53\hsize}
            \begin{prooftree}
                \AXC{$\Gamma  \ottsym{;}  \Sigma \, \vdash \, \ottnt{M}  \ottsym{:}  \ottnt{A}  \multimap  \ottnt{B}$}
                \AXC{$\Gamma  \ottsym{;}  \Sigma' \, \vdash \, \ottnt{N}  \ottsym{:}  \ottnt{A}$}
              \RightLabel{$ \rlimpe $}
              \BIC{$\Gamma  \ottsym{;}  \Sigma  \ottsym{,}  \Sigma' \, \vdash \,  \ottnt{M} \, \ottnt{N}   \ottsym{:}  \ottnt{B}$}
            \end{prooftree}
          \end{minipage}
        \end{tabular}
      
        \begin{tabular}{c}
          \begin{minipage}{0.37\hsize}
            \begin{prooftree}
                \AXC{$\Gamma  \ottsym{;}  \emptyset \, \vdash \, \ottnt{M}  \ottsym{:}  \ottnt{A}$}
              \RightLabel{$ \rbangi $}
              \UIC{$\Gamma  \ottsym{;}  \emptyset \, \vdash \,  !  \ottnt{M}   \ottsym{:}  \ottsym{!}  \ottnt{A}$}
            \end{prooftree}
          \end{minipage}
      
          \begin{minipage}{0.53\hsize}
            \begin{prooftree}
                \AXC{$\Gamma  \ottsym{;}  \Sigma \, \vdash \, \ottnt{M}  \ottsym{:}  \ottsym{!}  \ottnt{A}$}
                \AXC{$\Gamma  \ottsym{,}  \ottmv{x}  \ottsym{:}  \ottnt{A}  \ottsym{;}  \Sigma' \, \vdash \, \ottnt{N}  \ottsym{:}  \ottnt{B}$}
              \RightLabel{$ \rbange $}
              \BIC{$\Gamma  \ottsym{;}  \Sigma  \ottsym{,}  \Sigma' \, \vdash \, \ottkw{let} \,  !  \ottmv{x}   \ottsym{=}  \ottnt{M} \, \ottkw{in} \, \ottnt{N}  \ottsym{:}  \ottnt{B}$}
            \end{prooftree}
          \end{minipage}
        \end{tabular}
        \end{center}
    \end{minipage}
    }
  \end{tabular}

  \vspace{0.5em}
  
  \hspace{-1em}
  \begin{tabular}{c}
    {
    \hspace{-1.2em}
    \begin{minipage}{0.98\hsize}
      \paragraph*{Reduction rule}
        \vspace{-1em}
        \begin{alignat*}{3}
           \ottsym{(}  \lambda  \ottmv{x}  \ottsym{:}  \ottnt{A}  \ottsym{.}  \ottnt{M}  \ottsym{)} \, \ottnt{N}                 & \leadsto   \ottnt{M}  [  \ottmv{x}  :=  \ottnt{N}  ]  &&\\
           \lambda  \ottmv{x}  \ottsym{:}  \ottnt{A}  \ottsym{.}  \ottnt{M} \, \ottmv{x}                  & \leadsto  \ottnt{M}       &&\\
          \ottkw{let} \,  !  \ottmv{x}   \ottsym{=}   !  \ottnt{M}  \, \ottkw{in} \, \ottnt{N} & \leadsto   \ottnt{N}  [  \ottmv{x}  :=  \ottnt{M}  ]  &&\\
          \ottkw{let} \,  !  \ottmv{x}   \ottsym{=}  \ottnt{M} \, \ottkw{in} \,  C [   !  \ottmv{x}   ]   & \leadsto   C [  \ottnt{M}  ]     &&\\
           \ottsym{(}  \ottkw{let} \,  !  \ottmv{x}   \ottsym{=}  \ottnt{M} \, \ottkw{in} \, \ottnt{N}  \ottsym{)} \, \ottnt{L}   & \leadsto   \ottkw{let} \,  !  \ottmv{x}   \ottsym{=}  \ottnt{M} \, \ottkw{in} \, \ottnt{N} \, \ottnt{L}  &&\\
          \ottkw{let} \,  !  \ottmv{y}   \ottsym{=}  \ottsym{(}  \ottkw{let} \,  !  \ottmv{x}   \ottsym{=}  \ottnt{M} \, \ottkw{in} \, \ottnt{N}  \ottsym{)} \, \ottkw{in} \, \ottnt{L} & \leadsto  \ottkw{let} \,  !  \ottmv{x}   \ottsym{=}  \ottnt{M} \, \ottkw{in} \, \ottkw{let} \,  !  \ottmv{y}   \ottsym{=}  \ottnt{N} \, \ottkw{in} \, \ottnt{L} &&\\
          \lambda  \ottmv{y}  \ottsym{:}  \ottnt{A}  \ottsym{.}  \ottsym{(}  \ottkw{let} \,  !  \ottmv{x}   \ottsym{=}  \ottnt{M} \, \ottkw{in} \, \ottnt{N}  \ottsym{)} & \leadsto  \ottkw{let} \,  !  \ottmv{x}   \ottsym{=}  \ottnt{M} \, \ottkw{in} \, \lambda  \ottmv{y}  \ottsym{:}  \ottnt{A}  \ottsym{.}  \ottnt{N}
            &\hspace{0.3em}\text{(if $y \not\in \fv{M}$)}&\\
           \ottnt{L} \, \ottsym{(}  \ottkw{let} \,  !  \ottmv{x}   \ottsym{=}  \ottnt{M} \, \ottkw{in} \, \ottnt{N}  \ottsym{)}  & \leadsto   \ottkw{let} \,  !  \ottmv{x}   \ottsym{=}  \ottnt{M} \, \ottkw{in} \, \ottnt{L} \, \ottnt{N}  &&\\
          \lambda  \ottmv{z}  \ottsym{:}  \ottnt{A}  \ottsym{.}  \ottsym{(}  \ottkw{let} \,  !  \ottmv{y}   \ottsym{=}  \ottnt{M} \, \ottkw{in} \, \ottkw{let} \,  !  \ottmv{x}   \ottsym{=}  \ottnt{L} \, \ottkw{in} \, \ottnt{N}  \ottsym{)} & \leadsto  \ottkw{let} \,  !  \ottmv{x}   \ottsym{=}  \ottnt{L} \, \ottkw{in} \, \lambda  \ottmv{z}  \ottsym{:}  \ottnt{A}  \ottsym{.}  \ottsym{(}  \ottkw{let} \,  !  \ottmv{y}   \ottsym{=}  \ottnt{M} \, \ottkw{in} \, \ottnt{N}  \ottsym{)}
            &\hspace{0.3em}\text{(if $y \not\in \fv{L}$)}&
        \end{alignat*}
        where $C[-]$ is a linear context defined by the following grammar:
        \begin{align*}
          C ::= [-] ~|~ \lambda  \ottmv{x}  \ottsym{:}  \ottnt{A}  \ottsym{.}   C  ~|~   C  \, \ottnt{M}  ~|~  \ottnt{M} \,  C   ~|~ \ottkw{let} \,  !  \ottmv{x}   \ottsym{=}   C  \, \ottkw{in} \, \ottnt{M} ~|~ \ottkw{let} \,  !  \ottmv{x}   \ottsym{=}  \ottnt{M} \, \ottkw{in} \,  C \\
        \end{align*}
    \end{minipage}
    }
  \end{tabular}

%% file: definitions/embedding.tex
  \begin{center}
  \begin{tabular}{c|c}
    {
      \hspace{-1em}
      \begin{minipage}{0.40\hsize}
        \vspace{-0.5em}
        \underline{$ \fembed{ \ottnt{A} } $}
        \vspace{-1.5em}
        \begin{align*}
           \fembed{ \ottmv{p} }       &\defeq \ottmv{p}\\
           \fembed{ \ottnt{A}  \multimap  \ottnt{B} }  &\defeq  \fembed{ \ottnt{A} }   \multimap   \fembed{ \ottnt{B} } \\
           \fembed{ \ottsym{!}  \ottnt{A} }      &\defeq \ottsym{!}   \fembed{ \ottnt{A} } \\
           \fembed{  \bangbox \ottnt{A}  }    &\defeq \ottsym{!}   \fembed{ \ottnt{A} } 
        \end{align*}

        \vspace{-1.5em}
        \rule{\textwidth}{0.4pt}

        \underline{$ \fembed{ \Gamma } $}

        \vspace{-1.5em}
        \begin{align*}
           \fembed{ \Gamma }  &\defeq \{ (x :  \fembed{ \ottnt{A} } ) ~|~ (x : A) \in \Gamma \}
        \end{align*}
      \end{minipage}
    }
    &
    {
      \begin{minipage}{0.52\hsize}
        \vspace{0.2em}
        \underline{$ \fembed{ \ottnt{M} } $}
        \vspace{-2.5em}
        \begin{align*}
           \fembed{ \ottmv{x} }                       &\defeq \ottmv{x}\\
           \fembed{ \lambda  \ottmv{x}  \ottsym{:}  \ottnt{A}  \ottsym{.}  \ottnt{M} }                  &\defeq \lambda  \ottmv{x}  \ottsym{:}   \fembed{ \ottnt{A} }   \ottsym{.}   \fembed{ \ottnt{M} } \\
           \fembed{  \ottnt{M} \, \ottnt{N}  }                     &\defeq   \fembed{ \ottnt{M} }  \,  \fembed{ \ottnt{N} }  \\
           \fembed{  !  \ottnt{M}  }                  &\defeq  !   \fembed{ \ottnt{M} }  \\
           \fembed{ \ottkw{let} \,  !  \ottmv{x}   \ottsym{=}  \ottnt{M} \, \ottkw{in} \, \ottnt{N} }     &\defeq \ottkw{let} \,  !  \ottmv{x}   \ottsym{=}   \fembed{ \ottnt{M} }  \, \ottkw{in} \,  \fembed{ \ottnt{N} } \\
           \fembed{  \bangbox  \ottnt{M}  }               &\defeq  !   \fembed{ \ottnt{M} }  \\
           \fembed{ \ottkw{let} \,  \bangbox  \ottmv{x}   \ottsym{=}  \ottnt{M} \, \ottkw{in} \, \ottnt{N} }  &\defeq \ottkw{let} \,  !  \ottmv{x}   \ottsym{=}   \fembed{ \ottnt{M} }  \, \ottkw{in} \,  \fembed{ \ottnt{N} } 
        \end{align*}
      \end{minipage}
    }
  \end{tabular}
  \end{center}

%% file: definitions/bracket_abstraction.tex
  \hspace{-1em}
  \begin{tabular}{c}
    {
    \begin{minipage}{0.97\hsize}
      \vspace{0.2em}
      \underline{$ \lambda_{*}  \ottmv{x} .  \ottnt{M} $}
        \vspace{-1em}
        \begin{alignat*}{3}
           \lambda_{*}  \ottmv{x} .  \ottmv{x}       &\defeq& &~ \mathrm{I}                   &&\\
           \lambda_{*}  \ottmv{x} .  \ottsym{(}   \ottnt{M} \, \ottnt{N}   \ottsym{)}   &\defeq& &~   \mathrm{C}  \, \ottsym{(}   \lambda_{*}  \ottmv{x} .  \ottnt{M}   \ottsym{)}  \, \ottnt{N}  &&~~~\text{if $x \in \mathrm{FV}(M)$ }\\
           \lambda_{*}  \ottmv{x} .  \ottsym{(}   \ottnt{M} \, \ottnt{N}   \ottsym{)}   &\defeq& &~   \mathrm{B}  \, \ottnt{M}  \, \ottsym{(}   \lambda_{*}  \ottmv{x} .  \ottnt{N}   \ottsym{)}  &&~~~\text{if $x \in \mathrm{FV}(N)$ }
        \end{alignat*}
    \end{minipage}
    }
  \end{tabular}

  \rule{\textwidth}{0.4pt}

  \hspace{-1em}
  \begin{tabular}{c|c}
    {
    \begin{minipage}{0.45\hsize}
      \vspace{-1.6em}
      \underline{$ \lambda^{!}_{*}  \ottmv{x} .  \ottnt{M} $}
        \vspace{-1em}
        \begin{alignat*}{3}
          \hspace{-1em} \lambda^{!}_{*}  \ottmv{x} .  \ottmv{x}       &\defeq& &~ \mathrm{T}^!                                  &&\\
          \hspace{-1em} \lambda^{!}_{*}  \ottmv{x} .  \ottnt{M}       &\defeq& &~  \mathrm{K}^!  \, \ottnt{M}                                &&~~~\text{if $(a)$}\\
          \hspace{-1em} \lambda^{!}_{*}  \ottmv{x} .  \ottsym{(}   \ottnt{M} \, \ottnt{N}   \ottsym{)}   &\defeq& &~   \mathrm{C}  \, \ottsym{(}   \lambda^{!}_{*}  \ottmv{x} .  \ottnt{M}   \ottsym{)}  \, \ottnt{N}                &&~~~\text{if $(b)$}\\
          \hspace{-1em} \lambda^{!}_{*}  \ottmv{x} .  \ottsym{(}   \ottnt{M} \, \ottnt{N}   \ottsym{)}   &\defeq& &~   \mathrm{B}  \, \ottnt{M}  \, \ottsym{(}   \lambda^{!}_{*}  \ottmv{x} .  \ottnt{N}   \ottsym{)}                &&~~~\text{if $(c)$}\\
          \hspace{-1em} \lambda^{!}_{*}  \ottmv{x} .  \ottsym{(}   \ottnt{M} \, \ottnt{N}   \ottsym{)}   &\defeq& &~   \mathrm{S}^!  \, \ottsym{(}   \lambda^{!}_{*}  \ottmv{x} .  \ottnt{M}   \ottsym{)}  \, \ottsym{(}   \lambda^{!}_{*}  \ottmv{x} .  \ottnt{N}   \ottsym{)}  &&~~~\text{if $(d)$}\\
          \hspace{-1em} \lambda^{!}_{*}  \ottmv{x} .  \ottsym{(}   !  \ottnt{M}   \ottsym{)}   &\defeq& &~   \mathrm{B}  \, \ottsym{(}    \mathrm{D}^!  \,  !  \ottsym{(}   \lambda^{!}_{*}  \ottmv{x} .  \ottnt{M}   \ottsym{)}    \ottsym{)}  \,  \mathrm{4}^!          &&
        \end{alignat*}
    \end{minipage}
    }
    &
    {
    \begin{minipage}{0.45\hsize}
      \vspace{0.1em}
      \underline{$ \lambda^{\smallbangbox}_{*}  \ottmv{x} .  \ottnt{M} $}
        \vspace{-1em}
        \begin{alignat*}{3}
          \hspace{-1em} \lambda^{\smallbangbox}_{*}  \ottmv{x} .  \ottmv{x}         &\defeq& &~ \mathrm{T}^{\smallbangbox}                                       &&\\
          \hspace{-1em} \lambda^{\smallbangbox}_{*}  \ottmv{x} .  \ottnt{M}         &\defeq& &~  \mathrm{K}^{\smallbangbox}  \, \ottnt{M}                                     &&\text{if $(a)$}\\
          \hspace{-1em} \lambda^{\smallbangbox}_{*}  \ottmv{x} .  \ottsym{(}   \ottnt{M} \, \ottnt{N}   \ottsym{)}     &\defeq& &~   \mathrm{C}  \, \ottsym{(}   \lambda^{\smallbangbox}_{*}  \ottmv{x} .  \ottnt{M}   \ottsym{)}  \, \ottnt{N}                     &&\text{if $(b)$}\\
          \hspace{-1em} \lambda^{\smallbangbox}_{*}  \ottmv{x} .  \ottsym{(}   \ottnt{M} \, \ottnt{N}   \ottsym{)}     &\defeq& &~   \mathrm{B}  \, \ottnt{M}  \, \ottsym{(}   \lambda^{\smallbangbox}_{*}  \ottmv{x} .  \ottnt{N}   \ottsym{)}                     &&\text{if $(c)$}\\
          \hspace{-1em} \lambda^{\smallbangbox}_{*}  \ottmv{x} .  \ottsym{(}   \ottnt{M} \, \ottnt{N}   \ottsym{)}     &\defeq& &~   \mathrm{S}^{\smallbangbox}  \, \ottsym{(}   \lambda^{\smallbangbox}_{*}  \ottmv{x} .  \ottnt{M}   \ottsym{)}  \, \ottsym{(}   \lambda^{\smallbangbox}_{*}  \ottmv{x} .  \ottnt{N}   \ottsym{)}     &&\text{if $(d)$}\\
          \hspace{-1em} \lambda^{\smallbangbox}_{*}  \ottmv{x} .  \ottsym{(}   !  \ottnt{M}   \ottsym{)}     &\defeq& &~   \mathrm{B}  \, \ottsym{(}    \mathrm{D}^!  \, \ottsym{(}   !  \ottsym{(}   \lambda^{\smallbangbox}_{*}  \ottmv{x} .  \ottnt{M}   \ottsym{)}   \ottsym{)}   \ottsym{)}  \, \ottsym{(}     \mathrm{B}  \,  \mathrm{E}   \,  \mathrm{4}^{\smallbangbox}\,    \ottsym{)}   &&\\
          \hspace{-1em} \lambda^{\smallbangbox}_{*}  \ottmv{x} .  \ottsym{(}   \bangbox  \ottnt{M}   \ottsym{)}   &\defeq& &~   \mathrm{B}  \, \ottsym{(}    \mathrm{D}^{\smallbangbox}  \, \ottsym{(}   \bangbox  \ottsym{(}   \lambda^{\smallbangbox}_{*}  \ottmv{x} .  \ottnt{M}   \ottsym{)}   \ottsym{)}   \ottsym{)}  \,  \mathrm{4}^{\smallbangbox}\,         &&
        \end{alignat*}
    \end{minipage}
    }
  \end{tabular}

  \rule{\textwidth}{0.4pt}

  where $(a), (b), (c), (d)$ means the conditions $(x \not\in \fv{M})$, $(x \in \fv{M} \text{ and } x \not\in \fv{N})$,
  $(x \not\in \fv{M} \text{ and } x \in \fv{N})$, $(x \in \fv{M} \text{ and } x \in \fv{N}) $, respectively.
  \vspace{0.5em}